\documentclass[10pt,letterpaper]{article}
\usepackage[margin=1in]{geometry}
\usepackage{authblk}
\usepackage{amsmath,amssymb,amsthm}
\usepackage[toc,page]{appendix}
\usepackage{braket}
\usepackage{graphicx}
\usepackage{xcolor}
\usepackage{dsfont}
\usepackage{enumitem}
\usepackage[mathscr]{euscript}
\usepackage{comment}
\usepackage[normalem]{ulem}
\usepackage{diagbox}
\usepackage[colorlinks=true,allcolors=blue]{hyperref}
\usepackage[capitalise]{cleveref}

\DeclareMathOperator{\sech}{sech}

\DeclareMathOperator{\sgn}{sign}

\newcounter{propcount}
\numberwithin{propcount}{section}

\newtheorem{theorem}[propcount]{Theorem}
\newtheorem{lemma}[propcount]{Lemma}

\newtheorem{proposition}[propcount]{Proposition}
\newtheorem{corollary}[propcount]{Corollary}
\newtheorem{conjecture}[propcount]{Conjecture}
\theoremstyle{definition}
\newtheorem{definition}[propcount]{Definition}
\newtheorem{remark}[propcount]{Remark}

\newcommand{\Id}{\mathds{1}}

\newcommand{\vect}[1]{\mathbf{#1}}

\newcommand{\dv}{{\vect{d}}}

\newcommand{\vv}{{\vect{v}}}
\newcommand{\wv}{{\vect{w}}}
\newcommand{\xv}{{\vect{x}}}
\newcommand{\yv}{{\vect{y}}}
\newcommand{\zv}{{\vect{z}}}

\newcommand{\RR}{\mathbb{R}}
\newcommand{\FF}{\mathbb{F}}
\newcommand{\EE}{\mathbb{E}}
\newcommand{\ZZ}{\mathbb{Z}}
\newcommand{\PP}{\mathbb{P}}
\newcommand{\QQ}{\mathbb{Q}}
\newcommand{\GG}{\mathbb{G}}

\def\bX{\boldsymbol{X}}
\def\bY{\boldsymbol{Y}}

\def\bSig{\boldsymbol{\Sigma}}
\newcommand{\cA}{\mathcal{A}}

\newcommand{\cC}{\mathcal{C}}
\newcommand{\cF}{\mathcal{F}}
\newcommand{\cN}{\mathcal{N}}

\newcommand{\eps}{\varepsilon}
\newcommand{\tep}{\tilde{\eps}}
\newcommand{\iid}{\textnormal{i.i.d.}}

\DeclareMathOperator{\Law}{Law}
\DeclareMathOperator{\poly}{poly}
\DeclareMathOperator{\diam}{diam}
\def\Unif{{\sf Unif}}
\DeclareMathOperator*{\E}{\EE}
\def\round{{\mathsf{round}}}

\newcommand{\dyn}{\textnormal{dyn}}
\newcommand{\ER}{\textnormal{ER}}
\newcommand{\DQI}{\textnormal{DQI}}
\newcommand{\USD}{\textnormal{USD}}
\newcommand{\ALG}{\textnormal{ALG}}
\newcommand{\OGP}{\textnormal{OGP}}
\newcommand{\RS}{\textnormal{RSB}}
\newcommand{\TV}{\textnormal{TV}}

\def\lt{\left}
\def\rt{\right}
\def\la{\langle}
\def\ra{\rangle}

\DeclareMathOperator{\Poisson}{Poisson}
\DeclareMathOperator{\rank}{rank}
\DeclareMathOperator{\range}{range}
\DeclareMathOperator*{\argmax}{\arg\max}

\numberwithin{equation}{section}
\numberwithin{table}{section}

\begin{document}

\title{\huge Gibbs Sampling in the Shattered Phase by \\ Decoded Quantum Interferometry}

\author[1]{Leo Zhou}
\author[2]{Noah Shutty}
\author[3]{Mark Sellke}
\author[2]{Stephen P.~Jordan}

\affil[1]{\small{\it{Electrical \& Computer Engineering Department, University of California, Los Angeles, CA, USA}}}
\affil[2]{\small{\it{Google Quantum AI, Venice, CA, USA}}}
\affil[3]{\small{\it{Department of Statistics, Harvard University, Cambridge, MA, USA}}}

\date{September 30, 2026}

\maketitle

\begin{abstract}
We apply Decoded Quantum Interferometry (DQI) to sample from the Gibbs measures of classical Ising spin Hamiltonians.
We show that this Gibbs sampling problem reduces to a quantum decoding problem, and the temperature achievable by DQI is determined by the performance of decoding algorithms.
We then focus on the task of Gibbs sampling for classical Ising $k$-spin glasses (or Max-$k$-XORSAT) on random Erd\H{o}s-R\'enyi hypergraphs with average degree $D\ge k$. 
In a temperature range beginning asymptotically at the predicted dynamical phase transition, $\beta_{\dyn}(k,D) = \sqrt{(2\ln k)/D}\times [1+o_{k\to\infty}(1)]$, we show that shattering and disorder chaos form a topological barrier that obstructs many algorithms, including Glauber dynamics and any algorithm whose output distribution is ``stable'' under perturbations of the input.
In contrast, we prove that this barrier can be broken both by a classical algorithm based on Prange's method, and by DQI equipped with a quantum decoder. 
For example, when $D=\alpha k$ with fixed $\alpha>1$, both Prange's algorithm and DQI can sample at any inverse temperature $\beta < \tanh^{-1}(1/\alpha)$ for sufficiently large $k$, well beyond the dynamical threshold $\beta_\dyn \sim \sqrt{2\ln k / (\alpha k)}$.
Therefore, our results show that DQI can overcome topological barriers that obstruct stable algorithms.
\end{abstract}

\tableofcontents

\section{Introduction}
\label{sec:intro}

Spin glasses are paradigmatic models of random optimization problems and disordered systems, and naturally serve as a testbed for optimization and sampling algorithms.
Such problems are believed to be computationally hard, even in the average-case setting, due to the topological barriers that separate the many local minima in their solution space.
A productive recent line of work has formalized these barriers through the overlap gap property (OGP)~\cite{gamarnik2014limits, gamarnik2021overlap}, leading to many provable hardness results against broad classes of algorithms, such as local algorithms~\cite{rahman2017local, chen2019suboptimality}, low-degree polynomials~\cite{gamarnik2020low,Wein2022, bresler2022kSAT}, low-depth quantum algorithms~\cite{farhi2020typical,basso2022performance,anshu2023concentration, chen2023local}, and any algorithm that is ``stable'' under perturbations of the input~\cite{gamarnik2024low,  GKPX2022,huang2022tight}. These hardness results are particularly compelling because, in many models, the OGP-predicted thresholds for stable algorithms match the best-known polynomial-time classical algorithms (see \cite{gamarnik2021overlap} for a review). They also raise a natural question: can an efficient quantum algorithm overcome these barriers that obstruct stable algorithms?

Decoded Quantum Interferometry (DQI) \cite{jordan2024optimization} is a promising quantum algorithmic framework for investigating this question. The DQI framework reduces combinatorial optimization to decoding an error-correcting code, so that the quality of the solution achieved is determined by the performance of decoders.
It was recently shown that DQI can efficiently find approximate optima for certain optimization problems related to polynomial regression for which no polynomial-time classical algorithm is known~\cite{jordan2024optimization}. 
It remains an open question whether DQI can also achieve quantum advantage for algebraically unstructured problems such as random spin glass models. 
However, under assumptions on the decoder, \cite{anschuetz2025decoded} shows that DQI is stable and consequently faces OGP-based limitations when optimizing certain sparse spin glasses.
This makes it important to understand which limitations persist when more general decoders are allowed.

In this work, we consider a slight generalization of the DQI framework applied to sample from the Gibbs measures of classical spin glass Hamiltonians. 
In fact, the algorithmic techniques we use are closer in spirit to the original reductions of Regev, Aharonov, and Ta-Shma \cite{ATS03,R04,AR05,R09} than they are to DQI as introduced in \cite{jordan2024optimization};
nevertheless, DQI and Regev's reduction are efficiently interconvertible~\cite{SMR26}. Here we stick to the terminology of DQI, which also generalizes more naturally than does Regev's reduction to sampling from Gibbs states of quantum Hamiltonians \cite{schmidhuber2025hamiltonian}.
Via the quantum Fourier transform, DQI reduces the Gibbs sampling problem to the problem of decoding a codeword from a classical error-correcting code after it has been subjected to a coherent quantum superposition of bit flip errors. Sampling from Gibbs measure at a larger inverse temperature $\beta$ requires decoding at a higher bit flip error rate $\gamma=\frac12-\frac12\sech\beta$, as we show in \cref{thm:gamma_vs_beta}. The capabilities of DQI thus depend on the specific choice of decoder.

Gibbs sampling is an even more demanding problem than optimization, since it requires outputting not just one but all low-energy configurations with correct relative probabilities. The Sherrington-Kirkpatrick model is a clear illustrative example, where an efficient optimization algorithm~\cite{montanari2019optimization, lopatto2026replica} exists to find near-ground states ($\beta\to\infty$), while efficient sampling is only known in the high-temperature regime $\beta<1$ and provably impossible for stable algorithms at any $\beta >1$ ~\cite{AMS22sampling,Celentano2024}.
For many spin glass models, physicists have long predicted that thermal processes (e.g. Glauber dynamics) equilibrate rapidly only for $\beta < \beta_\dyn$, a dynamical phase transition point defined by the emergence of nontrivial solution to the so-called 1RSB equations (see e.g.~\cite{Sompolinsky1981,mezard1987spin,Montanari2003} and \cref{app:spin-glass}).
More recently, hardness of sampling results against stable algorithms have been proved for dense Ising $k$-spin glasses~\cite{alaoui2023sampling, alaoui2024near},  which are conjectured to hold for all  $\beta > \beta_\dyn(k) \simeq \sqrt{2\ln k / k}$.
Nevertheless, the algorithmic thresholds for sampling and optimization need not coincide; for comparison, we can associate the OGP-predicted energy threshold for stable optimization algorithms~\cite{AMS21optimization,huang2022tight} on these dense models with an effective temperature $\beta_\OGP(k)$. (For the sparse graph, a natural definition is the best factor of IID or ``local'' algorithm as in \cite{lyons2017factors}, or alternatively one can study low degree polynomials as in \cite{wein2025computational}.) 
Numerical evaluation of the predicted threshold shows that $\beta_\OGP(k) > \beta_\dyn(k)$ for various $3\le k \le 20$ (see \cref{tab:dense-betas}), and one can prove that $\beta_\OGP(k)=\Theta\big(\sqrt{\ln k / k}\big)$ \cite{anschuetz2025decoded,bogpinprogress}.

\begin{figure}[t]
\centering
\includegraphics[height=2.6in]{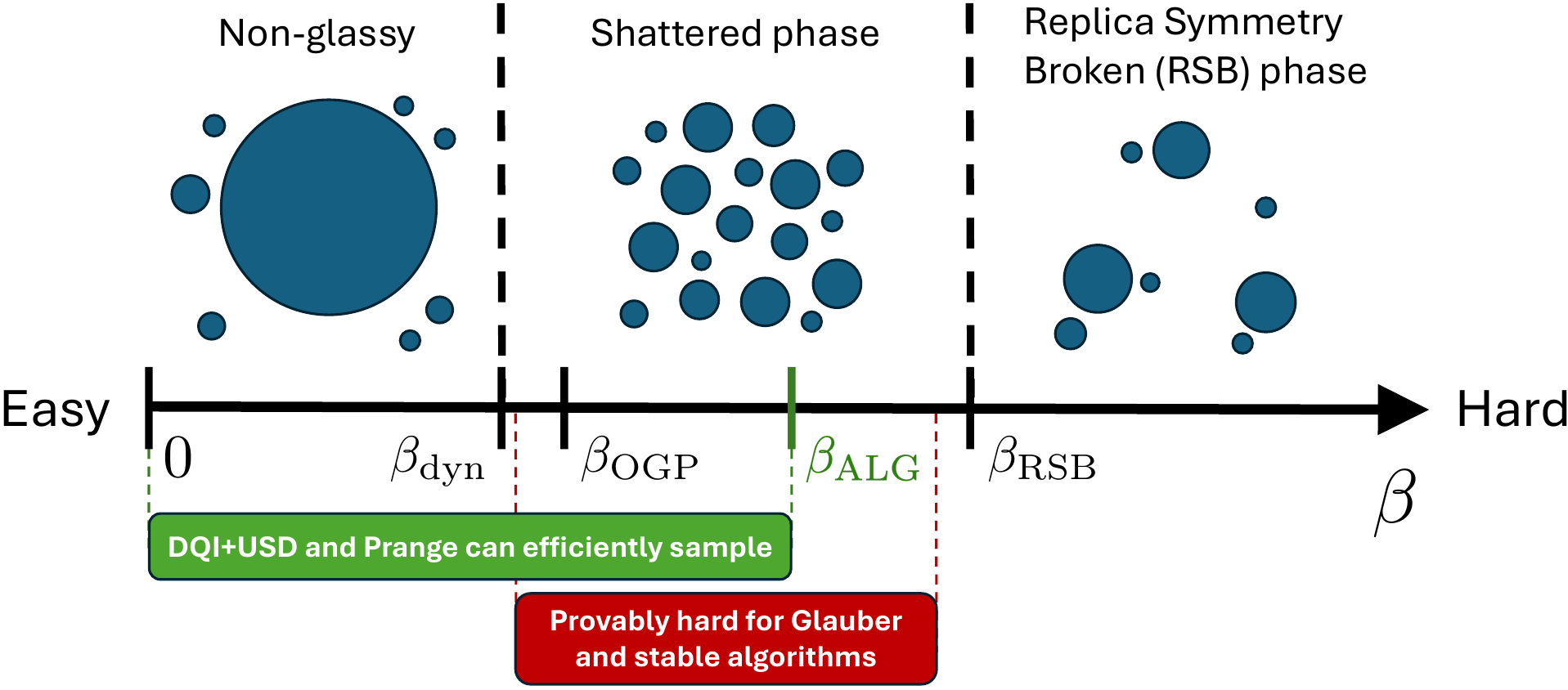}
\caption{
The different phases in an Ising $k$-spin glass for large $k$, visualized by a cluster decomposition of the Gibbs measure at different inverse temperature $\beta$. Prior work identified $E(\beta_{\OGP})$ and $\beta_{\dyn}$ as barriers for broad classes of algorithms, motivating the possibility that they reflect more general computational barriers. Here we prove hardness of Gibbs sampling for Glauber dynamics and all stable algorithms throughout the indicated red interval in the shattered phase. In contrast, DQI+USD and Prange's algorithm efficiently Gibbs sample for all $\beta < \beta_{\ALG}$, including at temperatures deep in the shattered phase.
}
\end{figure}

Motivated by these hardness results, we apply DQI to sample from the Gibbs measures of Ising $k$-spin glasses on random sparse Erd\H{o}s-R\'enyi hypergraphs with average degree $D \ge k$.
Equivalently, these are random Max-$k$-XORSAT instances.
For this ensemble of spin glasses, the dynamic transition temperature is believed to be $\beta_\dyn(k,D)=\sqrt{2\ln k / D}\cdot[1+o_{k}(1)]$. We show in \cref{thm:shatter-hardness} that for any fixed $\eps>0$, all $\beta \in [(1+\eps)\sqrt{2\ln k / D}, \sqrt{k/(2D)} \,]$ and sufficiently large $k$, the Gibbs measure shatters into exponentially many well-separated clusters. We then use shattering to prove that no stable algorithms (classical or quantum) can Gibbs sample in this temperature range, and that Glauber dynamics has exponentially slow mixing, similar to recent hardness results~\cite{alaoui2023sampling,alaoui2024near}.

On the algorithmic side, we show that DQI with a quantum decoder is able to sample in this shattered temperature range, which immediately implies that DQI is not subject to the OGP-like topological obstructions for stable algorithms. Specifically, we solve the quantum decoding problem corresponding to the Erd\H{o}s-R\'enyi ensemble, using a quantum decoding technique based on Unambiguous State Discrimination (USD) measurements introduced in \cite{chailloux2023quantum}. Therefore, we refer to this version of DQI as DQI+USD. 
In \cref{thm:gamma_USD} we quantify the maximum coherent bit flip error rate that can be reliably decoded using the USD technique, which can then be substituted into \cref{thm:gamma_vs_beta} to yield $\beta_\ALG(k,D) = \tanh^{-1}(k/D) \cdot [1+o_{k}(1)]$ (uniformly over $D \ge \alpha k$ for any fixed $\alpha>1$),  such that DQI can efficiently Gibbs sample for all $\beta < \beta_\ALG$.
For example, when $D=3k$ and $k$ is large, $\beta_\ALG$ lies squarely in the shattering phase: $\beta_\ALG(k,3k)\to\frac12\ln2<1/\sqrt6$, while $\beta_\dyn(k,3k) \asymp \sqrt{\ln k/k}$ tends to zero.

Our results show that DQI+USD can sample from certain classical Gibbs measures in the shattered phase, at temperatures colder than can be reached by any stable classical algorithm. This raises the question of whether DQI+USD achieves quantum advantage for these problems over \emph{all} classical algorithms. Prange's algorithm~\cite{prange1962use}, which originates in the cryptanalysis literature, was previously shown in \cite{chailloux2023quantum} to match DQI+USD on certain problems. We show in \cref{prop:prange-works} that Prange's algorithm can be adapted to Gibbs sampling and reaches the same temperatures that can be reached by DQI+USD.
Therefore, DQI+USD does not achieve advantage over all classical algorithms for this problem. This also shows shattering and OGP do not necessarily predict intractability for quantum or classical algorithms in these spin glasses. 
In a similar spirit, recent work \cite{li2024some} shows that the shortest-path problem exhibits OGP while being algorithmically easy, although a follow-up \cite{koehler2026overlap} illustrates that OGP disappears when the optimization landscape is modified.
It would thus be interesting to develop a more universal theory of average-case hardness in these spin glass models, and investigate whether there exists a parameter regime where a genuine quantum advantage is possible.

\section{Main Results}
\label{sec:main_results}

Our first main result is the reduction of the Gibbs sampling problem to a decoding problem, where a mapping between the inverse temperature $\beta$ of the Gibbs ensemble and the decodable error rate $\gamma$ is derived. We start by formally defining the two problems:

\begin{definition}[$(H, \beta,\eps)$-Gibbs Sampling Problem]
    Let $H:\{+1,-1\}^N \to \mathbb{R}$ be a classical Hamiltonian that maps a spin configuration to a corresponding energy. The corresponding Gibbs measure at inverse temperature $\beta$ is then
    \begin{equation}
        \label{eq:gibbs}
        \mu_{\beta,H}(\mathbf{z}) = \frac{1}{Z_{\beta,H}} e^{-\beta H(\mathbf{z})}, \quad \textrm{where} \quad Z_{\beta,H} = \sum_{\mathbf{z}} e^{-\beta H(\mathbf{z})}.
    \end{equation}
    The \emph{$(H, \beta,\eps)$-Gibbs sampling problem} is to sample from a distribution over $\{+1,-1\}^N$ whose total variation distance from $\mu_{\beta,H}$ is at most $\eps$.
\end{definition}

\begin{definition}[$(B, \gamma,\delta)$-Quantum Decoding Problem]
\label{def:QDP}
    Let $B \in \mathbb{F}_2^{M\times N}$ and let $C^\perp$ be the classical error correcting code defined by taking $B^T$ as its parity check matrix:
    \begin{equation}
        \label{eq:Cperp}
        C^\perp = \{ \mathbf{d} \in \mathbb{F}_2^M : B^T \mathbf{d} = \mathbf{0}\}.
    \end{equation}
    Given $\gamma \in [0,1/2]$ let
    \begin{equation}
        \ket{\psi_{0,\gamma}} = \sqrt{1-\gamma} \ket{0} + \sqrt{\gamma} \ket{1} \quad \textrm{and} \quad \ket{\psi_{1,\gamma}} = \sqrt{1-\gamma} \ket{1} + \sqrt{\gamma} \ket{0}.
    \end{equation}
    Given $\mathbf{d} = (d_1,\ldots,d_M) \in \mathbb{F}_2^M$ let
    \begin{equation}
        \label{eq:productform}
        \ket{\psi_{\mathbf{d},\gamma}} = \bigotimes_{j=1}^M \ket{\psi_{d_j,\gamma}}.
    \end{equation}
    The \emph{$(B, \gamma,\delta)$-quantum decoding problem} at error rate $\gamma$ is, for uniformly random $\mathbf{d}\in C^\perp$, to identify $\mathbf{d}$ from $\ket{\psi_{\dv,\gamma}}$ with probability at least $1 - \delta$.
\end{definition}

\begin{remark}
    The problem in \cref{def:QDP} is a special case of the more general quantum decoding problem formulated in \cite{chailloux2023quantum} and explored further in \cite{SMR26}. Here, motivated by Gibbs sampling, we restrict our attention to input states with the tensor product form \cref{eq:productform}.
    If we take a classical algorithm for decoding LDPC codes and implement it as a reversible circuit, then upon the input state \eqref{eq:productform},  the algorithm only ``sees'' a classical error model where each bit has been independently flipped with probability $\gamma$; such classical decoders were the starting point for DQI in~\cite{jordan2024optimization}. However, in this paper we will use an intrinsically quantum decoding method from \cite{chailloux2023quantum} that takes advantage of the coherence of these bit flip errors.
\end{remark}
The Decoded Quantum Interferometry algorithm reduces certain $(H, \beta,\eps)$-Gibbs sampling problems to an associated $(B, \gamma,\delta)$-quantum decoding problem, as we show in the below theorem (proved in \cref{sec:DQI}):
\begin{theorem}[Gibbs sampling from quantum decoding]
    \label{thm:gamma_vs_beta}
    Given any $B \in \mathbb{F}_2^{M \times N}$, and uniformly random $J_1,\ldots,J_M \sim_\iid \{+1,-1\}$. Let $H_{B,J}:\{+1,-1\}^N \to \mathbb{Z}$ be the classical Hamiltonian 
    \begin{equation}
        \label{eq:HBJ}
        H_{B,J}(\mathbf{z}) = -\sum_{i=1}^M J_i \prod_{j:B_{ij} = 1} z_j.
    \end{equation}
Fix any $\delta \in [0,1]$ and $\gamma \in [0,1/2]$, and set
    \begin{equation}
        \label{eq:betamax}
        \beta = \sech^{-1} \left( 1- 2 \gamma \right).
    \end{equation}
    Given a subroutine which efficiently solves the $(B, \gamma,\delta)$-quantum decoding problem, then DQI can efficiently solve the $(H_{B, J}, \beta, \eps)$-Gibbs sampling problem for any $\eps\ge \sqrt{\delta}$, with probability at least $1 - \sqrt{\delta}/\eps$ over the $J_i$'s.
\end{theorem}

To more concretely assess the performance of the DQI sampling algorithm, we consider the following ensemble of $(B,J)$ corresponding to Ising $k$-spin glass models on Erd\H{o}s-R\'enyi hypergraphs. Note that minimizing \eqref{eq:HBJ} is equivalent to a Max-$k$-XORSAT problem that seeks to maximize the number of satisfied constraints in the linear system $B\xv =\vv$, with $z_i=(-1)^{x_i}$ and $J_i = (-1)^{v_i}$.

\begin{definition}[Erd\H{o}s-R\'enyi ensemble $\GG_\ER(N,k,D)$]
\label{def:erdos_renyi}
    Given positive integers $N \ge k \ge 2$, and any $D >0$, we define $\GG_\ER(N,k,D)$ be the following ensemble of random weighted $k$-uniform hypergraphs with average degree (approximately) $D$. Let the vertex set be $V=\{1,\ldots,N\}$. Let $M=\lfloor ND/k\rfloor$, and the hyperedge multiset $E=(e_1,\ldots,e_M)$ is obtained by drawing each $e_i$ independently and uniformly from the collection $\binom{V}{k}$ of all $k$-element subsets of $V$ (i.e., with replacement). 
    The weights $J=(J_1,\ldots,J_M)\in\{+1,-1\}^M$ are i.i.d.\ uniform Rademacher and independent of $E$. 
    Each hypergraph $G=(V,E)$ is equivalently described by an incidence matrix $B\in\FF_2^{M\times N}$ whose row $i$ has ones exactly at the $k$ positions in $e_i$, so $\GG_\ER(N,k,D)$ also defines a distribution on matrices over $\FF_2$ in which each row has exactly $k$ non-zero entries placed uniformly at random.
    We will often write $H=(B,J)\sim\GG_\ER(N,k,D)$, where the relationship between $H$ and $(B,J)$ is given in Eq.~\eqref{eq:HBJ}. 
\end{definition}

Many natural sampling algorithms are predicted to fail at temperatures beyond a dynamical phase transition. For the sparse hypergraph ensemble above, the predicted transition occurs at inverse temperature $\beta_\dyn(k, D)$ (see \cref{app:beta-dyn} and Eq.~\eqref{eq:xor-beta-dyn-def} for a precise definition), which asymptotically is
\begin{equation}
\label{eq:beta-dyn-asymp}
\beta_{\dyn}(k, D) = \sqrt{\frac{2\ln k}{D}}\cdot [1+o_{k}(1)].
\end{equation}
For example, recent results~\cite{AMS22sampling, alaoui2023sampling} showed that for densely connected variants of these $k$-spin glasses, no stable algorithms can sample their Gibbs measures at temperatures below the onset of either shattering or replica symmetry breaking (RSB). \cite{huang2025hardness} has also extended the hardness result for $k=2$ case to sparse random graphs for temperatures in the RSB regime.
For large $k$, however, shattering is predicted to obstruct stable sampling algorithms outside the RSB regime~\cite{alaoui2023sampling, alaoui2024near}.
Here, we prove the following hardness of sampling result from shattering for the sparse $\GG_\ER$ ensemble:
\begin{theorem}[Shattering and failure of stable samplers]
\label{thm:shatter-hardness} 
	For any fixed $\eps >0$, there exists a cutoff $k_0(\eps)$ such that the following holds for all fixed $k\ge k_0(\eps)$ and $D\ge k$ as $N\to\infty$.
	Let $H=H_{B,J}$ with $(B,J)\sim\GG_\ER(N,k,D)$ as in \cref{def:erdos_renyi}. For any $\beta$ satisfying
	\begin{equation} \label{eq:beta-range}
	(1+\eps)\sqrt{\frac{2\ln k }{D} } \le \beta  \le \sqrt{\frac{k}{2D}},
	\end{equation}
	the Gibbs measure $\mu_{\beta, H}$ shatters into exponentially many well-separated clusters with high probability.
	In this interval of $\beta$, no stable algorithm can sample from $\mu_{\beta,H}$ with vanishing $o_N(1)$ expected normalized Wasserstein distance, and Glauber dynamics has $e^{\Omega(N)}$ mixing time with high probability.
\end{theorem}
The proof is in \cref{sec:stable-hardness}, where we provide a formal definition of stable algorithms in \cref{def:stability}. 
Roughly, we say an algorithm is stable if resampling one input clause changes the output distribution by a vanishing amount in Wasserstein distance.
Many natural algorithms are stable, including low-degree polynomial algorithms. 
The exponential mixing time result for Glauber dynamics follows separately from shattering-induced bottlenecks in the Gibbs measure and does not require stability of the dynamics.

To establish the above shattering and hardness results, we work in the replica-symmetric regime and analyze a simpler ``planted'' model before transferring the result to $\GG_\ER$ (see e.g.,~\cite{alaoui2023shattering}).
In this temperature interval, we prove that a typical Gibbs sample has an annular gap in Hamming distance where every configuration has strictly worse energy---a phenomenon called ``soft OGP''~\cite{alaoui2024near}. Following similar techniques as in \cite{alaoui2024near}, we then use the soft OGP to construct well-separated, small-diameter sets that carry almost all the Gibbs mass, yielding a shattering decomposition (see formal statement in \cref{def:shattering}). 
To prove hardness against stable algorithms, we show that the shattered Gibbs measure is highly sensitive to changing any single hyperedge clause.
This sensitivity, also known as disorder chaos, is incompatible with accurate sampling by a stable algorithm.

We note that the lower endpoint of \eqref{eq:beta-range} is asymptotically $(1+\eps)\beta_\dyn(k, D)$  as $k\to\infty$ for any $\eps >0$, consistent with the predicted dynamical transition. We leave open a proof of hardness for stable samplers throughout $\beta > \beta_\dyn(k, D)$. In particular, we expect our shattering-based obstruction extends across the full interval of $\big(\beta_\dyn(k, D), \beta_\RS(k,D)\big)$, where $\beta_\RS(k,D)$ marks the static or replica symmetry breaking transition.
For $\beta > \beta_\RS(k,D)$, we expect the onset of replica symmetry breaking also implies disorder chaos and thus failure of stable algorithms. Such an implication is known from \cite{alaoui2023sampling} in dense spin glasses, and was extended to sparse spin glasses in \cite{huang2025hardness} for the $k=2$ case.

Despite the above hardness result, we show that DQI is able to Gibbs sample in the shattered phase:
\begin{theorem}[Gibbs sampling by DQI+USD]
\label{thm:gamma_USD} 
    Fix any $k\ge 3$ and $D\ge k$.
    Let
    \begin{equation}
        \gamma_\USD(k,D) = \frac{1}{2} - \frac{1}{2} \sqrt{1-\Big( \frac{K_*(k)}{D} \Big)^2},
    \end{equation}
    where we define $K_*(k)$ variationally via
    \begin{equation}
        \Phi_{d,k}(x) = e^{-d x^{k-1}} + d x^{k-1} - \frac{d(k-1)}{k}x^k - \frac{d}{k} 
        \quad \text{and} \quad 
        K_*(k) := \sup\Bigl\{ d>0 : \max_{x\in[0,1]}\Phi_{d,k}(x)=1-\frac{d}{k} 
        %\ \text{(with unique maximizer at $x=0$)}
        \Bigr\}.
        \label{eq:admissible}
    \end{equation}
    Let  $(B, J) \sim\mathbb{G}_{\mathrm{ER}}(N,k,D)$, where $B\in \FF_2^{M\times N}$ is the incidence matrix of a random Erd\H{o}s-R\'enyi hypergraph.
    For every fixed $\delta \in (0,1)$ and $\gamma \in [0, \gamma_\USD)$,
    the quantum USD decoder of \cite{chailloux2023quantum} can efficiently solve the $(B, \gamma,\delta)$-quantum decoding problem with high probability $(1-o_{N}(1))$ over the choice of $B$.
    With \cref{thm:gamma_vs_beta}, this implies that for any $\eps >0$ and
    \begin{equation} \label{eq:beta-ALG}
        \beta <  \beta_{\ALG}(k, D) := \tanh^{-1} \Big(\frac{K_*(k)}{D} \Big),
    \end{equation}
    and sufficiently large $N$, DQI can efficiently solve the $(H_{B, J}, \beta,\eps)$-Gibbs sampling problem with probability at least $1-\sqrt{\delta}/{\eps} - o_N(1)$ over the choice of $(B,J)$. 
\end{theorem}
We prove \cref{thm:gamma_USD} in Section~\ref{sec:USD}. The proof gives a polynomial-time DQI-based sampler with $o_N(1)$ expected total variation distance over $(B,J)$.
\begin{remark}[Meaning of $K_*(k)$]
$K_*(k)$ is in fact quite close to $k$. We show later in \cref{cor:dk-asymptotics} that
    \[
        \frac{K_*(k)}{k}\longrightarrow 1
        \qquad\text{as }k\to\infty.
    \]
Consequently, for $D=\alpha k$ with fixed $\alpha>1$,  $\beta_\ALG(k,\alpha k)\to\tanh^{-1}(1/\alpha)$ as $k\to\infty$.
\end{remark}

\begin{remark}[Total variation versus normalized Wasserstein distance]
\label{rmk:tvd-vs-wasserstein}
Throughout the paper we use two distance measures on probability distributions over $\bSig_N=\{+1,-1\}^N$. Our positive results for DQI (\cref{thm:gamma_vs_beta,thm:gamma_USD}) are stated in the stronger total variation distance $d_\TV$, while our hardness results against stable algorithms (\cref{thm:shatter-hardness}) are stated in the weaker normalized Wasserstein distance $W_1$ (see \cref{eq:W1-def} for a formal definition). The two are related by
\[
W_1(\mu,\nu)\le 2\,d_\TV(\mu,\nu)
\qquad\text{for all }\mu,\nu\text{ on }\bSig_N,
\]
which follows by taking the optimal TV coupling and using $\|\bX-\bY\|_1\le 2N$ for any $\bX,\bY\in\bSig_N$. Hence, an $\eps$-TVD sampler is also a $2\eps$-Wasserstein sampler, and hardness in $W_1$ also implies hardness in total variation.
\end{remark}

Observe that for any fixed $\alpha > 1$ and $D=\alpha k$, \cref{thm:shatter-hardness,thm:gamma_USD} tell us that the DQI+USD can efficiently sample up to $\beta_{\ALG}(k, \alpha k)\simeq \tanh^{-1}(1/\alpha)$ for large $k$, while stable algorithms provably fail at a much smaller $\beta_{\dyn}(k, \alpha k) \approx \sqrt{2\ln k / (\alpha k)}$.
(In particular the latter tends to $0$ as $k\to\infty$ but the former does not.)
This shows that failure of stable algorithms does not pose a barrier to the success of DQI.

However, we emphasize that these results do not prove quantum advantage.
In fact, the performance of the quantum USD decoder was previously shown in \cite{chailloux2023quantum} to be matched by Prange's algorithm~\cite{prange1962use}.
Indeed, as we show in Appendix~\ref{sec:prange}, Prange's algorithm can be adapted to give a classical (but non-stable) method which provably samples with $o_N(1)$ expected total variation distance whenever $D \tanh\beta  < K_*(k)$.
This is exactly the same sufficient condition \eqref{eq:beta-ALG} for DQI+USD at each fixed $k\ge 3$ and $D\ge k$.
To adapt Prange's algorithm for sampling, one simply retains each clause independently with probability $\tanh \beta$ and samples uniformly from the solutions of the resulting linear system.
Both Prange's and DQI+USD evade the obstruction against stable algorithms because of their linear-algebraic subroutines.
Nevertheless, DQI with other decoders was shown to surpass Prange's algorithm on certain XORSAT instances~\cite{jordan2024optimization}, so we imagine that better temperatures than $\beta_{\ALG}$ might be achievable with DQI using better decoders.

To put possible improvements to DQI-based sampling in context, we observe that information-theoretic limits can be placed on the maximum decodable rate for random $B$ from a sparse Max-XORSAT ensemble such as $\GG_{\ER}(N, k,D)$. 
This provides a limitation to the DQI approach to Gibbs sampling. 
If DQI uses a decoder based on executing a classical algorithm in superposition to uncompute codewords of $C^\perp$, then the Shannon bound yields the following inverse temperature limit
\begin{equation}
\beta_{\mathrm{Shannon}}(k,D) = \sech^{-1} \left[ 1 - 2 H_2^{-1} \left( \frac{k}{D} \right) \right],
\end{equation}
where $H_2^{-1}$ is the inverse of the binary entropy function for its branch on $[0,1/2]$.
For a more general quantum decoder, the Holevo bound gives
\begin{equation}  \label{eq:beta-Holevo}
 \beta_{\mathrm{Holevo}}(k,D) = \sech^{-1} \left[ 2 \sqrt{ H_2^{-1} \left( 1 - \frac{k}{D} \right)\left( 1 - H_2^{-1}\left( 1 - \frac{k}{D} \right)\right)} \right].
\end{equation}
See \cref{app:holevo_shannon} for a derivation. We remark these channel bounds are necessary conditions for reliable decoding, and they do not imply that efficient decoders exist to attain them. Furthermore, as noted in \cite[Theorem 12]{garrido2026regev}, there exist circumstances where Regev's reduction can, in principle, reliably yield the desired output state even when it is information-theoretically impossible to reliably solve the quantum decoding problem. Thus, while the Holevo bound is a natural threshold, which characterizes most natural approaches to coherent Gibbs sampling via DQI and Regev reductions, it does not rigorously imply success or failure of the approach in complete generality.

Several interesting questions remain open:
\begin{enumerate}
    \item Can DQI with an efficient classical decoder surpass stable algorithms? When $k \ll D$, the Shannon upper bound is
    \[
    \beta_{\rm Shannon}(k, D)
        \simeq 2\sqrt{H_2^{-1}(k/D)} = \sqrt{\frac{k/D}{\Theta(\ln (D/k))}},
    \]
    which exceeds the stable algorithmic threshold $\beta_\dyn$ for some $k,D$. 
	While this shows that it is information-theoretically possible for DQI with a classical decoder to beat stable algorithms, known polynomial-time classical decoding algorithms fall increasingly short of the Shannon bound as $k$ and $D$ grow.
	For various finite $(k,D)$, it would be interesting to compare the predicted threshold $\beta_\dyn(k,D)$, which we give in \cref{app:spin-glass}, with the performance of DQI with various classical LDPC decoders.

    \item Can DQI efficiently perform Gibbs sampling at temperatures below the replica symmetry breaking (RSB) transition? The RSB transition occurs at a colder temperature than the shattering/dynamical transition ($\beta_\RS\ge \beta_\dyn$), and is thought to be more challenging to reach classically.
	The Holevo bound gives the following upper bound on the inverse temperature accessible through reliable quantum decoding:
    \begin{equation*}
        \beta_{\rm Holevo}(k,D)  \simeq  \sqrt{\frac{k \ln 4}{D}} \qquad \text{ when } \quad k \ll D.
    \end{equation*}
    The exact expression \eqref{eq:beta-Holevo} coincides with the inverse temperature at which the annealed entropy vanishes, and marginally exceeds $\beta_\RS(k,D)$ in the Erd\H{o}s-R\'enyi ensemble at every finite $D > k$ (see \cref{prop:beta-RS-Holevo}).
    The question, however, is whether polynomial-size quantum circuits can approach this information-theoretic limit closely enough to also cross the replica symmetry breaking barrier.
    \item Is there an ensemble of Hamiltonians where DQI efficiently samples at temperatures beyond those reached by all known classical algorithms?
	Answering this question would require, at a minimum, comparing the performance of DQI with the best available efficient decoders against the predicted thresholds for stable algorithms, as well as linear-algebraic methods such as Prange's algorithm.
	Besides Max-XORSAT ensembles at various finite connectivity parameters $(k,D)$, it would also be interesting to consider more general LINSAT problems over larger finite fields for which DQI is suitable~\cite{jordan2024optimization}.
	For example, the Max-$q$-cut problem can also be reduced via DQI to decoding linear code over $\FF_q$; the corresponding stable algorithm threshold corresponds to the dynamical transition in the antiferromagnetic Potts model. 
	We may also consider quantum Hamiltonians, for which DQI can be generalized to produce Gibbs samples~\cite{schmidhuber2025hamiltonian}. Recently, \cite{zlokapa2025average} showed that topological barriers from glass transition in certain random ensembles of Pauli Hamiltonians also obstruct stable quantum algorithms. It would be interesting to explore if Hamiltonian DQI can overcome these barriers, similar to what we have shown in this work for classical spin glasses.
\end{enumerate}

\section{Gibbs Sampling by Decoded Quantum Interferometry}\label{sec:gibbsSamplingProofs}
In this section,  we prove Theorems \ref{thm:gamma_vs_beta} and \ref{thm:gamma_USD}. The first theorem is established in \cref{sec:DQI}, where we reduce the Gibbs sampling problem at inverse temperature $\beta$ to a quantum decoding problem at error rate $\gamma=\frac12 - \frac12\sech \beta$. Then in \cref{sec:USD}, we characterize the performance of the USD-based decoder for the Erd\H{o}s-R\'enyi ensemble and prove the second theorem.

\subsection{Reducing Classical Gibbs Sampling to Quantum Decoding}
\label{sec:DQI}
For any bit string $\mathbf{x} \in \mathbb{F}_2^N$ let $(-1)^{\mathbf{x}} = ((-1)^{x_1},\ldots,(-1)^{x_N})$. To sample from the classical Gibbs distribution $\mu_{\beta,H}(\mathbf{z}) = \frac{1}{Z_{\beta,H}} e^{-\beta H(\mathbf{z})}$ it would suffice to produce the quantum state
\begin{equation}
    \ket{\sqrt{\mu_{\beta,H}}} = \sum_{\mathbf{x} \in \mathbb{F}_2^N } \sqrt{\mu_{\beta,H}[(-1)^{\mathbf{x}}]} \ket{\mathbf{x}}
\end{equation}
and then measure in the computational basis. Instead, given a Hamiltonian $H_{B,J}$ of the form \cref{eq:HBJ} we find it more convenient to produce the state
\begin{equation}
    \label{eq:phitarget}
    \ket{\Phi_{\mathrm{target}}} =  \frac{1}{\sqrt{|\ker B|}} \sum_{\xv \in \FF_2^N}  \sqrt{\mu_{\beta,H_{B,J}}[(-1)^{\xv}]} \ket{B \xv}.
\end{equation}
Since $H_{B,J}(\zv=(-1)^\xv)$ in \cref{eq:HBJ} depends on $\mathbf{x}$ only through the string $B\mathbf{x}$, any two inputs $\mathbf{x}$ that produce the same $\mathbf{y}=B\mathbf{x}$ carry the same Gibbs weight $\mu_{\beta,H_{B,J}}[(-1)^{\mathbf{x}}]$, and their amplitudes in \cref{eq:phitarget} pile up on the single basis state $\ket{\mathbf{y}}$ for some $\yv=B\xv$. Every such $\mathbf{y}$ has the same number of preimages, namely $|\ker B|$, so the prefactor $1/\sqrt{|\ker B|}$ makes $\ket{\Phi_{\mathrm{target}}}$ a unit vector. Measuring it in the computational basis then returns a given $\mathbf{y}$ with probability equal to the total Gibbs weight of its preimages $\{\mathbf{x}: B\mathbf{x}=\mathbf{y}\}$, and since those preimages are equally weighted, the Gibbs distribution restricted to them is uniform. Therefore, given $\ket{\Phi_{\mathrm{target}}}$ one can sample from $\mu_{\beta,H}$ exactly: measure in the computational basis to obtain a string $\mathbf{y}=B\mathbf{x}$, then output a uniformly random solution $\mathbf{x}$ of $B\mathbf{x}=\mathbf{y}$, which Gaussian elimination over $\mathbb{F}_2$ produces as a particular solution plus a uniformly random element of $\ker B$. Note when $B$ has full column rank, $\ker B=\{\mathbf{0}\}$ and this last step reduces to deterministically inverting $\mathbf{y}=B\mathbf{x}$.

By \cref{eq:gibbs} and \cref{eq:HBJ}, we can rewrite \cref{eq:phitarget} as
\begin{equation}
    \ket{\Phi_{\mathrm{target}}} = \sum_{\mathbf{c} \in C} \mathcal{F}(\mathbf{c}-\mathbf{v}) \ket{\mathbf{c}}
\end{equation}
where $C$ is the code generated by $B$,
\begin{equation}
    C = \{B \mathbf{x} : \mathbf{x} \in \mathbb{F}_2^N \},
\end{equation}
$\mathbf{v} \in \mathbb{F}_2^M$ is the bit string such that $(-1)^{\mathbf{v}} = (J_1,\ldots,J_M)$, namely
\begin{equation}
    v_i = (1-J_i)/2,
\end{equation}
and
\begin{equation}
    \label{eq:F}
    \mathcal{F}(\mathbf{y}) \propto e^{-\beta |\mathbf{y}|},
\end{equation}
with the proportionality constant fixed by the normalization of $\ket{\Phi_{\mathrm{target}}}$.
Here $|\cdot|$ denotes Hamming weight.
(One might naively expect $\mathcal{F}(\mathbf{y}) \propto  e^{-\beta |\mathbf{y}|/2}$ as the amplitude is the square root of the Gibbs measure, but recall that $H$ assigns an energy of $-1$ to each satisfied clause and $+1$ to each unsatisfied clause and thus, up to an offset, is equal to twice the Hamming weight of $\mathbf{y}$.)

By the Fourier convolution theorem (see e.g., \cite{SMR26}), the Hadamard transform of $\ket{\Phi_{\mathrm{target}}}$ is
\begin{equation}
    \label{eq:phihat}
    \ket{\widehat{\Phi}_{\mathrm{target}}} = \sum_{\mathbf{d} \in C^\perp} \sum_{\mathbf{e} \in \mathbb{F}_2^M} (-1)^{\mathbf{v} \cdot{\mathbf{e}}} \widehat{\mathcal{F}}(\mathbf{e}) \ket{\mathbf{d} + \mathbf{e}},
\end{equation}
where $\widehat{\mathcal{F}}$ is the Hadamard transform of $\mathcal{F}$ and $C^\perp$ is the code dual to $C$, given by \cref{eq:Cperp}. By \cref{eq:F}, we see that $\mathcal{F}$ is a product of functions on each variable $y_1,\ldots,y_M$, namely
\begin{equation}
    \mathcal{F}(\mathbf{y}) \propto \prod_{i=1}^M e^{-\beta y_i}.
\end{equation}
This makes it easy to obtain the Hadamard transform, namely
\begin{equation}
    \label{eq:Fhat}
    \widehat{\mathcal{F}}(\mathbf{e}) \propto \sqrt{1-\gamma}^{M-|\mathbf{e}|} \sqrt{\gamma}^{|\mathbf{e}|},
\end{equation}
where
\begin{equation}
    \label{eq:gamma_from_beta}
    \gamma = \frac{1}{2} - \frac{1}{2} \sech \left(\beta \right).
\end{equation}
Solving this for $\beta$ yields \cref{eq:betamax}.

The DQI/Regev method of preparing $\ket{\Phi_{\mathrm{target}}}$ is as follows. First, prepare
\begin{equation}
    \left( \frac{1}{\sqrt{|C^\perp|}} \sum_{\mathbf{d} \in C^\perp} \ket{\mathbf{d}} \right) \otimes \sum_{\mathbf{e} \in \mathbb{F}_2^M} \widehat{\mathcal{F}}(\mathbf{e}) \ket{\mathbf{e}}.
\end{equation}
Preparing a uniform superposition over a linear code can be done by very efficient quantum circuits, as shown in e.g., \cite{FJ24}. Preparing $\sum_{\mathbf{e} \in \mathbb{F}_2^M} \widehat{\mathcal{F}}(\mathbf{e}) \ket{\mathbf{e}}$ is efficient because by \cref{eq:Fhat} it is a tensor product over single-qubit states. Second, reversibly add the content of the first register to the second register.
\begin{equation}
        \label{eq:before_dec}
        \to \frac{1}{\sqrt{|C^\perp|}} \sum_{\mathbf{d} \in C^\perp} \ket{\mathbf{d}} \sum_{\mathbf{e} \in \mathbb{F}_2^M} \widehat{\cF}(\mathbf{e}) \ket{\mathbf{d} + \mathbf{e}}.
\end{equation}
Next, we wish to impose the phase $(-1)^{\mathbf{e} \cdot \mathbf{v}}$ and uncompute the string $\mathbf{d}$ from the first register.

To show how to do this, let us illustrate the idea with a simpler case where the states $\sum_{\mathbf{e} \in \mathbb{F}_2^M} \widehat{\cF}(\mathbf{e}) \ket{\mathbf{d} + \mathbf{e}}$ for different $\mathbf{d} \in C^\perp$ are orthogonal. This is possible, for example, if the support of $\hat\cF(\mathbf{e})$ is limited to $\mathbf{e}$ of Hamming weight less than half the minimum Hamming distance between codewords in $C^\perp$. In this case, there exists a unitary $\tilde{U}_{\mathrm{dec}}$ that solves the quantum decoding problem perfectly. That is, 
\begin{equation}
\label{eq:udec_ideal}
    \tilde{U}_{\mathrm{dec}} \ket{0\ldots0} \sum_{\mathbf{e} \in \mathbb{F}_2^M} \widehat{\mathcal{F}}(\mathbf{e}) \ket{\mathbf{d} + \mathbf{e}} = \ket{\mathbf{d}} \sum_{\mathbf{e} \in \mathbb{F}_2^M} \widehat{\mathcal{F}}(\mathbf{e}) \ket{\mathbf{d} + \mathbf{e}}.
\end{equation}
(In fact, to prove \cref{thm:gamma_vs_beta} we will need to consider non-orthogonal states and therefore use a $U_{\mathrm{dec}}$ that achieves this 
only up to a small decoder failure probability $\delta$,
and we analyze the effect of this error at the end of this subsection.
For illustrative purposes, we here consider the ideal error-free case.) 

Under this assumption \eqref{eq:udec_ideal} about the idealized decoder $\tilde{U}_{\rm dec}$, we can obtain the state $\ket{\widehat{\Phi}_{\mathrm{target}}}$ described in \cref{eq:phihat} as follows.
Let $Z^\mathbf{v} = Z^{v_1} \otimes \ldots \otimes Z^{v_M}$ be the quantum circuit operation in which a Pauli $Z$ is applied to each bit for which $\mathbf{v}$ is one. Start by applying $Z^{\mathbf{v}}$ to the first register of the state \cref{eq:before_dec}, yielding
\begin{equation}
        \to \frac{1}{\sqrt{|C^\perp|}} \sum_{\mathbf{d} \in C^\perp} (-1)^{\mathbf{d} \cdot \mathbf{v}} \ket{\mathbf{d}} \sum_{\mathbf{e} \in \mathbb{F}_2^M} \widehat{\mathcal{F}}(\mathbf{e}) \ket{\mathbf{d} + \mathbf{e}}.
\end{equation}
Next, apply $\tilde{U}_{\mathrm{dec}}^\dag$ to the pair of registers. This can be done by applying the inverses of the gates composing $\tilde{U}_{\mathrm{dec}}$ in reverse order. By \cref{eq:udec_ideal}, this yields
\begin{equation}
        \to \frac{1}{\sqrt{|C^\perp|}} \sum_{\mathbf{d} \in C^\perp} (-1)^{\mathbf{d} \cdot \mathbf{v}} \ket{\mathbf{0}} \sum_{\mathbf{e} \in \mathbb{F}_2^M} \widehat{\mathcal{F}}(\mathbf{e}) \ket{\mathbf{d} + \mathbf{e}}.
\end{equation}
Next, apply $Z^{\mathbf{v}}$ to the second register, yielding
\begin{equation}
        \to \frac{1}{\sqrt{|C^\perp|}} \sum_{\mathbf{d} \in C^\perp}  \ket{\mathbf{0}} \sum_{\mathbf{e} \in \mathbb{F}_2^M} (-1)^{\mathbf{e} \cdot \mathbf{v}} \widehat{\mathcal{F}}(\mathbf{e}) \ket{\mathbf{d} + \mathbf{e}}.
\end{equation}
Next, we can discard the first register, since it is unentangled with the rest. This leaves our desired state $\ket{\widehat{\Phi}_{\mathrm{target}}}$. Therefore the final step is to apply the Hadamard transform, which then yields the final state $\ket{\Phi_{\mathrm{target}}}$.

By inverting \cref{eq:gamma_from_beta} one obtains \cref{eq:betamax}. However, this does not yet complete the proof of \cref{thm:gamma_vs_beta}, because the quantum decoding problem cannot be solved perfectly and \cref{eq:udec_ideal} cannot be exact. Instead, one must show that a small failure probability of the decoder translates into a small total variation distance between the quantum algorithm output and the target Gibbs distribution. We address this now, using the following theorem proven in \cite{SMR26}, but restated here in notation more suitable for the present manuscript.

\begin{theorem}[{Adapted from \cite[Theorem~4]{SMR26}}]\label{thm:distributionalcloseness}
    Fix $B \in \mathbb{F}_2^{M \times N}$ and $\beta > 0$, with $\gamma$ related to $\beta$ via \cref{eq:gamma_from_beta}.
    Let $U_{\mathrm{dec}}$ be any unitary that solves the $(B, \gamma,\delta)$-quantum decoding problem of \cref{def:QDP}.
    Let $\mathcal{D}_\mathrm{target}(B,J,\beta)$ denote the distribution over $\mathbf{y}\in\mathbb{F}_2^M$ obtained by measuring the ideal state $\ket{\Phi_{\mathrm{target}}}$ of \cref{eq:phitarget} in the computational basis (which is supported on the code $C=\{B\mathbf{x}:\mathbf{x}\in\mathbb{F}_2^N\}$), and let $\mathcal{D}_\mathrm{actual}(B,J,\beta)$ denote the output distribution of the above procedure when $U_{\mathrm{dec}}$ replaces the idealized decoder $\tilde{U}_{\mathrm{dec}}$.
    Then, for uniformly random $J \in \{+1,-1\}^M$,
    \begin{equation}
        \label{eq:tvd}
        \mathbb{E}_{J} \Big[d_\TV\bigl(\mathcal{D}_\mathrm{actual}(B,J,\beta),\, \mathcal{D}_\mathrm{target}(B,J,\beta)\bigr)\Big] \leq \sqrt{\delta},
    \end{equation}
    where $d_\TV$ denotes the total variation distance.
\end{theorem}

\begin{proof}[Proof of \cref{thm:gamma_vs_beta}]
Any concrete decoder, such as the USD decoder analyzed in \cref{sec:USD} below, will have a small but nonzero failure rate. In the context of quantum decoding algorithms, the effect of decoding errors is subtle to analyze because the error terms from different codewords in the superposition over $C^\perp$ can add constructively. Consequently, the naive expectation that a decoder failure rate $\delta$ implies that the output of DQI deviates from the target distribution by at most $O(\delta)$ is not always true. (See \cite{chailloux2023quantum} for an explicit counterexample.) Nevertheless, for average-case coefficients $J$, \cref{thm:distributionalcloseness} tells us that a small decoder failure rate induces only a small deviation from the desired output distribution. Interestingly, this suggests that for the decoding errors to be tolerable it is necessary that the coefficients $J_1,\ldots,J_M \in \{+1,-1\}$ in the Hamiltonian be chosen randomly rather than adversarially. We obtain \cref{thm:gamma_vs_beta} as a corollary of \cref{thm:distributionalcloseness} by the following argument 

Fix an arbitrary spin configuration $\mathbf{z}_0\in\{+1,-1\}^N$, and define a stochastic
post-processing channel $K_B$ as follows. On input $\mathbf{y}\in C$, choose $\mathbf{x}$ uniformly from
\[
    A_\mathbf{y} := \{\mathbf{x}\in\mathbb{F}_2^N:B\mathbf{x}=\mathbf{y}\}
\]
and output $\mathbf{z}=(-1)^{\mathbf{x}}$.  On input $\mathbf{y}\notin C$, output $\mathbf{z}_0$.

Let
\[
    \nu^{\DQI}_{B,J,\beta}
    :=
    K_B\mathcal{D}_{\mathrm{actual}}(B,J,\beta).
\]
Because $H_{B,J}((-1)^{\mathbf{x}})$ depends on $\mathbf{x}$ only through $B\mathbf{x}$, the Gibbs measure is uniform on the set $A_\mathbf{y}$.  Hence
\[
    K_B\mathcal{D}_{\mathrm{target}}(B,J,\beta)
    =
    \mu_{\beta,H_{B,J}}.
\]
Total variation distance contracts under stochastic post-processing, so
\begin{align}
    \mathbb{E}_J\!\left[
        d_\TV\!\left(
            \nu^{\DQI}_{B,J,\beta},~
            \mu_{\beta,H_{B,J}}
        \right)
    \right]
    &\leq
    \mathbb{E}_J\!\left[
        d_\TV\!\left(
            \mathcal{D}_{\mathrm{actual}}(B,J,\beta),
            \mathcal{D}_{\mathrm{target}}(B,J,\beta)
        \right)
    \right] \nonumber \\
    &\leq \sqrt{\delta}.
\end{align}
Therefore, by Markov's inequality, for every $\eps>0$,
\[
    \PP_J\!\left[
        d_\TV\!\left(
            \nu^{\DQI}_{B,J,\beta},~
            \mu_{\beta,H_{B,J}}
        \right)>\eps
    \right]
    \leq
    \frac{\sqrt{\delta}}{\eps}.
    \qedhere
\]
\end{proof}

\subsection{Solving the Quantum Decoding Problem Using USD}
\label{sec:USD}

In this section, we briefly recount the method for solving quantum decoding problems using Unambiguous State Discrimination (USD) measurements, which was introduced in \cite{chailloux2023quantum}. We then analyze its performance on our specific ensemble of codes arising from the Erd\H{o}s-R\'enyi ensemble, and prove \cref{thm:gamma_USD} as a consequence.

An obvious way to infer $\mathbf{d} \in C^\perp$ given the state $\sum_{\mathbf{e} \in \mathbb{F}_2^M} \widehat{\mathcal{F}}(\mathbf{e}) \ket{\mathbf{d} + \mathbf{e}}$ would be to measure the state in the computational basis and try to classically infer $\mathbf{d}$ from the measured value of $\mathbf{d} + \mathbf{e}$. Due to the product structure of $\widehat{\mathcal{F}}$, apparent from \cref{eq:Fhat}, our quantum decoding problem is one of inferring the codeword $\mathbf{d}$ after each bit in $\mathbf{d}$ has been sent separately through the following channel that takes a classical input to a quantum output state:
\begin{equation}
    \label{eq:quantum_channel}
    d_j \to \sqrt{1-\gamma} \ket{d_j} + \sqrt{\gamma} \ket{\lnot d_j}.
\end{equation}
By measuring in the computational basis, we collapse this classical-quantum channel to a classical bit flip channel where each bit is flipped with probability $\gamma$. This is called the binary symmetric channel $\mathrm{BSC}(\gamma)$. Inferring the original codeword $\mathbf{d}$ from the bit string $\mathbf{d} + \mathbf{e}$ resulting from the binary symmetric channel is a well-studied problem in classical error correction, for which classical algorithms such as belief propagation are highly effective when $B^T$ is very sparse.

In the context of Regev's reduction, one cannot actually measure $\ket{\mathbf{d} + \mathbf{e}}$ because one must solve a superposition over quantum decoding problems for different codewords $\mathbf{d} \in C^\perp$. Instead, one must implement the chosen classical decoding algorithm as a reversible circuit and apply it to the superposition. The measurements serve only as a conceptual tool to aid us in thinking about the quantum decoding problem.

By using the conceptual tool of fictitious measurements one can formulate other approaches to the quantum decoding problem. In \cite{chailloux2023quantum}, Unambiguous State Discrimination (USD) measurements were introduced for this purpose. Given two non-orthogonal states $\ket{\psi_0}$ and $\ket{\psi_1}$, the corresponding USD measurement is a three-outcome POVM that takes $\ket{\psi_b}$ as input for unknown $b \in \{0, 1\}$ and will either identify with certainty that $b = 0$, identify with certainty that $b = 1$, or produce the $\perp$ outcome which indicates no information regarding the value of $b$. A larger inner product $|\langle \psi_0 | \psi_1 \rangle|$ necessitates a larger probability on the $\perp$ outcome. Specifically, in the case of uniform priors, the optimal USD measurement achieves
\begin{equation}
    p_{\perp} = |\langle \psi_0 | \psi_1 \rangle|.
\end{equation}
Hence, if one applies a USD measurement to the states produced by the classical-quantum channel of \cref{eq:quantum_channel} then instead of a classical binary symmetric channel with bit flip probability $\gamma$ one obtains a classical erasure channel with bit erasure probability
\begin{equation}
    \label{eq:erasure_rate}
    \hat{e} = 2 \sqrt{\gamma (1-\gamma)}.
\end{equation}

Information-theoretically, the bit flip channel has higher capacity than this erasure channel. However, the advantage of an erasure channel is that it is easier to saturate the information-theoretic capacity using polynomial-time decoders. Specifically, one can find the erased bit values by solving the set of $\mathbb{F}_2$-linear equations defined by the parity check matrix. This does not depend on $B^T$ being sparse. Rather, $e$ erasure errors can be reliably decoded if random $N \times e$ submatrices of $B^T \in \mathbb{F}_2^{N \times M}$ have high probability of having rank at least $e$. This is a purely classical question about the matrices $B$ drawn from the Erd\H{o}s-R\'enyi hypergraph ensemble, which we now address.

We first describe precisely what must be decoded. By \cref{def:erdos_renyi}, the rows of $B$ are i.i.d.\ uniform $k$-subsets of $V$. The USD erasure channel erases each of the $M$ coordinates of a codeword $\mathbf{d}\in C^\perp\subseteq\FF_2^M$ independently with probability $\hat e$; let $S\subseteq\{1,\dots,M\}$ denote the resulting (random) set of erased coordinates. The erased bits $(d_j)_{j\in S}$ are uniquely recoverable from the parity checks $B^T\mathbf{d}=\mathbf{0}$ by Gaussian elimination if and only if the columns of $B^T$ indexed by $S$ are linearly independent over $\FF_2$, equivalently if and only if the row-submatrix $B_S$ (the rows of $B$ indexed by $S$) has full row rank. Conditioned on $|S|=m$, the matrix $B_S$ consists of exactly $m$ i.i.d.\ uniform $k$-subsets of $V$, as in \cref{def:erdos_renyi} of the ensemble $\GG_\ER$.
The only discrepancy between $B_S$ and the ensemble $\GG_\ER(N,k,\hat e D)$ is that the number of rows $|S|=m\sim\text{Bin}(M,\hat e)$ is random, whereas \cref{def:erdos_renyi} prescribes the deterministic row count $\lfloor N\hat e D/k\rfloor$ for $\GG_\ER$. This discrepancy slightly complicates the proof that follows, but nevertheless can be reconciled as we will see at the end of this subsection.

We next invoke the following result, quoted from \cite[Theorem~2]{coja2024k} (which was previously proved in \cite{dubois20023,pittel2016satisfiability} in the respective cases $k=3$ and $k>3$). The premise of \cite[Theorem~2]{coja2024k} matches our $\GG_\ER(N,k,D)$ exactly: each row of the matrix is an independent uniformly random $k$-subset of $[N]$ and the row count is deterministic.

\begin{theorem}[{Adapted from \cite[Theorem~2]{coja2024k}}]
\label{thm:coja-fullrank}
    For $k\ge 3$, let $K_*(k)$ be as defined in \cref{eq:admissible}, which we reproduce here for the reader's convenience:
    \[
        \Phi_{D,k}(x) = e^{-D x^{k-1}} + D x^{k-1} - \frac{D(k-1)}{k}x^k - \frac{D}{k} 
        \quad \text{and} \quad 
        K_*(k) := \sup\Bigl\{ D>0 : \max_{x\in[0,1]}\Phi_{D,k}(x)=1-\frac{D}{k} 
        %\ \text{(with unique maximizer at $x=0$)}
        \Bigr\}.
    \]
    Then an incidence matrix $B\in \FF_2^{M\times N}$ drawn from the Erd\H{o}s-R\'enyi ensemble $\mathbb{G}_\ER(N,k,D=kM/N)$ in \cref{def:erdos_renyi} has full row rank with high probability if $D< K_*(k)$, and fails to have full row rank with high probability if $D> K_*(k)$.
\end{theorem}

\begin{lemma}
    \label{cor:dk-asymptotics}
    We have
    \[
        \lim_{k\to\infty}\frac{K_*(k)}{k}=1.
    \]
\end{lemma}

\begin{proof}
The upper bound $K_*(k)\le k$ is immediate. Indeed, if $K_*(k)>(1+\eps) k$ for any $\eps>0$, then one can choose $D=(1+\eps/2)k < K_*(k)$, but the matrix $B$ cannot have full row rank when $M=ND/k = (1+\eps/2)N > N$, yielding a contradiction.

It remains to prove the matching lower bound. 
We note from \cite[Fact 2.8]{coja2024k} that the maximizer of $\Phi_{D,k}(x)$ is unique and located at $x=0$ for $D< K_*(k)$.
Fix $\delta\in(0,1)$ and set $D=(1-\delta)k$ large. Define
\begin{equation}
    G_{D,k}(x):=\Phi_{D,k}(x)-\Phi_{D,k}(0) = e^{-D x^{k-1}} + D x^{k-1} - \frac{D(k-1)}{k}x^k -1 .
\end{equation}
We must show that $G_{D,k}(x)\le 0$ for all $x\in[0,1]$, with strict inequality for $x>0$, once $k$ is sufficiently large.\\

\noindent
To proceed, let us split $x \in (0,1]$ into two regions.\\

\noindent
\textbf{Region 1: $x\in (0,1-k^{-2/3}]$.} Using the fact that $e^{-u} \le 1-u+u^2/2$ for $u\ge0$, we have
\begin{equation}
    G_{D,k}(x) \le \frac{D^2}{2} x^{2k-2} - \frac{D(k-1)}{k}x^k = D x^k \Bigl(\frac{D}{2} x^{k-2} - \frac{k-1}{k}\Bigr).
\end{equation}
In this region, $x^{k-2} \le e^{-k^{1/3}/2}$ for sufficiently large $k$. With $D=(1-\delta)k$, we have $\frac{D}{2}x^{k-2}\le (1-\delta)\frac{k}{2}e^{-k^{1/3}/2}\xrightarrow{k\to\infty}0$. Therefore, $G_{D,k}(x) < 0$ throughout this region for sufficiently large $k$. \\

\noindent
\textbf{Region 2: $x\in [1-k^{-2/3},1]$.} Let us parameterize $x=1-y/k$ with $y\in[0,k^{1/3}]$. Then $x^{k-1}=e^{-y}(1+O(y^2/k))$. With $D=(1-\delta) k$, we have for sufficiently large $k$ that
\begin{equation}
    \label{eq:exp-d-x-k-1-bound}
    e^{-Dx^{k-1}} = \exp\Bigl[-(1-\delta) k e^{-y} \times (1+o(1) )\Bigr].
\end{equation}
Furthermore,
\begin{equation}
    D x^{k-1} - \frac{D(k-1)}{k}x^k = D x^{k-1} \Bigl[1- \Big(1-\frac{1}{k}\Big)\Big(1-\frac{y}{k}\Big)\Bigr] 
    = (1-\delta)(1+y) e^{-y} + O(e^{-y} (y+y^3)/k).
\end{equation}
Then
\begin{equation}
    \label{eq:Gdk-approx}
    G_{D,k}(x) = e^{-Dx^{k-1}} - 1 + (1-\delta)(1+y) e^{-y}  (1+ o(1)).
\end{equation}
To show $G_{D,k}(x) < 0$ in this region, we consider two further subcases:
\begin{itemize}
    \item When $y \le \frac12 \ln k$, we estimate the term \eqref{eq:exp-d-x-k-1-bound} via $k e^{-y} \ge \sqrt{k}$. Furthermore, note $(1+y)e^{-y} \le 1$ for all $y\ge 0$. Thus 
    \[
    G_{D,k}(x) \le -1 + (1-\delta) + o(1) = -\delta + o(1) < 0
    \]
    for sufficiently large $k$.
    \item When $\frac12 \ln k \le y \le k^{1/3}$, we have $ke^{-k^{1/3}} \le k e^{-y} \le \sqrt{k}$. Furthermore, 
    \[
    (1+y)e^{-y} \le (1+k^{1/3})/\sqrt{k} = o(1).
    \]
    If $ke^{-y}>\Omega(1)$ when $k$ is large, then recalling \eqref{eq:exp-d-x-k-1-bound} gives 
    \[
    e^{-Dx^{k-1}}-1 < -\Omega(1) 
    \]
    and plugging into \eqref{eq:Gdk-approx} gives $G_{D,k}(x) < 0$. If $ke^{-y}\leq o(1)$, then 
    \[
    e^{-Dx^{k-1}} - 1 = -Dx^{k-1} (1+o(1))=-(1-\delta)ke^{-y} (1+o(1)).
    \]
    Since $k \gg 1+y$, grouping terms in \eqref{eq:Gdk-approx} again implies that $G_{D,k}(x)<0$.
\end{itemize}

We have shown that $G_{D,k}(x)<0$ for all $x\in (0,1]$ for sufficiently large $k$, if $D=(1-\delta)k$ for any fixed $\delta\in (0,1)$. Consequently,
\begin{equation}
    \max_{x\in[0,1]}\Phi_{D,k}(x)=\Phi_{D,k}(0)=1-\frac{D}{k},
\end{equation}
so $D=(1-\delta)k$ satisfies the condition \eqref{eq:admissible} of $K_*(k)$ for all large $k$. Since $\delta>0$ was arbitrary, this shows
\begin{equation}
    \liminf_{k\to\infty}\frac{K_*(k)}{k}\ge 1.
\end{equation}
Combining with the upper bound $K_*(k)\le k$ completes the proof.
\end{proof}

We want to apply \cref{thm:coja-fullrank} to the submatrix $B_S$ sampled under the erasure channel to show that the USD-based decoder succeeds with high probability.
However, there is a mild discrepancy we mentioned earlier:  $B_S$ has random row count, whereas the premise of \cref{def:erdos_renyi} and \cref{thm:coja-fullrank} assumes an ensemble of matrices with deterministic row count.
We reconcile this by a one-sided coupling argument, exploiting that ``full row rank'' is preserved under deletion of rows.

\begin{proof}[Proof of \cref{thm:gamma_USD}]
Let $C^\perp$ be a code whose parity-check matrix's transpose $B$ is drawn from the Erd\H{o}s-R\'enyi hypergraph ensemble $\GG_\ER(N,k,D)$, so that the codewords of $C^\perp$ have length $M=\lfloor ND/k\rfloor$. Fix an erasure rate $\hat e$ with $\hat e D<K_*(k)$, and recall from the discussion above that decoding succeeds precisely when the row-submatrix $B_S$ indexed by the erased set $S$, with $|S|\sim\text{Bin}(M,\hat e)$, has full row rank. We now reduce this to \cref{thm:coja-fullrank} by a one-sided coupling that absorbs the random row count.

Since $\hat e D<K_*(k)$ strictly, we may fix $\eta>0$ small enough that $(1+\eta)\,\hat e D<K_*(k)$, and set $m^+:=\big\lfloor(1+\eta)\hat e M\big\rfloor$, so that the associated effective degree satisfies $k m^+/N=(1+\eta)\hat e D\,(1+o(1))<K_*(k)$ for large $N$. Draw $m^+$ i.i.d.\ uniform $k$-subsets of $V$, forming a matrix $B^+\sim\GG_\ER\!\big(N,k,km^+/N\big)$ with the deterministic row count required by \cref{thm:coja-fullrank}. Couple $B_S$ to $B^+$ by taking $B_S$ to be the first $|S|$ rows of $B^+$ on the event $\{|S|\le m^+\}$; this is a valid coupling, since conditioned on $|S|\le m^+$ both $B_S$ and the first $|S|$ rows of $B^+$ are $|S|$ i.i.d.\ uniform $k$-subsets of $V$. On the event $\{|S|\le m^+\}$ the rows of $B_S$ are a subset of the rows of $B^+$, and because full row rank is preserved under deletion of rows,
\[
    \{B^+\text{ has full row rank}\}\cap\{|S|\le m^+\}\ \subseteq\ \{B_S\text{ has full row rank}\}.
\]
Consequently
\[
    \PP\big[B_S\text{ not full row rank}\big]\ \le\ \PP\big[B^+\text{ not full row rank}\big]+\PP\big[|S|>m^+\big].
\]
The first term is $o_N(1)$ by \cref{thm:coja-fullrank}, applied to $B^+$ at effective degree $(1+\eta)\hat e D<K_*(k)$. The second term is $o_N(1)$ by a Chernoff bound:
 $|S|\sim\text{Bin}(M,\hat e)$ has mean $\hat e M= \Theta(N)$ for any fixed $\hat e$, and the threshold is $m^+ = (1+\eta) \hat{e} M+O(1)$, giving $\PP[|S| > m^+] \le e^{-\Omega(N)}$.
Hence, under the joint law of $B$ and $S$,
\[
    \PP_{B,S}\!\left[
        B_S\text{ is not full row rank}
    \right]
    =
    o_N(1).
\]
Since this probability is taken jointly over $B$ and $S$,
Markov's inequality applied to the conditional failure
probability gives, for every fixed $\delta>0$,
\[
    \PP_B\!\left[
        \PP_S\!\left(
            B_S\text{ is not full row rank}
            \,\middle|\,B
        \right)>\delta
    \right]
    \leq
    \frac{
        \PP_{B,S}\!\left[
            B_S\text{ is not full row rank}
        \right]
    }{\delta}
    =
    o_N(1).
\]
Thus, with probability $1-o_N(1)$ over $B$, the USD decoder
has failure probability at most $\delta$, and hence solves the
$(B,\gamma,\delta)$-quantum decoding problem.
Combining the condition $\hat e D<K_*(k)$ with the USD erasure
rate \cref{eq:erasure_rate}, we conclude that, for every fixed
$\delta>0$, the USD decoder solves the
$(B,\gamma,\delta)$-quantum decoding problem with probability
$1-o_N(1)$ over $B$ whenever
\begin{equation}
    \label{eq:gamma_USD_implicit}
    2\sqrt{\gamma(1-\gamma)}
    <
    \frac{K_*(k)}{D}.
\end{equation}
For the relevant range $\gamma\in[0,1/2]$, define the threshold
$\gamma_\USD(k,D)$ by equality in
\cref{eq:gamma_USD_implicit}; namely,
\begin{equation}
    \label{eq:gamma_USD}
    \gamma_\USD(k,D)
    =
    \frac{1}{2}
    -
    \frac{1}{2}
    \sqrt{
        1-\left(\frac{K_*(k)}{D}\right)^2
    }.
\end{equation}
Thus the USD decoder succeeds for every
$\gamma<\gamma_\USD(k,D)$, as claimed in
\cref{thm:gamma_USD}.

Substituting \cref{eq:gamma_USD} into
\cref{eq:betamax} gives
\begin{equation}
    \label{eq:betamax_val}
    \beta_{\ALG}(k,D)
    =
    \tanh^{-1}\!\left(\frac{K_*(k)}{D}\right).
\end{equation}
Applying \cref{thm:gamma_vs_beta} then gives the Gibbs-sampling
conclusion of \cref{thm:gamma_USD}.
\end{proof}

\section{Stable Algorithms, Shattering, and Hardness of Sampling}
\label{sec:stable-hardness}

In the remainder of this paper, we formalize the framework of stable algorithms and shattering, and prove our main hardness-of-sampling result, \cref{thm:shatter-hardness}.
We first formalize what it means for an algorithm to approximately sample from a Gibbs measure (\cref{subsec:stable-sampling}),
then introduce the shattering property and show how it implies disorder chaos and the failure of stable algorithms (\cref{sec:shattering}).
Finally, we illustrate the generality of our stability conditions through low-degree polynomial algorithms, and apply shattering to obstruct Glauber dynamics via a bottleneck argument (\cref{subsec:concrete-hardness}).

In what follows, we will often write
\begin{equation}
\label{eq:f-i-def}
H(\zv) = -\sum_{e\in E} J_{e} z_{e_1} z_{e_2} \cdots z_{e_k} =: -\sum_{i=1}^M f_i(\zv),
\end{equation}
i.e. $f_i(\zv)=1$ when $\zv$ satisfies clause $e^{(i)}\in E$ and otherwise $f_i(\zv)=-1$.

\subsection{Stable Sampling Algorithms}
\label{subsec:stable-sampling}

We consider an ``algorithm'' $\cA$ to be a function that takes as input a Hamiltonian $H$ on the set of $N$-bit strings $\bSig_N = \{+1,-1\}^N$, and outputs $\cA(H)$ which is a probability distribution over $\bSig_N$.
This encompasses the classical notion of randomized algorithms by interpreting $\cA$ as the conditional law of the output given the input $H$.
In the context of quantum algorithms, $\cA(H)$ corresponds to the diagonal of the output state's density matrix in the measurement basis. 

To formalize the notion of approximation in sampling tasks, we use the normalized Wasserstein distance, which is defined for probability measures $\mu,\nu$ on the Boolean cube as
\begin{equation}
    W_1(\mu,\nu) =  \inf_{\pi \in \Pi(\mu,\nu)} \frac{1}{N} 
    \E_{(\bX,\bY) \sim \pi} \Big[\big\|\bX - \bY\big\|_1\Big],
    \label{eq:W1-def}
\end{equation}
where the infimum is over all couplings $(\bX,\bY) \sim \pi$ with marginals $\bX \sim \mu$ and 
$\bY \sim \nu$, and $\|\cdot\|_1$ is the $\ell_1$-norm on $\RR^N$. Note that since $\bX, \bY \in \{+1,-1\}^N$, we have $\|\bX - \bY\|_1 = 2\, d_H(\bX,\bY)$, where $d_H$ denotes Hamming distance. The normalization by $N$ ensures $W_1 \in [0,2]$.

Then we define the following notion of approximate sampling algorithm:
\begin{definition}
\label{def:sampling}
We say an algorithm $\cA$ is an \textbf{$\eta$-approximate sampling algorithm} for an ensemble $\GG$ at specified temperature $\beta$ if it obeys the averaged Wasserstein estimate:
\begin{equation}
\label{eq:sampling-guarantee}
\EE_{H\sim \GG}[W_1(\mu_{\beta,H},\cA(H)) ]\leq \eta.
\end{equation}
\end{definition}

For the purpose of showing computational hardness, it is useful to consider the notion of \emph{stable algorithms} that was used for provable hardness results for optimization~\cite{gamarnik2020low,huang2022tight} and Gibbs sampling~\cite{AMS22sampling,alaoui2023sampling} for various dense spin glass models.
For the sparse models we consider in this paper, we define stability under changing a single clause in $H$.

\begin{definition}
\label{def:stability}
Let $H=-\sum_{i=1}^M f_i$ be a Hamiltonian as in \eqref{eq:f-i-def}, where each clause $f_i$ is chosen i.i.d. from some distribution.
Let $H' = H + f_M - f'_M$ be a Hamiltonian obtained by replacing a single clause $f_M$ of $H$ with another random i.i.d. copy $f'_M$.
We say an algorithm $\cA$ is $\iota$-\textbf{stable} if for $(H,H')$ chosen as above:
\[
\mathbb E
W_1\big(
\cA(H), \cA(H')
\big)
\leq \iota.
\]
\end{definition}

This notion of stability is rather lenient: it only requires that it is rare for a macroscopic fraction of the output to be affected by a single input clause change. 
For example, for fixed $(k, D)$, a local algorithm whose output on a given vertex depends only on its neighborhood of fixed radius is $O(1/N)$-stable, since resampling one clause can only affect $O(1)$ output coordinates in expectation. More generally, many natural classical algorithms---including message-passing algorithms, low-depth Boolean circuits, and low-degree polynomial algorithms---are stable. 
We will make this precise for low-degree polynomials in \cref{subsec:concrete-hardness}.

Our main hardness of sampling result, which was already mentioned as part of \cref{thm:shatter-hardness} is as follows.

\begin{theorem}[Failure of stable algorithms]
\label{thm:main-hardness}
Consider $H\sim\GG_{\rm ER}(N,k,D)$. Fix any sufficiently small $\eps>0$. 
For all $k \ge k_0(\eps)$, $D\ge k$, and $\beta$ satisfying
\begin{equation} \label{eq:shatter-range}
    \sqrt{\frac{2\ln k}{D}}(1+\eps) \le \beta \le  \sqrt{\frac{k}{2D}}
\end{equation}
stable algorithms fail to sample from the Gibbs measure $\mu_{\beta,H}$ with vanishing $o_N(1)$ expected normalized Wasserstein distance.
More precisely, there exists $\iota_* = \iota_*(k,D,\beta,\eps) > 0$ such
that if $\cA_N$ is $\iota$-stable for some $\iota\in (0,\iota_*]$ independent of $N$, then
    \[
    \liminf_{N\to\infty}
    \EE_{H\sim\GG_\ER(N,k,D)} \big[W_1(\mu_{\beta,H}, \cA_N(H))\big] >0.
    \]
\end{theorem}

We defer the proof of Theorem~\ref{thm:main-hardness} to the end of \cref{sec:shattering}, after establishing the necessary shattering and disorder chaos results.

\subsection{Shattering of Gibbs Measure and Disorder Chaos}
\label{sec:shattering}

The strategy to prove \cref{thm:main-hardness} has two steps:
first, we show that the Gibbs measure \emph{shatters} into exponentially many well-separated clusters;
second, we use shattering to establish \emph{disorder chaos}, i.e., extreme sensitivity of the Gibbs measure to single-clause perturbations.
See also recent work establishing shattering and hardness in related dense spin glass models~\cite{alaoui2023sampling, gamarnik2023shattering, alaoui2024near}.

\begin{definition}[Shattering]
\label{def:shattering}
For any $\eps>0$, we say $\mu_{\beta,H}$ is an $\eps$-\emph{shattered} Gibbs measure if for some constant $c>0$ (possibly depending on $\eps$) and w.h.p. with respect to randomness of $H$ as $N\to\infty$, there exist non-empty subsets $\cC_1,\ldots, \cC_K \subseteq \bSig_N$ called clusters such that:
\begin{enumerate}[label=\sf{S}\arabic*]
    \item 
    \label{it:each-cluster-small}
    (No dominant cluster)
    $\max_i \mu_{\beta,H}(\cC_i)\leq e^{-cN}$.
    \item 
    \label{it:diameter-bound}
    (Small cluster diameter)
    Each cluster has diameter $\diam(\cC_i)\leq 2\sqrt{\eps N}$ with respect to Euclidean distance on $\{+1,-1\}^N$, or Hamming distance $\eps N$.
    \item 
    \label{it:separation-bound}
    (Well-separated clusters)
    For any $i\neq j$, the Euclidean distance between the nearest points in $\cC_i$ and $\cC_j$ satisfies $d(\cC_i,\cC_j)\geq 2\sqrt{\eps^3 N}$.
    In terms of Hamming distance, this means $d_H(\cC_i,\cC_j) \ge \eps^3 N$.
    \item 
    \label{it:cover-gibbs-measure}
    (Clusters carry all mass)
    $\mu_{\beta,H}\bigl(\bigcup_{i=1}^K \cC_i\bigr)\geq 1-e^{-cN}$.
\end{enumerate}
\end{definition}

\begin{remark}
In the above, we use the correspondence between Euclidean and Hamming distances for $\{+1,-1\}^N$ vectors: $\|\zv - \zv'\|_2^2 = 4\, d_H(\zv,\zv')$, so $\|\zv - \zv'\|_2 = 2\sqrt{d_H(\zv,\zv')}$.
The separation bound \ref{it:separation-bound} in Euclidean distance $2\sqrt{\eps^3 N}$ thus corresponds to Hamming distance $\eps^3 N$.
Note also that properties \ref{it:each-cluster-small} and \ref{it:cover-gibbs-measure} together force the number of clusters $K$ to be exponentially large in $N$.
\end{remark}

As foreshadowed in \cref{thm:shatter-hardness}, we will show that the random XORSAT problem has a shattered Gibbs measure in a certain temperature range.

\begin{theorem}[Shattering in Max-XORSAT]
\label{thm:shattering}
Consider $H\sim\GG_{\rm ER}(N,k,D)$. For any $\eps>0$ and $(k,D,\beta)$ satisfying the hypotheses of \cref{thm:main-hardness},
the Gibbs measure $\mu_{\beta,H}$ is $(\eps/k)$-shattered.     
\end{theorem}

The onset of shattering can then be used to show that the Gibbs measure of $H$ is extremely sensitive to changes in the clauses, a phenomenon known as \emph{disorder chaos}.

\begin{theorem}[Disorder Chaos]
\label{thm:single-clause-disorder-chaos}
Fix any $\eps$ and $(k,D,\beta)$ satisfying the hypotheses of \cref{thm:main-hardness}.
Let $H,H'$ be sampled to share $M-1$ clauses (i.e. $f_i=f_i'$ for $1\leq i\leq M-1$) and have the last clauses $f_M,f_M'$ chosen independently, so that $H,H'$ are coupled with the same Erd\H{o}s-R\'enyi distribution $\GG_\ER(N,k,D)$.
Then 
\[
\liminf_{N\to\infty} \E W_1(\mu_{\beta,H},\mu_{\beta,H'})>0.
\]
\end{theorem}
We defer the proofs of \cref{thm:shattering} and \cref{thm:single-clause-disorder-chaos} to \cref{sec:shattering-main-proofs}.

Disorder chaos leads to the failure of stable algorithms, since it implies the true Gibbs measure must change substantially upon perturbation whereas the output distributions of stable algorithms change slowly by definition.
This yields the following proof of our main hardness result against stable algorithms:
\begin{proof}[Proof of Theorem~\ref{thm:main-hardness}]
Consider $H,H'$ with a single resampled clause.
By Theorem~\ref{thm:single-clause-disorder-chaos}, there exists a constant $\iota_* = \iota_*(k,D,\beta,\eps)>0$ such that
\[
\liminf_{N\to\infty}\E W_1(\mu_{\beta,H},\mu_{\beta,H'})\geq 3\iota_*
\]
Now let $\cA$ be an $\iota$-stable algorithm with $\iota \le \iota_*$, meaning $\E [W_1(\cA(H), \cA(H'))] \le \iota$.
By the triangle inequality for $W_1$:
\[
W_1(\mu_{\beta,H}, \cA(H))
+
W_1(\mu_{\beta,H'}, \cA(H'))
\geq 
W_1(\mu_{\beta,H},\mu_{\beta,H'})
-
W_1(\cA(H), \cA(H')).
\]
The right-hand side has $\liminf$ at least $3\iota_* - \iota\ge 2\iota_*$ in expectation, while the left-hand side consists of two identically distributed random variables, so we deduce that 
\[
\liminf_{N\to\infty}
\E W_1(\mu_{\beta,H}, \cA(H))
\geq \frac{3\iota_*-\iota}{2} \ge \iota_*>0
\]
as desired.
\end{proof}

\subsection{Obstructions to Concrete Algorithms: Low-Degree Polynomials and Glauber Dynamics}
\label{subsec:concrete-hardness}

Here we prove stability of sampling algorithms based on low-degree polynomials, and then use shattering directly to prove slow mixing of Glauber dynamics.

For this discussion, we view $H=(B,J)$ as an ordered list of $M$ independent clauses. Let $(\Omega, \PP_\omega)$ be a probability space, where $\omega$ represents randomness fed as input to the algorithm independent of $H$. 
We say a map $\mathbf X^\circ(H,\omega)\in\RR^N$ is a degree-$d$ clause polynomial if, for every fixed $\omega \in \Omega$, it has an orthogonal decomposition
\begin{equation}
\label{eq:clause-polynomial-decomposition}
\mathbf X^\circ(H,\omega)=\sum_{S\subseteq[M],\,|S|\le d}\mathbf X_S^\circ((f_i)_{i\in S},\omega),
\end{equation}
where each nonconstant term has conditional mean zero in each clause on which it depends.
This definition does not require a specific numerical encoding of a hyperedge. We assume the map is unchanged by a permutation of the clauses, since it must be when its input is the Hamiltonian $H$ rather than an ordered list. (If $\mathbf X$ isn't permutation-invariant, then we can symmetrize it via a random clause permutations, which cannot increase the expected sampling error as the clauses are i.i.d.)
We call the polynomial $C$-regular if
\[
\EE^{H,\omega}\|\mathbf X^\circ(H,\omega)\|_2^2\leq CN.
\]

Let $\vec U=(U_1,\dots,U_N)\stackrel{\mathrm{i.i.d.}}{\sim}\Unif([-1,1])$ be independent of $(H,\omega)$, and define the bitstring-valued output
\begin{equation}
\label{eq:round-stable-alg}
\mathbf X(H,\omega,\vec U)
:=
\round_{\vec U}(\mathbf X^{\circ}(H,\omega))\in\bSig_N.
\end{equation}
Here we define the coordinate-wise function $\round_{\vec U}(\xv)=\big(\round_{U_1}(x_1),\dots,\round_{U_N}(x_N)\big)$ for
\[
\round_U(x)=
\begin{cases}
1,&x\ge U,\\
-1,&x<U.
\end{cases}
\]
The associated distribution-valued algorithm is $\cA(H):=\Law\big(\mathbf X(H,\omega,\vec U)\mid H\big).$

\begin{proposition}
\label{prop:LDP-stable}
The distribution-valued algorithm $\cA$ above satisfies
\begin{equation}
\label{eq:ldp-stability-bound}
\EE W_1(\cA(H),\cA(H'))\leq \sqrt{\frac{2Cd}{M}},
\end{equation}
where $H'$ is obtained by resampling one clause as in \cref{def:stability}. In particular, it is $\iota$-stable whenever $d\leq \iota^2M/(2C)$.
\end{proposition}

\begin{proof}
For $r\in[M]$, let $H^{(r)}$ be obtained by replacing $f_r$ by an independent copy, while keeping $\omega$ fixed. Orthogonality of the terms in \eqref{eq:clause-polynomial-decomposition} gives
\[
\sum_{r=1}^M
\EE\big\|\mathbf X^\circ(H,\omega)-\mathbf X^\circ(H^{(r)},\omega)\big\|_2^2
=2\sum_S |S|\,\EE\|\mathbf X_S^\circ\|_2^2
\le 2d\,\EE\|\mathbf X^\circ\|_2^2
\le 2dCN.
\]
Clause-permutation invariance makes the $M$ summands equal. Hence, for the pair $(H,H')$ in \cref{def:stability},
\begin{equation}
\label{eq:single-clause-L2-influence}
\EE\big\|\mathbf X^\circ(H,\omega)-\mathbf X^\circ(H',\omega)\big\|_2^2
\le \frac{2dCN}{M}.
\end{equation}

Couple the two rounded outputs by using the same $(\omega,\vec U)$. For any $x,y\in\RR$ and $U\sim\Unif[-1,1]$,
\[
\EE_U\big|\round_U(x)-\round_U(y)\big|
=\big|[-1,1]\cap[\min(x,y),\max(x,y)]\big|
\le |x-y|.
\]
The definition of $W_1$, Cauchy--Schwarz, and \eqref{eq:single-clause-L2-influence} now give
\begin{align*}
\EE W_1(\cA(H),\cA(H'))
&\le \frac1N\EE\big\|\mathbf X(H,\omega,\vec U)-\mathbf X(H',\omega,\vec U)\big\|_1\\
&\le \frac1N\EE\big\|\mathbf X^\circ(H,\omega)-\mathbf X^\circ(H',\omega)\big\|_1\\
&\le \sqrt{\frac1N\EE\big\|\mathbf X^\circ(H,\omega)-\mathbf X^\circ(H',\omega)\big\|_2^2}
\le \sqrt{\frac{2Cd}{M}}.
\qedhere
\end{align*}
\end{proof}

\begin{corollary}
\label{cor:SLDH}
Fix parameters $(k, D, \beta)$ as in \cref{thm:main-hardness}. Let $d_N=o(N)$ and let $\mathbf X_N^\circ$ be degree-$d_N$ clause polynomial that is invariant under clause permutations and $C$-regular with $C=O(1)$. Then its rounded distribution-valued algorithm $\cA_N$ is not an $o_N(1)$-approximate sampler. 
More precisely, for some $\eta=\eta(k,D,\beta)>0$ and all sufficiently large $N$,
\[
\EE_{H\sim\GG_\ER(N,k,D)}
W_1\big(\mu_{\beta,H},\cA_N(H)\big)\ge \eta.
\]
\end{corollary}

\begin{proof}
Because $M=\lfloor ND/k\rfloor=\Theta(N)$, \cref{prop:LDP-stable} gives
$\EE W_1(\cA_N(H),\cA_N(H'))=o_N(1)$. For all sufficiently large $N$, this is at most the constant $\iota_*$ in \cref{thm:main-hardness}, which then supplies the stated positive lower bound.
\end{proof}

To understand the implications of this result, we recall the  \emph{low-degree ansatz} of \cite{hopkins2018statistical,kunisky2019notes}.
This heuristic states that degree $d$ polynomial algorithms are often as powerful as all algorithms with running time at most $e^{\tilde O(d)}$. 
This ansatz has been disproved via counterexamples based in coding theory \cite{holmgren_wein,buhai2025quasi,mao2026polynomial} yet has made a number of highly plausible concrete predictions.
It naively suggests via Corollary~\ref{cor:SLDH} that sampling from $\mu_{\beta}$ requires $e^{\tilde\Omega(N)}$ time within the shattered phase.
However, as we show in this paper, DQI and Prange's algorithm only require $\poly(N)$ time.

The onset of shattering also directly creates barriers for Markov chain Monte Carlo methods. Recall that the Glauber dynamics for $\mu_{\beta, H}$ is the Markov chain on $\bSig_N$ that, at each step, selects a uniformly random coordinate $i \in [N]$ and resamples $z_i$ from the conditional distribution $\mu_{\beta}(\cdot \mid z_{-i})$.
This chain is reversible with respect to $\mu_{\beta}$.

\begin{theorem}
\label{thm:slow-mixing}
Fix any $\eps \in (0,\frac12)$. Suppose $\mu_{\beta,H}$ is $\eps$-shattered as in \cref{def:shattering} with high probability over $H \sim \GG_\ER(N,k,D)$.
Then the Glauber dynamics chain for $\mu_{\beta,H}$ has mixing time $e^{\Omega(N)}$ with high probability.
\end{theorem}

\begin{proof}
For any Gibbs measure exhibiting the shattering property as in \cref{def:shattering}, we can refine it so that all clusters have an additional \emph{bottleneck} property.
Letting $\cN_r(\cC)$ be the set of points within Hamming distance $r$ of a set $\cC$ which are not themselves in $\cC$, the bottleneck property requires that
\[
\mu_{\beta}(\cC_i)\geq 
e^{cN/10}
\mu_{\beta}\big(\cN_{r_b}(\cC_i)\big),
\]
for $r_b=\eps^3 N/3$.
The refinement simply removes all clusters where the above inequality does not hold.
Indeed suppose WLOG that the removed ones are the first $J\leq K$ clusters, then
\[
\mu_{\beta}(\cup_{j\leq J} \cC_j)
\leq 
e^{cN/10}
\mu_{\beta}\lt(\bSig_N\backslash\cup_{\ell \leq K}\cC_\ell \rt)
\leq 
e^{-cN/10}.
\]
Since only an exponentially small amount of Gibbs mass is removed this way, it follows that the refined clustering obeys all shattering properties from before in addition to the bottleneck property.
Now slow mixing is immediate, see e.g., \cite[Corollary 2.8, Corollary 2.9]{alaoui2023shattering}.
\end{proof}

Note that Theorems~\ref{thm:main-hardness}, \ref{thm:shattering}, and \ref{thm:slow-mixing} together imply \cref{thm:shatter-hardness}.

\section{Shattering and Disorder Chaos of Max-$k$-XORSAT}
\label{sec:shattering-main-proofs}

The goal of this section is to prove Theorems~\ref{thm:shattering} and \ref{thm:single-clause-disorder-chaos}.
The first one, \cref{thm:shattering}, shows that in the temperature range \eqref{eq:beta-k-def}, the Gibbs measure $\mu_{\beta,H}$ shatters into exponentially many well-separated clusters.
\cref{thm:single-clause-disorder-chaos} then deduces disorder chaos from shattering.

\paragraph{Parameter Hierarchy.} We will work with several parameters, chosen such that
\begin{equation}
\label{eq:eps-to-zero-slowly}
    0<1/N\ll c\ll \delta\ll 1/D \le 1/k\ll \eps\ll 1.
\end{equation}
Here $a\ll b$ means $a$ is chosen sufficiently small depending on the parameters already fixed. (In particular, this is stronger than merely requiring that $a/b$ tends to zero.)
We first choose $\eps$ sufficiently small, then $k\ge k_0(\eps)$ sufficiently large, then fix any $D\ge k$ and an inverse temperature $\beta$ in the range  \eqref{eq:beta-k-def}. Then we choose $\delta$ sufficiently small depending on $(\eps,k,D)$, then $c$, and finally consider the $N\to\infty$ limit with other parameters fixed.
For clarity, we make the error bounds used below explicit by choosing
\begin{equation}
\label{eq:cond-eps-k-D-delta}
\frac{\ln(e/\eps)}{\ln k } \ll \eps,
\qquad
\sqrt{\frac{\ln k}{k}} \ll \eps, 
\quad \text{ and } \quad
\delta  \ll \min\bigg\{\Big(\frac{\eps}{k}\Big)^2, \eps^2 \sqrt{\frac{\ln k}{D}}\bigg\}.
\end{equation}
Note this is consistent with \eqref{eq:eps-to-zero-slowly} as the first two requirements concern choosing $k$ large given $\eps$, and the last requirement applies when choosing $\delta$ after $(\eps,k,D)$.
We consider inverse-temperatures within the range
\begin{equation}
\label{eq:beta-k-def}
(1+\eps)\sqrt{\frac{2\ln k}{D}}
\leq \beta
\le \sqrt{\frac{k}{2D}}
\le 
\beta_\RS(k,D).
\end{equation}
The threshold $\beta_\RS(k,D)$ marks the onset of replica-symmetry breaking, so that \eqref{eq:beta-k-def} is contained within the so-called replica-symmetric phase.  We show $\beta_\RS(k, D) \ge \sqrt{k/(2D)}$ in \cref{lem:annealed-FE}.

We note that the value $\sqrt{\frac{2\ln k}{k}}\cdot (1+o_{k\to\infty}(1))$ is the predicted onset of shattering for the Ising $k$-spin glasses~\cite{alaoui2023sampling}. In a sense, the Ising $k$-spin glasses describe the $D\to\infty$ limit of Max-$k$-XORSAT on random degree-$D$ hypergraphs, provided one rescales the Hamiltonian via $\tilde{H}=(\sqrt{k/D})H$.
The existence of shattering of the Ising $k$-spin glasses above this threshold was shown recently by \cite{alaoui2024near} and implies hardness of sampling for stable algorithms.
While suggestive, \cite{alaoui2024near} does not have direct implications for Max-$k$-XORSAT at moderate $D$.
Relatedly, \cite{huang2025hardness} proved a hardness result for $k=2$, using an interpolation argument that translates properties of dense Ising spin glasses to sparse random graphs for $\beta >\beta_\RS$.
As far as we can tell, their technique does not work for the regime $\beta<\beta_\RS$ considered in this paper.

In this paper, we directly show that a similar shattering property holds for general Max-$k$-XORSAT within the range \eqref{eq:beta-k-def} and obstructs a large class of sampling algorithms.

\subsection{Replica-Symmetric Free Energy}
\label{sec:RS}

Our first step is to show the typical value of the free energy ($\ln Z_{\beta,H}$) approaches its ``annealed'' value ($\ln \EE Z_{\beta,H}$) in the $N\to\infty$ limit in certain temperature regime. This is the so-called replica-symmetric phase, where many calculations can be simplified.

\begin{definition}[Replica-symmetric threshold]
\label{def:beta-RS}
For any sequence of ensembles $\GG=(\GG_N)_{N\ge 1}$ of Hamiltonians $H$, the \emph{replica-symmetric threshold} $\beta_\RS(\GG)$ is the maximum value such that for all $0\le \beta \le \beta_\RS(\GG)$, 
\begin{equation} \label{eq:annealed-FE}
    \lim_{N\to\infty} \frac{1}{N}\big|\EE \ln Z_{\beta,H} - \ln \EE Z_{\beta,H}\big|=0.
\end{equation}
\end{definition}
For $\GG_N = \GG_\ER(N,k,D)$ with fixed $(k,D)$, we abbreviate this threshold as $\beta_\RS(k,D)$.

\begin{lemma} \label{lem:annealed-FE}
For the ensemble $\GG_\ER(N,k,D)$ at any $k\ge 2$, we have 
$\beta_{\RS}(k,D) \geq \sqrt{k/(2D)}$. In other words, \eqref{eq:annealed-FE} holds for all $0\le \beta \leq \sqrt{k/(2D)}$.
\end{lemma}

The strategy to prove this lemma is the ``concentration-enhanced second moment method''.
Indeed it is known since \cite{talagrand1998sherrington} that to show results such as \cref{eq:annealed-FE}, it suffices to prove that for $N$ large,
\begin{equation} \label{eq:2nd-moment-match}
\frac{\EE [Z_{\beta,H}^2]}{(\EE [Z_{\beta,H}])^2}\leq e^{o(N)}
\end{equation}
thanks to concentration of the free energy  and the Paley--Zygmund inequality; see also \cite{frieze1990independence}.

We begin by computing the first moment. Since each $f_i(\zv)$ is an independent Rademacher random variable for any fixed $\zv$, we have by direct calculation that 
\begin{align}
\E [Z_{\beta,H}]
&=
\sum_{\zv} \E \prod_{i=1}^M e^{\beta f_i(\zv)}
= \sum_\zv (\cosh\beta)^M 
= 2^N (\cosh\beta)^M, \nonumber \\
(\E[Z_{\beta,H}])^2
&=
4^N (\cosh\beta)^{2M}.
\end{align}

To obtain the second moment, we decompose the sum
\begin{equation}
\label{eq:2nd-moment}
\EE [Z_{\beta,H}^2]=
\EE \sum_{\zv,\zv'} e^{-\beta(H(\zv)+H(\zv'))}
\end{equation}
by the overlap $R(\zv,\zv') := \zv^\top \zv'/N$.

Consider the contribution to Eq.~\eqref{eq:2nd-moment} from $\zv,\zv' \in \bSig_N$ which have overlap $R=R(\zv, \zv')\in [-1,1]$, so that the two configurations disagree on exactly $\frac{1-R}{2}N$ coordinates. A given clause takes the same truth value on $\zv$ and $\zv'$ if and only if its hyperedge contains an even number of disagreeing coordinates. If the $k$ vertices of a hyperedge were sampled i.i.d.\ uniformly from $V$, each would independently land on a disagreeing coordinate with probability $\frac{1-R}{2}$, and the probability that the clause agrees would be
\[
\sum_{\ell \text{ even} } \binom{k}{\ell} \Big(\frac{1-R}{2}\Big)^\ell \Big(\frac{1+R}{2}\Big)^{k-\ell}=
\frac{1+R^k}{2}.
\]
For the $\GG_\ER(N,k,D)$ ensemble in \cref{def:erdos_renyi}, each hyperedge is instead a uniformly random $k$-subset of $V$ (the vertices are sampled \emph{without} replacement), so the number of disagreeing coordinates in the edge is hypergeometric rather than binomial. The two sampling schemes can be coupled by drawing $k$ i.i.d.\ vertices and redrawing in the event of a collision, which has probability at most $\binom{k}{2}/N$; hence the exact agreement probability satisfies
\begin{equation}
\label{eq:pk-hypergeometric}
p_k(R)
=\frac{1}{\binom Nk}
\sum_{\ell\text{ even}}
\binom{(1-R)N/2}{\ell}\binom{(1+R)N/2}{k-\ell}
= \frac{1+R^k}{2} + O\Big(\frac{k^2}{N}\Big),
\qquad\text{uniformly in } R.
\end{equation}
For such $(\zv,\zv')$, each clause contributes
\[
\mathbb E[e^{\beta(f_i(\zv)+f_i(\zv'))}]
=
p_k \cosh(2\beta)
+
(1-p_k)
=
1+ p_k \cdot (\cosh(2\beta)-1).
\]

Note that the probability for a uniformly random pair $(\zv,\zv')$ of $\{\pm1\}$-vectors to have overlap $R$ is $e^{N(h(R)\pm o(1))}$, where 
\[
h(R)=
-\Big(\frac{1+R}{2}\Big)\ln(1+R)
-
\Big(\frac{1-R}{2}\Big)\ln(1-R)
\]
is the rescaled entropy on $(-1,1)$.
Then
\[
\EE[Z_{\beta,H}^2]
= 4^N \exp\lt(
N
\max_{-1\leq R\leq 1}
\lt[
\frac{D}{k}
\log\Big(1 + p_k(R) \cdot (\cosh(2\beta)-1)\Big)
+ h(R) +o(1)
\rt]
\rt).
\]
Then the desired condition \eqref{eq:2nd-moment-match} holds provided the variational expression inside the bracket above is uniquely maximized at $R=0$.
Indeed, this quantity at $R=0$ equals $(2D/k)\log\cosh\beta$, and thus the right-hand-side matches $(\E[Z_{\beta,H}])^2$ up to $e^{o(N)}$. 
In other words, it is enough to have
\begin{equation}
\label{eq:2nd-moment-condition}
\argmax_{-1\leq R\leq 1}
\lt(
\frac{D}{k}
\log\lt[
1+\Big(\frac{1+R^k}{2}\Big)
\cdot 
\big(\cosh(2\beta)-1\big)
\rt]
+
h(R)
\rt)
=0.
\end{equation}

\begin{lemma}
\label{lem:check-2nd-moment-condition}
For any $(k,D)$ where $k\ge 2$, the condition \eqref{eq:2nd-moment-condition} holds if $\beta^2\leq k/(2D)$.
\end{lemma}

\begin{proof}
    Note that $h(R)\leq -R^2/2$ for all $R\in (-1,1)$, since $h(0)=h'(0)=0$ and $h''(R)=\frac{-1}{1-R^2}\leq -1$ for all $R$.
    Also,
    \[
    1+\Big(\frac{1+R^k}{2}\Big)
    \cdot 
    \big(\cosh(2\beta)-1\big)
    =
    \cosh^2(\beta)
    +
    R^k \sinh^2(\beta)
    \]
    Thus it suffices to show that
    \[
    g(R)
    \equiv 
    \frac{D}{k}
    \log\lt(1 +  R^k \tanh^2(\beta)
    \rt)
    -
    \frac{R^2}{2}
    \]
    is maximized at $R=0$.
    Since $\log(1+x)\leq x$ we have 
    \[
    g(R)
    \leq 
    \frac{D}{k}
    \cdot R^k \tanh^2(\beta)
    - \frac{R^2}{2},
    \]
    again with equality at $R=0$.
    Finally $R^k\leq R^2$ for all $k\ge 2$ and $R\in [-1,1]$ with equality at $0$.
    Tracing back, the condition \eqref{eq:2nd-moment-condition} holds as long as
    \begin{equation}
    \frac{D}{k}
    \tanh^2(\beta)
    \le 1/2.
    \end{equation}
Since $\beta^2 \ge \tanh^2(\beta)$, it also suffices to have the simpler inequality $\beta^2 \le \frac{k}{2D}$ as claimed.
\end{proof}

\begin{proof}[Proof of \cref{lem:annealed-FE}]
We now complete the proof. 
Due to \cref{lem:check-2nd-moment-condition}, $\beta^2 \le k/(2D)$ is a sufficient condition for \eqref{eq:2nd-moment-condition}, which in turn implies \eqref{eq:2nd-moment-match}. The latter, together with Paley--Zygmund inequality yields
\[
\PP[Z_{\beta,H} \ge \EE[Z_{\beta,H}]/2] \ge \frac14 \frac{(\E Z_{\beta,H})^2}{\E[Z_{\beta,H}^2]} \ge \frac14 e^{-o(N)}
\]
Note the leftmost probability is equal to $\PP[\ln Z_{\beta,H} \ge \ln \E Z_{\beta,H} - \ln 2]$.

We also need concentration of the free energy $F_{\beta,H}=\frac{1}{N} \ln Z_{\beta,H}$. Since changing a single clause will change $\ln Z_{\beta,H}$ by at most $2\beta$, we have from  McDiarmid's inequality 
\[
\PP\Big[\Big|\ln Z_{\beta,H} - \E[\ln Z_{\beta,H}]\Big| > t N\Big] 
\le 2\exp\!\Big({-}\frac{2t^2N^2}{M(2\beta)^2}\Big)
\le 2\exp\!\Big({-}\frac{k t^2N }{2D\beta^2}\Big),
\]
where the last inequality uses $M=\lfloor ND/k\rfloor \le ND/k$.

Applying Jensen's inequality to the logarithm, we have $\EE \ln Z_{\beta, H} \le \ln \EE Z_{\beta,H}$. 
Suppose for the sake of contradiction against \eqref{eq:annealed-FE} that for some $t>0$,  $\EE \ln Z_{\beta, H} \le \ln \EE Z_{\beta,H} - 2t N$ for arbitrarily large $N$.
Then
\begin{align*}
    2\exp\!\Big({-}\frac{k t^2N }{2D\beta^2}\Big) &\ge \PP\Big[\Big|\ln Z_{\beta,H} - \E[\ln Z_{\beta,H}]\Big| > t N\Big] \\
    &\ge \PP[\ln Z_{\beta,H} \ge \ln \EE [Z_{\beta,H}] - \ln 2] \\
    &\ge \frac14 e^{-o(N)},
\end{align*}
where the second inequality follows from observing that the event in the first line is implied by the event in the second line under the supposition.
For sufficiently large $N$, $2\exp\!\Big({-}\frac{k t^2N }{2D\beta^2}\Big) \not\ge\frac14 e^{-o(N)}$, yielding a contradiction as desired.
\end{proof}

We next observe the following corollary which will be applied later this section.

\begin{corollary}
\label{cor:most-pairs-orthogonal}
Suppose $\beta^2 \le k/(2D)$.
For any $\delta>0$, there exists $c>0$ such that with probability at least $1-e^{-cN}$ over $H\sim\GG_\ER(N,k,D)$, we have
\[
\mu_{\beta,H}^{\otimes 2}\!\left[|\zv^\top \zv'|/N\geq \delta\right]\leq e^{-cN}.
\]
\end{corollary}

\begin{proof}
Let
\[
A_H(\delta):=\sum_{\substack{\zv,\zv'\in \bSig_N\\ | \zv^ \top\zv'|/N\ge \delta}} e^{-\beta(H(\zv)+H(\zv'))}.
\]
Since the variational problem in \eqref{eq:2nd-moment-condition} is uniquely maximized at $R=0$, there exists $c_1=c_1(\delta,\beta,k,D)>0$ such that
\[
\E[A_H(\delta)]\le e^{-c_1N}(\E Z_{\beta,H})^2.
\]
By \cref{lem:annealed-FE} and McDiarmid's inequality, there exists $c_2>0$ such that with probability at least $1-e^{-c_2N}$ we have
\[
Z_{\beta,H}\ge e^{-c_1N/4}\E Z_{\beta,H}.
\]
On this event,
\[
\mu_{\beta,H}^{\otimes 2}\!\left[\frac{|\la \zv,\zv' \ra|}{N}\ge \delta\right]
=
\frac{A_H(\delta)}{Z_{\beta,H}^2}
\le
e^{c_1N/2}\frac{A_H(\delta)}{(\E Z_{\beta,H})^2}.
\]
Markov's inequality and the bound on $\E[A_H(\delta)]$ therefore imply
\[
\mathbb P\!\left[\mu_{\beta,H}^{\otimes 2}\!\left[\frac{|\la \zv,\zv' \ra|}{N}\ge \delta\right]\ge e^{-c_1N/4}\right]\le e^{-cN}
\]
for some $c>0$, which yields the claimed estimate after adjusting constants.
\end{proof}

\subsection{Planting and Exponential-Scale Contiguity}
We next introduce the planted trick, which is a proof technique to show shattering. This technique was used in several prior works including \cite{achlioptas2008algorithmic,alaoui2023sampling,alaoui2023shattering, alaoui2024near}.
The idea is to consider two distributions over pairs $(H, \zv^*)$: one that first draws $H$ and then samples $\zv^*$ from the Gibbs measure (the ``null model''), and one that first draws $\zv^*$ and then constructs $H$ biased towards $\zv^*$ (the ``planted model'').
Under the condition \eqref{eq:annealed-FE} where the annealed free energy is a good approximation, these models are ``exponentially contiguous,'' 
such that events with exponentially small probability under the planted model also have exponentially small probability under the null model.
This trick allows us to transfer desired properties in the planted model, which is easier to study, to the null model.

\begin{definition}[Null and planted models]
The \emph{null model} $\QQ$ over $(H,\zv^*)$ is defined by:
\begin{enumerate}
\item Draw $H = (G, J) \sim \GG_\ER(N,k,D)$.
\item Draw $\zv^*\sim\mu_{\beta,H}$.
\end{enumerate}
\noindent
The \emph{planted model} $\PP$ over $(H,\zv^*)$ is defined by:
\begin{enumerate}
\item Draw $\zv^*\in \{-1,1\}^N$ uniformly at random.
\item Draw a random $(G, J^0) \sim \GG_\ER(N,k,D)$, and keep only the hypergraph $G$.
\item For each hyperedge $e\in E$, independently set
\[
J_e = 
\begin{cases}
z^*_{e_1} z^*_{e_2} \cdots z^*_{e_k} & \text{with probability } \frac{e^{\beta}}{2\cosh(\beta)},\\[5pt]
-z^*_{e_1} z^*_{e_2} \cdots z^*_{e_k} & \text{with probability } \frac{e^{-\beta}}{2\cosh(\beta)}.
\end{cases}
\]
Equivalently, each clause $f_i$ satisfies $\PP[f_i(\zv^*)=1 \mid \zv^*] = e^{\beta}/(2\cosh\beta)$.
\end{enumerate}
\end{definition}

The Radon--Nikodym derivative between these distributions is given exactly by 
\[
\frac{d\PP}{d\QQ}(H, \zv^*)
=
\frac{Z_{\beta,H}}{\EE Z_{\beta,H}}.
\]
We say a property has \textbf{exponentially good probability} if there exists a constant $c>0$, such that the property holds with probability at least $1-e^{-cN}$ for all sufficiently large $N$.
The main transference property we use is as follows.

\begin{proposition}
\label{prop:contiguity}
Suppose $\beta\in [0,\beta_\RS(k, D)]$. 
If a property of $(H,\zv^*)$ has exponentially good probability under the planted model $\PP$, then the same holds under the null model $\QQ$. 
\end{proposition}

\begin{proof}
    This can be shown using the same ``contiguity at exponential scale'' argument as in \cite[Section 3.1]{alaoui2023shattering}, which we adopt for the present ensemble as follows.
Let $A$ be the event that the property fails, and suppose $\PP(A)\le e^{-aN}$. The displayed Radon--Nikodym derivative gives
\[
\frac{d\QQ}{d\PP}(H,\zv^*)=\frac{\EE Z_{\beta,H}}{Z_{\beta,H}}.
\]
Fix any $b\in(0,a)$.  Since \eqref{eq:annealed-FE} is implied by $\beta\in [0,\beta_\RS(k,D)]$, for all large $N$ we have
$\EE\ln Z_{\beta,H}\ge \ln\EE Z_{\beta,H}-bN/2$. McDiarmid's inequality, applied exactly as in the proof of \cref{lem:annealed-FE}, then gives a constant $t_b>0$ such that under the null model
\[
\QQ\left[Z_{\beta,H}<e^{-bN}\EE Z_{\beta,H}\right]\le e^{-t_bN}.
\]
On the complementary event, $d\QQ/d\PP\le e^{bN}$. Therefore
\[
\QQ(A)\le e^{bN}\PP(A)+e^{-t_bN}
\le e^{-(a-b)N}+e^{-t_bN},
\]
which proves the proposition.
\end{proof}

\subsection{A Soft OGP Implies Shattering and Disorder Chaos}
\label{sec:OGP-shatter-disorder}

\subsubsection{Soft OGP}

To prove shattering and disorder chaos, we adapt the notion of a \emph{soft overlap gap property} (soft OGP) from \cite{alaoui2024near}. Soft OGP asserts that around each typical Gibbs sample $\zv^*$, there is an annular ``gap'' in Hamming distance where no configuration achieves energy close to that of $\zv^*$. This can be used to obtain an explicit decomposition of the hypercube into separated clusters, such that almost all of $\mu_{\beta,H}$ is supported within this set. This will used to show that $\mu_{\beta,H}$ is shattered.

Let us denote the expected value of the per-clause energy 
\begin{equation}
\Gamma := -\tanh(\beta) = \E_{(H,\zv^*)\sim \PP}\lt[\frac{1}{M}H(\zv^*)\rt]
\end{equation}
 of a Gibbs sample $\zv^*$ in the planted model.
\begin{definition}[Soft OGP]
\label{def:soft-OGP}
Fix any $\eps >0$, and $0< \delta \ll (\eps/k)^2$ independent of $N$.
We say $(H,\zv^*)$ obeys the $(\eps,\delta)$-\emph{soft OGP} if
\begin{enumerate}[label=(\roman*)]
    \item 
    \label{it:energy-typical}
    $|H(\zv^*)-M\Gamma|\leq M\delta $.
    \item 
    \label{it:OGP-soft}
    For all $\zv$ satisfying
    \[
    \left|\frac{d_H(\zv,\zv^*)}{N}  - \frac{\eps}{k} \right|\le  \lt(\frac{\eps}{k}\rt)^2
    \]
    we have 
    \[
    H(\zv) - M \Gamma \ge 2 M \delta .
    \]
\end{enumerate}
\end{definition}
In other words, property~\ref{it:energy-typical} says $\zv^*$ has energy close to the expected value in the planted model. Property~\ref{it:OGP-soft} says all configurations at Hamming distance $\eps N/k$ from $\zv^*$ have strictly higher energy.

We will prove that soft OGP applies to both the planted and the null model. 
Recall from \eqref{eq:f-i-def} that $H=-\sum_{i=1}^M f_i$, and each $f_i$ is the $\pm 1$-valued indicator that clause $i$ is satisfied.
As preparation, we first show the following:

\begin{lemma}
\label{prop:pairs-of-terms}
In the planted model, condition on $\zv^*$ and choose $\zv$ with  Hamming distance $d_H(\zv,\zv^*) =\eps_* N/k$, such that $H$ and $\zv$ are conditionally independent given $\zv^*$.
Then the law of $(f_1(\zv^*),f_1(\zv))$ is given by $\PP[(f_1(\zv^*),f_1(\zv))=(a,b)]=p(a,b)$ where
\begin{align*}
    p(+1,+1)
    &=
    (1-\tep)
    \cdot \frac{e^{\beta}}{2\cosh(\beta)},
    &\qquad
    p(+1,-1)
    &=
    \tep
    \cdot \frac{e^{\beta}}{2\cosh(\beta)},
    \\
    p(-1,+1)
    &= 
    \tep
    \cdot \frac{e^{-\beta}}{2\cosh(\beta)},
    &\qquad
    p(-1,-1)
    &= 
    (1-\tep)
    \cdot \frac{e^{-\beta}}{2\cosh(\beta)},
\end{align*}
where
\begin{equation}
\label{eq:tilde-eps-def}
\tep = \frac{1-(1-2\eps_*/k)^k}{2} + O\lt(\frac{k^2}{N}\rt)= \eps_* - \Big(1-\frac{1}{k}\Big) \eps_*^2 + O\lt(\eps_*^3 + \frac{k^2}{N}\rt).
\end{equation}
\end{lemma}

\begin{proof} Here, we consider only the planted model $\PP$.
Note that given $\zv$ and $\zv^*$ at Hamming distance $\eps_* N/k$, the clause agrees as long as the two strings differ by an even number of bits involved in the clause.
Hence,
\[
\PP\big[f_1(\zv^*)=f_1(\zv)~|~(\zv,\zv^*)\big]
= \frac{1}{\binom{N}{k}} \sum_{\ell \text{ even}} \binom{\eps_*N/k}{\ell}\binom{N-\eps_*N/k}{k-\ell} =: 1-\tep.
\]
Following the same collision coupling scheme in Eq.~\eqref{eq:pk-hypergeometric}, we can approximate the distribution of the number of disagreeing coordinates as binomial instead of hypergeometric (by drawing the vertices in each clause hyperedge i.i.d.~with replacement), with error at most $O(k^2/N)$.
This yields
\begin{align*}
\tep
&= 1-\sum_{\ell \text{ even}} \binom{k}{\ell} \lt(\frac{\eps_*}{k}\rt)^\ell \lt(1-\frac{\eps_*}{k}\rt)^{k-\ell} + O(k^2/N) =
\frac{1-(1-2\eps_*/k)^k}{2} + O(k^2/N).
\end{align*}
Note that the event in question does not depend on the choice of weight $J_{e^{(1)}}$.
Due to this independence, we see that this event is independent of the event that $f_1(\zv^*)=1$, and so multiplying probabilities finishes the proof.
\end{proof}

\begin{proposition}
\label{prop:soft-OGP-planted}
For parameters as in \eqref{eq:eps-to-zero-slowly}, satisfying \eqref{eq:cond-eps-k-D-delta} and \eqref{eq:beta-k-def}, the $(\eps,\delta)$-soft OGP holds with exponentially good probability under the planted model $\PP$. 
Hence by Proposition~\ref{prop:contiguity}, the same holds under the null model $\QQ$.
\end{proposition}

\begin{proof}
Property~\ref{it:energy-typical} is immediate from the definition of $\PP$, and just amounts to exponential concentration for sums of i.i.d. Bernoulli variables via e.g., Hoeffding's inequality.

For property~\ref{it:OGP-soft}, we consider all $\zv$ with $d_H(\zv,\zv^*)=\eps_*N/k$ where $\eps_*$ varies in the set 
\begin{equation} \label{eq:eps-star-range}
\eps_*\in \Big[\eps-\frac{\eps^2}{k},\eps+\frac{\eps^2}{k}\Big]\cap (N/k)^{-1}\ZZ.
\end{equation}
We will estimate the maximum value of $H(\zv^*)-H(\zv)$ over such $\zv$ via a union bound.
The number of such points $\zv$ for any fixed $\eps_*$ is
\begin{align}
\binom{N}{\eps_* N/k} 
\le \exp\lt(\frac{N \eps_*}{k} \ln \frac{e k}{\eps_*}\rt)
\le 
\exp\lt(N \Big(1 + \frac{\eps}{2}\Big) \frac{\eps_* \ln(k)}{k}\rt),
\label{eq:entropy-of-nearby-points}
\end{align}
where the last inequality follows from \eqref{eq:cond-eps-k-D-delta} and $\eps_*=\eps+O(\eps^2/k)$.
For each such $\zv$, consider the energy difference
\begin{equation}
\label{eq:H-diff-sum}
H(\zv)-H(\zv^*)
=
\sum_{i=1}^M [f_i(\zv^*)-f_i(\zv)]
.
\end{equation}
\cref{prop:pairs-of-terms} shows that each summand term equals $0$ with probability $1-\tep$, and has expectation
\[
\EE^{\PP}[f_i(\zv^*)-f_i(\zv)]
=
2\tanh(\beta) \tep
\]
where $\tep=\eps_*\pm O(\eps_*^2 + k^2/N)$.
Thus 
\[
\EE^{\PP}[H(\zv)-H(\zv^*)] 
=
2M \tanh(\beta) \tep .
\]

Next we want to show $H(\zv) - H(\zv^*) \ge  3\delta M$ with high probability.
Along with property~\ref{it:energy-typical} which means that $H(\zv^*)\ge (\Gamma-\delta)M$, this would imply $H(\zv) \ge (\Gamma + 2\delta)M$ and prove property~\ref{it:OGP-soft}.
To do so we use Bennett's inequality, which allows us to leverage the reduced variance in each term (and is relatively sharp because the large deviation event we consider has small constant factors).
Note that 
\[
S(\zv)=H(\zv^*) - H(\zv) - 2\Gamma\tep M = \sum_{i=1}^M (f_i(\zv)-f_i(\zv^*)+2\tanh(\beta)\tep)
\]
is a sum of i.i.d. mean zero terms valued in $[-2-2\tep,2+2\tep]$ (since $|\tanh(\beta)|\leq 1$).
Furthermore, the variance of each term is at most $\mathbb E[(f_1(\zv)-f_1(\zv^*))^2]=4\tep$ by \cref{prop:pairs-of-terms}.
Therefore with $g(u)=(1+u)\ln(1+u)-u$, Bennett's inequality yields
\begin{align}
& \mathbb{P}[H(\zv^*) > H(\zv) - 3\delta M] = 
\mathbb P\Big[S(\zv) > \Big(2\tanh(\beta)\tep-3\delta\Big) M\Big] \nonumber\\
&\qquad 
\leq 
\exp\lt(
-\frac{4M\tep}{4(1+\tep)^2}
g\lt(\frac{2(1+\tep)(2\tanh(\beta)\tep-3\delta) M}{4M\tep}\rt)
\rt)\nonumber\\
&\qquad 
=
\exp\lt(
-\frac{M\tep}{(1+\tep)^2}
g\Big(
 (1+\tep)\Big[\tanh (\beta) - \frac{3\delta}{2\tep}\Big]
\Big)
\rt).
\label{eq:bennett-bound}
\end{align}
Let $\beta_0=(1+\eps)\sqrt{2\ln k/D}$. Note $\tanh\beta_0 > 3\delta/(2\tep)$ by \eqref{eq:cond-eps-k-D-delta}. The functions $g$ and $\tanh$ are increasing on $[0,\infty)$, so we may lower bound this exponent by replacing $\beta$ with $\beta_0$.
Since $\beta_0=O(\sqrt{\ln k / k})$ for $D \ge k$, and $g(u)=u^2/2+O(u^3)$ as $u\to0$, we have
\[
\frac{M\tep}{(1+\tep)^2} g\Big( (1+\tep)\Big[\tanh (\beta_0) - \frac{3\delta}{2\tep}\Big]\Big) = \frac{M\tep \beta_0^2}{2} \Big[1+O(\beta_0 + \frac{\delta}{\tep\beta_0})\Big]
\]
By condition \eqref{eq:cond-eps-k-D-delta}, both errors are small compared to $\eps$.
Furthermore, for $\eps_*$ in the allowed range \eqref{eq:eps-star-range}, \eqref{eq:tilde-eps-def} gives
\[
\frac{\tep}{\eps_*} = 1 - \Big(1-\frac{1}{k}\Big) \eps_* + O(\eps_*^2) + o_N(1) = 1-\eps + O(\eps^2) + o_N(1)
\]
Using $M=\lfloor N D/k\rfloor$, we conclude that uniformly for $\beta\ge\beta_0$ and over allowed $\eps_*$, the exponent in \eqref{eq:bennett-bound} satisfies
\begin{align}
\frac{M\tep}{(1+\tep)^2} 
g\Big( (1+\tep)\Big[\tanh (\beta) - \frac{3\delta}{2\tep}\Big]\Big)
&\ge 
N\frac{\eps_*\ln k}{k}(1+\eps)^2\Big[1-\eps - O\big(\eps^2 + \beta_0 + \frac{\delta}{\eps\beta_0}\big)-o_N(1)\Big] \nonumber \\
&
\ge N(1+0.99\eps)\frac{\eps_*\ln k}{k},
\label{eq:beta-bound}
\end{align}
for all sufficiently small $\eps$, followed by sufficiently large $k$ and then $N$.
This rate exceeds the entropy rate in \eqref{eq:entropy-of-nearby-points} by $\Omega(N\eps^2\ln k/k)$. To complete the proof, we apply a union bound over the configurations and the at most $N+1$ possible distances, which then shows property~\ref{it:OGP-soft} fails with probability at most $e^{-cN}$ for a sufficiently small $c>0$. 
\end{proof}

\subsubsection{From Soft OGP to Shattering (Proof of \cref{thm:shattering})}
\label{sec:softOGP-to-shattering}

Now we are ready to prove the shattering asserted in \cref{thm:shattering}.
To do this, we apply the soft OGP to obtain an explicit decomposition of the hypercube into separated clusters, such that this set carries almost all of $\mu_{\beta}$.

\begin{proposition}
\label{prop:soft-ogp-most-gibbs-measure}
    Given $H$, let $\hat S\subseteq \bSig_N$ denote the set of $\zv$ such that $(H,\zv)$ obeys the $(\eps,\delta)$-soft OGP as in Definition~\ref{def:soft-OGP}.
    Then for some $c>0$ independent of $N$, but possibly depending on all other parameters in \eqref{eq:eps-to-zero-slowly}, we have that for the null model $\QQ$
    \[
    \QQ[\mu_{\beta,H}(\hat S)\geq 1-e^{-cN/2}]\geq 1-e^{-cN/2}
    \]
\end{proposition}

\begin{proof}
    By \cref{prop:soft-OGP-planted}, we know that $\zv\in\hat S$ holds with exponentially good probability $1-e^{-cN}$ under $\QQ$, i.e.
    \[
    \EE^{\QQ}
    [\mu_{\beta,H}(\hat S)]
    \geq 
    1-e^{-cN}.
    \]
    Applying Markov's inequality to upper bound $\QQ[1-\mu_{\beta,H}(\hat{S}) \ge e^{-c N/2}]$ completes the proof.
\end{proof}

Given any Hamiltonian $H:\bSig_N\to\RR$ and soft-OGP parameters $(\eps,\delta)$, we construct a partition 
\begin{equation}
\label{eq:partition}
\hat S
=
\bigcup_{i=1}^K
\cC_i
\end{equation}
of the set $\hat S$  of $(\eps,\delta)$-soft-OGP bit strings into disjoint components as follows.
Each $\cC_i$ is a connected component of the subgraph induced by $\hat S \subseteq \bSig_N$, under the graph structure in which $\zv,\zv'\in \hat S $ are connected if and only if their Hamming distance is at most $(\eps/k)^3 N$.

The following proposition proves the shattering asserted in \cref{thm:shattering}.
\begin{proposition}
\label{prop:shattering}
For random $H\sim \GG_\ER(N,k,D)$ and parameters satisfying \eqref{eq:eps-to-zero-slowly} and \eqref{eq:beta-k-def}, $\mu_{\beta,H}$ is $(\eps/k)$-shattered as in \cref{def:shattering}, with the clusters given in \eqref{eq:partition}.
\end{proposition}

\begin{proof}
Since drawing $H$ from $\GG_\ER(N,k,D)$ is equivalent to drawing it from the null model $\QQ$, Property~\ref{it:cover-gibbs-measure} was just shown in Proposition~\ref{prop:soft-ogp-most-gibbs-measure}.
Property~\ref{it:separation-bound} follows by definition of connected components. 

Property~\ref{it:diameter-bound} holds by \cref{def:soft-OGP} of the soft OGP, since around each $z\in \hat S$ is a bottleneck at Hamming distance $\eps N/k$ with width $(\eps/k)^2N$ that cannot be crossed with step size at most $(\eps/k)^3 N$. 

Property~\ref{it:each-cluster-small} then follows from \cref{cor:most-pairs-orthogonal}, which shows for any i.i.d. Gibbs samples $\zv,\zv'$ have
\[
\PP(|\zv^\top \zv'|/N \ge \delta) \le e^{-cN}.
\]
If a single cluster $\cC_i$ had mass $\mu_{\beta,H}(\cC_i) > e^{-c N/2}$, then the two independent Gibbs samples would both land in $\cC_i$ with probability $> e^{-c N}$. But since two points in the same cluster satisfy $d_H(\zv,\zv')\le \eps N/k$, they would have overlap $|\zv^\top \zv'|/N\ge 1 - 2\eps/k \gg \delta$, contradicting the above bound.
\end{proof}

\subsubsection{From Shattering to Disorder Chaos (Proof of \cref{thm:single-clause-disorder-chaos})}

We now prove \cref{thm:single-clause-disorder-chaos} (disorder chaos) from the shattering established above, using a similar strategy as in \cite[Section 5.1]{alaoui2023shattering} applied to a different ensemble.

The idea is to consider a model with $M-1$ clauses (to which the preceding analysis applies with no changes) and show that adding another clause injects significant randomness into $\mu_{\beta}$.
The point is that ``adding two i.i.d.\ $M$-th clauses'' and comparing the results is equivalent to modifying $1$ existing clause out of $M$.
Call $H$ the original Hamiltonian with $M-1$ clauses, and $H',H''$ the perturbations with $1$ additional clause each.

Let 
\[
\tilde\mu(\zv)=\mathbb E[\mu_{\beta,H'}(\zv)|H]
\]
be the conditional expectation of $\mu_{\beta,H'}$ given $H$. In what follows, we work on the event that $\mu_{\beta, H}$, $\mu_{\beta, H'}$, and $\mu_{\beta, H''}$ are $\eps$-shattered.

Let $\cC_1, \cC_2,\ldots, \cC_K$ be a $(\eps/k)$-shattered cluster decomposition of the Gibbs measure for $H$ (with $M-1$ clauses) whose existence is guaranteed by \cref{prop:shattering}.
We note this decomposition remains valid for $H'$ and $H''$ up to adjusting $c$ by a constant factor (since adding one clause changes $H(\zv)$ by at most $1$ for any $\zv$, hence $e^{-2\beta} \le \mu_{\beta,H'}/\mu_{\beta,H} \le e^{2\beta}$ pointwise).
Let $\cC_0 = \bSig_N \setminus \bigcup_{i=1}^K \cC_i$ be the rest of the bit strings.
For the measures $\tilde{\mu}, \mu_{\beta,H'}, \mu_{\beta,H''}$, we define the following ``compressed'' measures supported on $\{0, 1,2,\dots,K\}$:
\begin{align*}
\hat\mu(\{i\}) &:=\tilde\mu(\cC_i) = \E\big[\mu_{\beta,H'}(\cC_i)\,\big|\, H\big], \\
\hat{\mu}'(\{i\}) &:= \mu_{\beta,H'}(\cC_i), \\
\hat{\mu}''(\{i\}) &:= \mu_{\beta,H''}(\cC_i).
\end{align*}

\begin{proposition}
\label{prop:compressed-TV-to-Wasserstein}
    Suppose that 
    \begin{equation}
    \label{eq:compressed-TV-condition}
    \liminf_{N\to\infty}
    \mathbb E[d_{\TV}(\hat\mu,\hat\mu')]
    >0.
    \end{equation}
    Then 
    \[
    \liminf_{N\to\infty}
    \mathbb E[W_1(\mu_{\beta,H''},\mu_{\beta,H'})]
    >0.
    \]
\end{proposition}

\begin{proof}
    First, since $W_1(\cdot, \nu)$ is convex in the first argument for any fixed $\nu$, Jensen's inequality (conditionally on $H$) implies that
    \[
    \mathbb E[W_1(\tilde\mu,\mu_{\beta,H'})]
    \leq 
    \mathbb E[W_1(\mu_{\beta,H''},\mu_{\beta,H'})].
    \]
    Thus it suffices to show that 
    \[
    \liminf_{N\to\infty}
    \mathbb E[W_1(\tilde\mu,\mu_{\beta,H'})]
    >0.
    \]
    For this, Proposition~\ref{prop:shattering} easily yields
    \[
    \mathbb E[W_1(\tilde\mu,\mu_{\beta,H'})]
    \geq 
    \frac{2\eps^{3}}{k^3}
    \EE[d_{\TV}(\hat\mu,\hat\mu')]
    -
    e^{-\Omega(N)}.
    \]
    More precisely, the mass that must be ``redistributed'' between clusters is
$d_\TV(\hat{\mu}, \hat{\mu}')$, and each unit of redistributed mass travels $\ell_1$-distance at least $\|\zv-\zv'\|_1 = 2d_H(\zv,\zv') \ge 2(\eps/k)^3 N$ (the minimum inter-cluster $\ell_1$-distance). The mass that stays within the same cluster contributes $\ge 0$ to $W_1$.
The mass outside all clusters is at most $e^{-\Omega(N)}$ (by shattering, property \ref{it:cover-gibbs-measure}).
    Combined with \eqref{eq:compressed-TV-condition}, this completes the proof.
\end{proof}

Thus, it remains to establish \eqref{eq:compressed-TV-condition}, which is the focus below.
We next show that it suffices to establish anti-concentration for the effect of the new clause on the ratio between most pairs of cluster weights.

\begin{proposition}
\label{prop:pair-anticoncentration-suffices}
    Let $(\nu_N,\nu_N')_{N\ge1}$ be a sequence of pairs of discrete probability measures on a countable set $\mathcal{X}_N$.
    If  $\lim_{N\to\infty} d_{\TV}(\nu_N,\nu_N')=0$, then for $i_N,j_N$ drawn i.i.d. from $\nu_N$,  the likelihood ratio
    \begin{equation}
    \label{eq:ratio-criterion}
    \frac{\nu_N(i_N)\cdot \nu_N'(j_N)}{\nu_N(j_N)\cdot \nu_N'(i_N)} \to 1 \qquad \text{in probability}.
    \end{equation}
\end{proposition}

\begin{proof}
By definition,
\[
d_{\TV}(\nu_N,\nu_N') =
\frac{1}{2}\sum_{x \in \mathcal{X}_N}\big|\nu_N(x) - \nu_N'(x)\big|
\ge
\frac{1}{2}\E_{x\sim \nu_N}\lt[\lt|\frac{\nu_N'(x)}{\nu_N(x)}-1\rt|\rt],
\]
where the inequality is due to the possible existence of $x\in \mathcal{X}_N$ where $\nu_N(x)=0$ but $\nu_N'(x)\neq 0$.
Thus both $\nu_N'(i_N)/\nu_N(i_N)$ and ${\nu_N'(j_N)}/{\nu_N(j_N)}$ converge to $1$ in probability, which easily yields the result.
\end{proof}

Thus it remains to consider the ratio \eqref{eq:ratio-criterion} (with $\nu_N=\hat\mu$ and $\nu_N'=\hat\mu'$) and show it does not converge to $1$ in probability (along any subsequence of $N$).

\begin{proposition}
\label{prop:ratio-anticoncentrates}
    With $\nu_N=\hat\mu$ and $\nu_N'=\hat\mu'$, the ratio \eqref{eq:ratio-criterion} does not converge to $1$ in probability (along any subsequence $N_1<N_2<\dots$).
\end{proposition}

\begin{proof}
Since the likelihood ratio in \eqref{eq:ratio-criterion} is unchanged if we rescale either measure by a positive constant, we may normalize the cluster weights on $\{1,\dots,K\}$. Let
\[
S=\bigcup_{\ell=1}^K \cC_\ell,
\qquad
\lambda(\{\ell\}) := \frac{\mu_{\beta,H}(\cC_\ell)}{\mu_{\beta,H}(S)}.
\]
By property~\ref{it:cover-gibbs-measure}, we have $\mu_{\beta,H}(S)=1-e^{-\Omega(N)}$.
Also, adding one clause changes every Gibbs weight by a factor in $[e^{-\beta},e^{\beta}]$, so after normalization there is a constant $C_\beta<\infty$ such that for every $\ell$,
\[
C_\beta^{-1}\lambda(\{\ell\}) \le \frac{\hat\mu(\{\ell\})}{\hat\mu(\{1,\ldots,K\})} \le C_\beta \lambda(\{\ell\}).
\]
Hence any event involving the label pair $(i,j)$ that has probability bounded below under $\lambda\otimes\lambda$ also has probability bounded below under the normalized version of $\hat\mu\otimes\hat\mu$. It is therefore enough to prove the proposition with $i,j$ drawn i.i.d. from $\lambda$, and we return to $\hat{\mu}$ at the end of this proof.

Now draw $i,j$ i.i.d. from $\lambda$, and conditional on $(i,j)$ let $\zv^{(i)},\zv^{(j)}$ be independent Gibbs samples from $\mu_{\beta,H}$ conditioned to lie in $\cC_i,\cC_j$. Averaged over the random labels $(i,j)\sim \lambda\otimes \lambda$, the pair $(\zv^{(i)},\zv^{(j)})$ is distributed as two i.i.d. Gibbs samples, conditioned on lying in $S$.
Since $\mu_{\beta,H}(S) = 1-e^{-\Omega(N)}$, this conditioning changes probabilities by at most $e^{-\Omega(N)}$.

Let $E=\{e_1,\dots,e_k\}$ be the random support of the new clause in $H'$, and write
\[
\chi_E(\zv):=\prod_{a=1}^k z_{e_a}.
\]
For each cluster $\cC_\ell$, let $\hat \zv^{(\ell)}\in\bSig_N$ be the coordinate-wise weighted majority vote within $\cC_\ell$, that is, $\hat{z}^{(\ell)}_a = \sgn(\sum_{\zv\in \cC_\ell} \mu_{\beta,H}(\zv) z_a)$ for $a = 1,2,\ldots,N$, with ties broken arbitrarily.
Since every pair of points in $\cC_\ell$ has Hamming distance at most $\eps N/k$ by property~\ref{it:diameter-bound}, and the majority vote minimizes the average Hamming distance, then
\[
\mathbb E\big[d_H(\zv,\hat \zv^{(\ell)}) ~\big|~ \zv \sim \mu_{\beta, H},~ \zv\in \cC_\ell\big]\le \eps N/k.
\]
This implies that on a uniformly random coordinate, $\zv$ and $\hat\zv^{(\ell)}$ differ with probability at most $\eps/k$.
For any fixed support $E$ and for each cluster $\ell$, define
\begin{equation}
\label{eq:minority-weight}
    q_\ell(E) := \PP\Big[\zv|_E \neq \hat \zv^{(\ell)}|_E \,\Big|\, \zv\sim \mu_{\beta,H},~ \zv\in \cC_\ell\Big].
\end{equation}
where $\zv|_E=\hat\zv|_E$ denotes $z_a=\hat{z}_a$ $\forall a\in E$.
For each $\ell$, over the uniformly random support $E$, a union bound over the $k$ coordinates give
\begin{equation}
    \EE_E[q_\ell (E)] \le k \cdot \frac{\eps}{k} = \eps \qquad
    \Longrightarrow \qquad \PP_E [q_\ell(E) > \sqrt{\eps}] \le \sqrt{\eps},
\end{equation}
where the implication follows from Markov's inequality.
We then define the event (depending only on $E$ and the labels $i,j$)
\begin{equation}
\label{eq:majority-vote-usually-correct}
A_1 := \Big\{q_i(E) \le \sqrt{\eps} \quad \text{and} \quad q_j(E) \le \sqrt{\eps} \Big\}, 
\qquad
\PP[A_1] \ge 1-2 \sqrt{\eps}.
\end{equation}

Next, \cref{cor:most-pairs-orthogonal} together with \cite[Proposition 1.4.14]{talagrand2011mean1} implies that for a uniformly random set of $O(1)$ coordinates, the joint marginal of two i.i.d. Gibbs samples is asymptotically uniform (up to expected total variational error tending to $0$ with $N$). Hence
\[
\PP \big[\chi_E(\zv^{(i)})\neq \chi_E(\zv^{(j)})\big] = \frac12 + o_N(1).
\]
Define
\[
s_\ell :=\chi_E(\hat{\zv}^{(\ell)})
\]
as the clause value of the majority representative (which is a deterministic function of $H,E,\ell$), and the event
\[
A_2 = \{s_i \neq s_j\}.
\]
On $A_1$ we have
\[
\chi_E\big(\zv^{(i)}\big)=s_i,
\qquad
\chi_E\big(\zv^{(j)}\big)=s_j,
\]
except with probability at most $q_i(E)+q_j(E)\le 2\sqrt\eps$, since $\zv|_E=\hat\zv^{(\ell)}|_E$ forces $\chi_E(\zv)=s_\ell$.
Therefore, for large $N$ and small $\eps$
\begin{align}
\PP[A_2] &\ge \PP\big[\chi_E(\zv^{(i)})\neq \chi_E(\zv^{(j)})\big] - 2\sqrt{\eps} - \PP[A_1^c]  \ge \frac12-O(\sqrt{\eps}) - o_N(1) \ge \frac13,
\\
\text{and}
\qquad
\PP[A_1\cap A_2] &\ge \PP[A_1] + \PP[A_2] - 1 \ge \frac13 - O(\sqrt{\eps})\ge \frac14.
\end{align}

Now fix $(H,E,i,j)$, and for $J\in\{+1,-1\}$ let $H^{(J)}=H-J\chi_E$ be the Hamiltonian obtained by adding the clause with support $E$ and sign $J$. Let $\hat\mu^{(J)}$ denote the corresponding compressed cluster measure. For each cluster $\cC_\ell$, define
\[
m_\ell(J):=\mathbb E\big[e^{\beta J\chi_E(\zv)} \,\big|\, \zv\sim \mu_{\beta,H},~ \zv\in \cC_\ell\big],
\]
where the expectation is only over the conditional Gibbs sample $\zv$.
Since $\mu_{\beta,H^{(J)}}(\zv)\propto \mu_{\beta,H}(\zv)e^{\beta J\chi_E(\zv)}$, we have
\[
\frac{\hat\mu^{(J)}(\{i\})}{\hat\mu^{(J)}(\{j\})}
=
\frac{\mu_{\beta,H}(\cC_i)\,m_i(J)}{\mu_{\beta,H}(\cC_j)\,m_j(J)}.
\]
Splitting the expectation defining $m_\ell(J)$ according to whether $\zv|_E=\hat\zv^{(\ell)}|_E$, and noting that the two possible values $e^{\pm\beta}$ of $e^{\beta J\chi_E(\zv)}$ differ by only $2\sinh\beta = O(\beta)$, we get on the event $A_1$ that
\[
m_i(J)=e^{\beta J s_i}(1+O(\beta\sqrt{\eps})),
\qquad
m_j(J)=e^{\beta J s_j}(1+O(\beta\sqrt{\eps})).
\]
Hence, on $A_1$,
\begin{equation}
\frac{\hat\mu^{(+)}(\{i\})/\hat\mu^{(+)}(\{j\})}{\hat\mu^{(-)}(\{i\})/\hat\mu^{(-)}(\{j\})}
=
e^{2\beta(s_i-s_j)}(1+O(\beta\sqrt{\eps})).
\label{eq:mu-ratio}
\end{equation}
Furthermore, on the event $A_1 \cap A_2$, we have $s_i-s_j \in \{\pm 2\}$, so the above value is $e^{\pm4\beta}(1+O(\beta\sqrt{\eps}))$. In particular, for $\eps$ sufficiently small, \eqref{eq:mu-ratio} lies outside of $[e^{-2\beta}, e^{2\beta}]$.

Choose $\delta_*>0$ so small that
\[
\frac{1+\delta_*}{1-\delta_*}<e^{2\beta}.
\]
With the uniformly small $\eps$ used above, on the event $A_1\cap A_2$ the two possible values of
\[
R_J:=\frac{\hat\mu(\{i\})\hat\mu^{(J)}(\{j\})}{\hat\mu(\{j\})\hat\mu^{(J)}(\{i\})}
\]
satisfy $\max\{R_+/R_-,R_-/R_+\}>(1+\delta_*)/(1-\delta_*)$, by \eqref{eq:mu-ratio}.
Since $J$ is an unbiased sign, it follows that conditional on $H,E,i,j$ and $A_1\cap A_2$, at least one of the two sign choices lies outside $[1-\delta_*,1+\delta_*]$, and therefore
\[
\mathbb P\big[R_{J_M}\notin [1-\delta_*,1+\delta_*] \,\big|\, H,E,i,j,A_1\cap A_2\big]\ge \frac12.
\]
Combining this with $\mathbb P[A_1\cap A_2]\ge 1/4$, we obtain
\[
\mathbb P\left[
\frac{\hat\mu(\{i\})\hat\mu'(\{j\})}{\hat\mu(\{j\})\hat\mu'(\{i\})}
\notin [1-\delta_*,1+\delta_*]
\right]\ge \frac18-o_N(1),
\]
where $\hat\mu'=\hat\mu^{(J_M)}$ is the compressed measure for the realized perturbation $H'$. By the pointwise comparability of $\hat\mu/\hat\mu(\{1,\ldots,K\})$ and $\lambda$, the same positive lower bound (up to a constant factor depending only on $\beta$) holds when $(i,j)$ are sampled from the normalized version of $\hat\mu$, as in \eqref{eq:ratio-criterion}. Therefore the likelihood ratio in \eqref{eq:ratio-criterion} does not converge to $1$ in probability (along any subsequence of $N$).
\end{proof}

\begin{proof}[Proof of Theorem~\ref{thm:single-clause-disorder-chaos}]
In the notation above, the pair $(H',H'')$ has exactly the coupling described in the theorem statement.
Suppose for contradiction that $\EE[W_1(\mu_{\beta,H'},\mu_{\beta,H''})]\to 0$ along some subsequence.
By \cref{prop:compressed-TV-to-Wasserstein}, $\EE[d_\TV(\hat\mu,\hat\mu')]\to0$ along a further subsequence, so $d_\TV(\hat\mu,\hat\mu')\to0$ in probability, and \cref{prop:pair-anticoncentration-suffices} then forces the likelihood ratio \eqref{eq:ratio-criterion} to converge to $1$ in probability along that subsequence, which contradicts \cref{prop:ratio-anticoncentrates}.
\end{proof}

\section*{Acknowledgments}
We thank  David Gamarnik, David Gosset, and Jarrod McClean for comments on a draft of this work. LZ thanks David Gamarnik for helpful discussions.
LZ is funded by ONR Award No. N000142612399 and the Google Research Scholar Award.

\paragraph{AI usage statement.} Almost all of the proof ideas in this work are human-generated. GPT‑5.6 Sol, GPT‑6 Astra, and Claude Fable 5.1 were used during the preparation of this manuscript to help verify the results and repair minor proof issues. Claude Fable 5.1 was used to implement Python code that generated the numerical data in Tables~\ref{tab:dense-betas}, \ref{tab:xorsat-grid}, and \ref{tab:xorsat-grid-reg}. The authors take full responsibility for the correctness of the results.

\clearpage
\appendix 

\section{Non-Stability of DQI via Right-invertible XORSAT}
\label{app:invertible}

In this appendix, we give a simple argument showing that DQI is not a stable algorithm.
Specifically, we consider a family of XORSAT instances with a solution set that exhibits shattering.
On the other hand, both DQI and a classical algorithm can sample from their Gibbs measure exactly.

\begin{theorem}
\label{thm:invertible}
For any $k\ge 3$, there exists $\alpha \in (0,1)$ such that with high probability, a random $k$-XORSAT instance $H\sim\GG_{\ER}(N,k,\alpha k)$ has a solution set that shatters into exponentially many well-separated clusters.
Nevertheless, both DQI and a classical (non-stable) algorithm can efficiently sample from its Gibbs measure $\mu_{\beta,H}$ at every $\beta\in[0,\infty]$.
\end{theorem}

Since the solution set (i.e., the zero-temperature Gibbs measure at $\beta=\infty$) exhibits shattering, we expect that shattering-to-hardness argument from \cref{thm:single-clause-disorder-chaos,thm:main-hardness} can be applied to rule out stable sampling algorithms. As DQI can sample this solution set as well as its Gibbs measure at any finite temperature, this implies that DQI is not a stable algorithm.

To show \cref{thm:invertible}, we start by discussing some linear algebraic properties of the XORSAT problem.
Let $R=\rank(B)$ over the finite field $\FF_2$.
Let $\range(B) \subseteq \FF_2^M$ be the range (or column space) of $B$.
The linear system of equation $B\xv=\vv$ is satisfiable (i.e., has a solution) if and only if $\vv \in \range(B) \subseteq \FF_2^M$.
The number of solutions (if they exist) is $2^{N-R}$, and all solutions can always be written in the form of
\begin{equation}
    \xv = B^+ \vv + \wv, \qquad \wv \in \ker(B) \subseteq \FF_2^N
\end{equation}
and $B^+$ is a generalized inverse of $B$ (any $B^+$  satisfying $BB^+B=B$). We also consider the following special case of XORSAT problems where $B^+$ is actually a right-inverse:
\begin{definition}
We say any XORSAT instance $H=(B,J)$ is \emph{right-invertible} when $B$ is right-invertible, i.e., there exists $B^+$ such that $BB^+ = I$. This is true iff $B$ has full row rank, i.e., $\rank(B)=M\le N$.
\end{definition}

\begin{lemma}\label{lem:classical-sample-alg}
Any right-invertible XORSAT instance is satisfiable, and there is an efficient classical algorithm to sample exactly from $\mu_{\beta,H}$ for all $\beta \in [0, \infty]$.
\end{lemma}
\begin{proof}
Since $B$ has full row rank, let $B^+$ be the right inverse of $B$. 
Note for each of the $2^M$ possible vectors $\tilde{\vv}\in \FF_2^M$, there are $2^{N-M}$ solutions to $B\xv=\tilde{\vv}$.
As a result, the XORSAT problem is satisfiable regardless of the choice of $\vv$. 
The following algorithm can sample from $\mu_{\beta,H}$ efficiently:
\begin{enumerate}
    \item Sample $s\in\{0,1,\ldots,M\}$ as the number of unsatisfied clauses according to the $\Pr(s)= \frac{1}{Z_s} \binom{M}{s}e^{\beta (M-2s)}$, where $Z_s=(2 \cosh\beta)^M$.
    \item Sample $\tilde{\vv}$ as one of $\binom{M}{s}$ vectors that disagree with $\vv$ on $s$ bits (obtain by flipping $s$ bits of $\vv$ at random).
    \item Output sample $\tilde\xv = B^+\tilde{\vv} + \wv$, where $\wv$ is chosen uniformly at random from the $2^{N-M}$ vectors in $\ker(B)$. This can be done efficiently by finding a basis for $\ker(B)$.
\end{enumerate}
At zero temperature ($\beta=\infty$), take $s=0$ deterministically.
Note every $\tilde{\vv}$ is sampled with probability proportional to $e^{-2\beta s}=e^{-2\beta |\vv-\tilde{\vv}|}$, and they have the same $2^{N-M}$ number of pre-images. Thus this protocol samples each $\xv$ with its Gibbs weight proportional to $e^{-\beta H((-1)^\xv)}$, as claimed.
\end{proof}

We now prove the main claim of this section, Theorem~\ref{thm:invertible}, by showing that the Gibbs measures of certain sparse random XORSAT problems exhibit shattering at $\beta=\infty$, but are nevertheless admit efficient sampling algorithms due to their right-invertibility.

\begin{proof}[Proof of Theorem~\ref{thm:invertible}]
Consider the choice of $\beta=\infty$. Then the Gibbs measure is simply the uniform distribution of all the solutions. It is known in e.g.~\cite{Ibrahimi_2015} that for $\alpha\in (\alpha_d(k), \alpha_s(k))$, where $\alpha_s(k)\le 1$, a random XORSAT instance from $\GG_{\ER}(N,k,\alpha k)$ is satisfiable with high probability but the set of ground states shatters into exponentially large (in $N$) number of clusters, each containing an exponential number of solutions.
(For example, $\alpha_s(3)=0.917935$~\cite{Dietzfelbinger2010cuckoo}, and $\alpha_d(3)=0.818469$ \cite[Theorem 1]{Ibrahimi_2015}.)

For any matrix $A \in \FF_2^{M\times N}$, let $Z_A\ge 1$ be the number of solutions $\xv$ to $A\xv=0$ (i.e., the cardinality of $\ker{A}$).
Note $Z_A=2^{N-\rank(A)}$, and $B$ is right-invertible iff $Z_{B^T}=1$.
Now, for any distribution $\GG$ over $H=(B,\vv)$ where $\vv$ is chosen uniformly from $\FF_2^M$ and independent of $B$, $H$ is satisfiable iff $\vv \in \range(B)$, so
\begin{equation} \label{eq:RI-upperbound}
    \PP(H \text{ is satisfiable}) = \EE\left[\frac{2^{\rank(B)}}{2^M}\right] =  2^{N-M} \EE[1/Z_{B}] \ge \PP(Z_{B^T} =1).
\end{equation}
 Furthermore, (see also \cite[Section 6.1]{dembo2009gibbs})
\begin{equation} \label{eq:RI-lowerbound}
    \PP(H \text{ is satisfiable})  \le \frac12 + \frac12\PP(Z_{B^T} = 1).
\end{equation}
Hence, if $H$ is satisfiable with high probability, then $H$ is right invertible with high probability. Hence the classical algorithm in \cref{lem:classical-sample-alg} can efficiently sample from $\mu_{\beta,H}$ with high probability.

Furthermore, if $B$ has full row rank, then $C^\perp = \ker B^T = \{0\}$, which only has only one codeword. In this case, the quantum decoding problem is trivial, and the DQI/Regev method in \cref{sec:DQI} prepares the target state for Gibbs sampling at every $\beta$.
\end{proof}

We remark, however, that under a basis transformation, the set of solutions for any right-invertible XORSAT no longer exhibits shattering. 
Specifically, if $B\in \FF_2^{M\times N}$ is right-invertible, then there exists a invertible matrix $U$ such that $BU = [I_M ~|~ 0]=: B'$ where $I_M$ is the $M\times M$ identity matrix. Hence, if we change the basis by considering $\yv=U^{-1} \xv$, then there is a bijection between solutions to $B\xv=\vv$ and solutions to $B'\yv=\vv$. The latter has a trivial solution space where $y_i = v_i$ for $1\le i \le M$ and $y_i$ is arbitrary for $M < i \le N$. Hence, the transformed solution set $\{U^{-1}\xv: B\xv=\vv\}$ does not exhibit shattering.
In a similar spirit, \cite{koehler2026overlap} show that changing the ``solution basis'' in the shortest path problem from paths to spanning trees removes the overlap gap limitation discovered in \cite{li2024some}.

For the above reason and more, we believe this example family of right-invertible, fully satisfiable XORSAT problems is somewhat unsatisfying, and thus have endeavored to consider the more general Max-XORSAT problem in the frustrated and glassy regime in the main body of this paper.

\section{Prange's Algorithm for Sampling}
\label{sec:prange}

Here we describe an algebraic sampling procedure, which we call \emph{Prange's algorithm} after \cite{prange1962use}, for the Erd\H{o}s-R\'enyi ensemble considered in \cref{def:erdos_renyi}.
In a sense, this is a natural generalization of the Gaussian elimination procedure in \cref{lem:classical-sample-alg} beyond the trivial right-invertible case.
Indeed, this is a canonical setting in which
algebraic structure allows one to bypass hardness barriers for \emph{stable} algorithms (in particular, Gaussian elimination is non-stable). 
Related examples based on lattice methods have also recently been connected to problems in statistics and learning theory \cite{song2021cryptographic,zadik2022lattice,diakonikolas2022non}.

Prange's algorithm is a generalization of the simple linear system example above to positive temperature.
Here, one retains each clause as a hard constraint with an appropriate independent probability, and discards the rest of the clauses. 
One then samples uniformly from the resulting \emph{random} linear system by Gaussian elimination. 
This approach intuitively represents a soft constraint as a mixture of a hard constraint and a non-constraint, but requires technical care to control the possibility of exponential multiplicative fluctuations in the number of solutions.
We show below that in a sizable parameter regime, this yields a Gibbs sampler with $o(1)$ total variation error (with high probability over the random instance), despite provable failure of stable sampling algorithms.

\paragraph{A mixture representation of the Gibbs weight.}
Write
\[
H(\zv) = -\sum_{m=1}^M f_m(\zv).
\]
Since $f_m(\zv)\in\{\pm1\}$, we may rewrite
\begin{equation}
\label{eq:cluster-expand-1}
e^{-\beta H(\zv)}=\prod_{m=1}^M e^{\beta f_m(\zv)}
\propto \prod_{m=1}^M e^{2\beta\,\Id\{f_m(\zv)=1\}}.
\end{equation}
Set $\lambda:=e^{2\beta}-1\ge 0$, so that for $f\in\{\pm1\}$,
\[
e^{2\beta\,\Id\{f=1\}}=1+\lambda\,\Id\{f=1\}.
\]
Plugging this into \eqref{eq:cluster-expand-1} and expanding gives the pointwise identity (for every $\zv\in\bSig_N$)
\begin{equation}
\label{eq:pointwise-highT}
e^{-\beta H(\zv)}
\propto
\prod_{m=1}^M \bigl(1+\lambda\,\Id\{f_m(\zv)=1\}\bigr)
=
\sum_{S\subseteq[M]} \lambda^{|S|}\,\Id\{f_s(\zv)=1\ \forall s\in S\}.
\end{equation}
For $S\subseteq[M]$, define
\[
\Omega_S:=\{\zv \in\bSig_N:\ f_s(\zv)=1\ \forall s\in S\},
\qquad
Z_S:=|\Omega_S|.
\]
Eq.~\eqref{eq:pointwise-highT} is most naturally interpreted as an identity of \emph{unnormalized measures on $\bSig_N$}.
Equivalently, introduce the joint unnormalized measure on $(S,\zv)$ by
\begin{equation}
\label{eq:joint-measure}
\Pi_\beta(S,\zv)\propto \lambda^{|S|}\,\Id\{\zv\in\Omega_S\}.
\end{equation}
Its $\zv$-marginal is proportional to $e^{-\beta H(\zv)}$ (and hence equals the Gibbs law after normalization), while $\zv\mid S$
is uniform on $\Omega_S$. Therefore exact Gibbs sampling can be performed by the following two steps:
sample $S$ from its marginal
\begin{equation}
\label{eq:tilde-nu-def}
\tilde\nu(S):=\Pi_\beta(S)=\frac{\lambda^{|S|}Z_S}{\sum_{T\subseteq[M]}\lambda^{|T|}Z_T},
\end{equation}
and then sample $\zv$ uniformly from $\Omega_S$. Summing \eqref{eq:pointwise-highT} over $\zv$ recovers the partition-function identity
\[
Z_\beta\propto\sum_S \lambda^{|S|}Z_S.
\]

\paragraph{Normalization and the proposal law.}
For fixed $S$ and random clauses, $\mathbb E[Z_S]=2^{N-|S|}$, so define
\[
\tilde Z_S:=2^{|S|-N}Z_S
\qquad\text{so that}\qquad
\mathbb E[\tilde Z_S]=1.
\]
Then
\begin{equation}
\label{eq:prange-expansion}
Z_{\beta}\propto \sum_{S\subseteq[M]}\Big(\frac{\lambda}{2}\Big)^{|S|}\tilde Z_S.
\end{equation}
Let $\nu$ be the product measure on subsets $S\subseteq[M]$ with
\[
\nu(S)\propto \Big(\frac{\lambda}{2}\Big)^{|S|}
\qquad\Longleftrightarrow\qquad
\nu(S)=\theta^{|S|}(1-\theta)^{M-|S|},
\qquad
\theta:=\frac{\lambda}{2+\lambda}=\frac{e^{2\beta}-1}{e^{2\beta}+1}=\tanh(\beta).
\]
Define the rescaled partition function
\[
\tilde Z_\beta:=\sum_{S\subseteq[M]}\nu(S)\tilde Z_S.
\]
Comparing \eqref{eq:tilde-nu-def} with the above definitions yields the Radon--Nikodym identity
\begin{equation}
\label{eq:RN-prange}
\tilde\nu(S)=\frac{\nu(S)\tilde Z_S}{\tilde Z_\beta},
\qquad\text{i.e.}\qquad
\frac{\tilde\nu(S)}{\nu(S)}=\frac{\tilde Z_S}{\tilde Z_\beta}.
\end{equation}
When $S\sim\nu$, the law of large numbers gives $|S|/M\to\theta$ in probability, so the retained system has effective average degree $k|S|/N\to\theta D$.
The following proposition applies the exact full-row-rank threshold from \cref{thm:coja-fullrank} to this independently thinned set of rows.
\begin{proposition}
\label{prop:restatement-sat-threshold}
Fix any $k\ge 3$. Let $H\sim\GG_\ER(N,k,D)$ and retain each clause independently with probability $\theta$. Let $K_*(k)$ be the quantity defined in \eqref{eq:admissible}. If
\begin{equation}
\label{eq:prange-rank-condition}
\theta D<K_*(k),
\end{equation}
then, with probability $1-o_N(1)$ over the instance and the retained set $S$,
\[
\tilde Z_S = 1\qquad\text{(equivalently, }Z_S=2^{N-|S|}\text{)}.
\]
\end{proposition}

\begin{proof}
Choose $D'$ with $\theta D<D'<K_*(k)$. A Chernoff bound gives $|S|\le m':=\lfloor ND'/k\rfloor$ with probability $1-o_N(1)$. Conditional on $|S|=s$, the selected rows are exactly $s$ i.i.d.\ uniform $k$-subsets of $[N]$, as required by \cref{def:erdos_renyi}; their signs remain independent uniform Rademacher variables. Couple these rows to the first $s$ rows of an $m'$-row matrix with the same row law. By \cref{thm:coja-fullrank}, the $m'$-row matrix has full row rank with probability $1-o_N(1)$ because $km'/N\le D'<K_*(k)$. On that event its first $s$ rows are also linearly independent. Every right-hand side on the selected rows is then consistent and has exactly $2^{N-s}$ solutions, so $\tilde Z_S=1$.
\end{proof}

\paragraph{Prange's algorithm.}

Motivated by the two-stage sampling procedure above in \eqref{eq:tilde-nu-def}, we approximate $\tilde\nu$ by the simple deterministic measure $\nu$, and then sample $z$ uniformly from $\Omega_S$. This gives the following Prange-like sampling algorithm:

\begin{enumerate}
\item Sample $S\sim \nu$ (i.e.\ retain each clause independently with probability $\theta=\tanh\beta$).
\item Run Gaussian elimination on the constraints in $S$.
If the system is inconsistent or not full row rank (equivalently, $\tilde Z_S\neq 1$), declare failure and return an arbitrary bit string.
\item Otherwise output a uniformly random $\zv\in\Omega_S$ (which Gaussian elimination can sample).
\end{enumerate}

\begin{proposition}
\label{prop:prange-works}
If $D \tanh(\beta) < K_*(k) $, then Prange's algorithm outputs a sample whose law is within $o(1)$ expected total variation
distance of the Gibbs measure $\mu_{\beta,H}$.
\end{proposition}

\begin{proof}
By \eqref{eq:joint-measure}--\eqref{eq:tilde-nu-def}, exact Gibbs sampling is $\tilde\nu U$, where
$U(S,\cdot)$ is the uniform law on $\Omega_S$. Prange's algorithm (on success) outputs $\nu U$.
Thus it suffices to show $\mathbb E\,d_{\mathrm{TV}}(\nu,\tilde\nu)=o(1)$ and that the failure probability is $o(1)$.

Let $S\sim\nu$, where each clause is sampled independently with probability $\theta=\tanh\beta$.
Since $\theta D = D \tanh \beta < K_*(k)$,
Proposition~\ref{prop:restatement-sat-threshold}
implies $\mathbb P(\tilde Z_S=1)=1-o(1)$ under the joint law of $(S,\{f_m\})$.
Because $\mathbb E[\tilde Z_S\mid S]=1$ and $\tilde Z_S\ge 0$, we have the identity
$\mathbb E|\tilde Z_S-1|=2\,\mathbb E[(1-\tilde Z_S)_+]$, where $(x)_+ = \max(0,x)$.  Since $(1-\tilde Z_S)_+\in[0,1]$ and
$(1-\tilde Z_S)_+\to 0$ in probability, it follows that
\begin{equation}
\label{eq:L1tildeZS}
\mathbb E|\tilde Z_S-1|=o(1),
\end{equation}
where the expectation is over the random instance and the draw $S\sim\nu$.

Next, by definition $\tilde Z_\beta=\mathbb E_{S\sim\nu}[\tilde Z_S\mid \{f_m\}]$, so Jensen and \eqref{eq:L1tildeZS} give
\[
\mathbb E|\tilde Z_\beta-1|
=
\mathbb E\Big|\mathbb E_{S\sim\nu}\big[\tilde Z_S-1\mid \{f_m\}\big]\Big|
\le
\mathbb E_{S\sim\nu}\,\mathbb E|\tilde Z_S-1|
=
o(1).
\]
In particular, $\tilde Z_\beta\stackrel{p}{\to}1$, and therefore $\tilde Z_S/\tilde Z_\beta\stackrel{p}{\to}1$ for $S\sim\nu$.

Using \eqref{eq:RN-prange} we can write
\[
d_{\mathrm{TV}}(\nu,\tilde\nu)
=
\mathbb E_{S\sim\nu}\Big(1-\frac{\tilde\nu(S)}{\nu(S)}\Big)_+
=
\mathbb E_{S\sim\nu}\Big(1-\frac{\tilde Z_S}{\tilde Z_\beta}\Big)_+.
\]
The integrand lies in $[0,1]$ and converges to $0$ in probability, so taking expectation over the random instance yields
$\mathbb E\,d_{\mathrm{TV}}(\nu,\tilde\nu)=o(1)$.

Finally, total variation contracts under post-processing, hence
$d_{\mathrm{TV}}(\nu U,\tilde\nu U)\le d_{\mathrm{TV}}(\nu,\tilde\nu)$.
Moreover, the failure event is exactly $\{\tilde Z_S\neq 1\}$, which has probability $o(1)$ by Proposition~\ref{prop:restatement-sat-threshold}.
Combining these bounds completes the proof.
\end{proof}

\section{Derivation of $\beta_{\mathrm{Holevo}}$ and $\beta_{\mathrm{Shannon}}$}
\label{app:holevo_shannon}

\subsection{Holevo Bound}

In DQI, we decode from the classical-quantum channel
\begin{align}
    \label{eq:cq_channel}
    0 &\to \ket{\psi_{0,\gamma}} = \sqrt{1 - \gamma} \ket{0} + \sqrt{\gamma} \ket{1} \\
    1 &\to \ket{\psi_{1,\gamma}} = \sqrt{\gamma} \ket{0} + \sqrt{1-\gamma} \ket{1}
\end{align}
The Holevo bound says that, if 0 and 1 are equiprobable then the rate $R$ of the channel is limited by
\begin{equation}
    \label{eq:Holevo}
    R \leq \chi(\gamma)
\end{equation}
where
\begin{equation}
    \chi(\gamma) = S \left( \frac{1}{2} \ket{\psi_{0,\gamma}}\bra{\psi_{0,\gamma}} + \frac{1}{2} \ket{\psi_{1,\gamma}}\bra{\psi_{1,\gamma}} \right)
\end{equation}
and $S$ denotes the von Neumann entropy (in bits). By direct calculation, one finds
\begin{equation}
    \label{eq:chi_val}
    \chi(\gamma) = H_2 \left( \frac{1}{2} - \sqrt{\gamma(1-\gamma)} \right),
\end{equation}
where $H_2(p) = - p \log_2 (p) - (1-p) \log_2 (1-p)$ is the binary entropy. Note $\chi(\gamma)= 1-2\gamma/\ln2 + O(\gamma^2)$, so we want $R$ to be small if we want to allow $\gamma$ large.
In the context of DQI, $R$ is the rate of the code $C^\perp$, which is
\begin{equation}
    R = 1 - \frac{\rank(B)}{M}.
\end{equation}
The best case scenario is when $B$ is full rank, i.e., $\rank(B)=N$. For Max-$k$-XORSAT on hypergraphs with average degree $D$, we have $N/M = k/D$. Thus, we will take
\begin{equation}
    \label{eq:R_asymp}
    R = 1 - \frac{k}{D}.
\end{equation}
Substituting \cref{eq:R_asymp} and \cref{eq:chi_val} into \cref{eq:Holevo} yields
\begin{equation}
    1 - \frac{k}{D} \leq  H_2 \left( \frac{1}{2} - \sqrt{\gamma(1-\gamma)} \right).
\end{equation}
Let $\gamma_{\mathrm{Holevo}}$ be the value of $\gamma$ that saturates this inequality. Treating the above inequality as an equality and solving yields
\begin{equation}
    \label{eq:gamma_Holevo}
    \gamma_{\mathrm{Holevo}} = \frac{1}{2} - \sqrt{ H_2^{-1}\left( 1 - \frac{k}{D} \right)\left( 1 - H_2^{-1}\left( 1 - \frac{k}{D} \right)\right)}.
\end{equation}
Note that $H_2^{-1}$ has two branches. Here and throughout we take the one lying between 0 and $\frac12$. From \cref{thm:gamma_vs_beta} we have
\begin{equation}
    \label{eq:betamax_repeat}
    \beta_{\max} = \sech^{-1} \left( 1 - 2 \gamma \right).
\end{equation}
Substituting \cref{eq:gamma_Holevo} into \cref{eq:betamax_repeat} shows that the maximum achievable $\beta$ by DQI given a decoder saturating the Holevo bound is
\begin{equation} \label{eq:beta-Holevo-app}
    \beta_{\mathrm{Holevo}} = \sech^{-1} \left[ 2 \sqrt{ H_2^{-1} \left( 1 - \frac{k}{D} \right)\left( 1 - H_2^{-1}\left( 1 - \frac{k}{D} \right)\right)} \right].
\end{equation}

\subsection{Shannon Bound}

By measuring in the computational basis, the channel \cref{eq:cq_channel} induces a classical binary symmetric channel where the probability of a bit flip error is $\gamma$. For such a channel, the Shannon bound states
\begin{equation}
    \label{eq:shannon}
    R \leq 1 - H_2(\gamma).
\end{equation}
By \cref{eq:R_asymp} we thus have
\begin{equation}
    1 - \frac{k}{D} \leq 1 - H_2(\gamma).
\end{equation}
We next treat this as an equality and solve for $\gamma_{\mathrm{Shannon}}$. The result is
\begin{equation}
    \gamma_{\mathrm{Shannon}} = H_2^{-1} \left( \frac{k}{D} \right).
\end{equation}
By \cref{eq:betamax_repeat} we thus have
\begin{equation}
    \beta_{\mathrm{Shannon}} = \sech^{-1} \left[ 1 - 2 H_2^{-1} \left( \frac{k}{D} \right) \right].
\end{equation}

\section{Phase Transitions and Algorithmic Thresholds in Spin Glasses}
\label{app:spin-glass}

\subsection{Dense $k$-spin Model: the $D\to\infty$ Limit}

Here we review some known results on phase transitions in the dense $k$-spin model, which is connected to the Max-$k$-XORSAT problems on sparse degree-$D$  graphs (such as the Erd\H{o}s-R\'enyi ensemble in \cref{def:erdos_renyi}) in the $D\to\infty$ limit~\cite{sen2018optimization}. For any given $k$, the dense $k$-spin Hamiltonian is
\begin{equation}
    H^{\rm dense}(\zv) = \frac{1}{\sqrt{N^{k-1}}}\sum_{i_1,\ldots,i_k=1}^N J_{i_1,\ldots,i_k} z_{i_1} \cdots z_{i_k}
\end{equation}
where $J_{i_1,\ldots,i_k} \sim \cN(0,1)$ are drawn from the standard Gaussian distribution.

There are two main relevant spin glass transitions in this model: (1) the dynamical transition at $\beta_\dyn(k)$, which is believed to mark the onset of shattering, ergodicity-breaking, and slowdown of Glauber dynamics; (2) the static transition at $\beta_\RS(k)$, which marks the onset of replica symmetry breaking.

\begin{conjecture}[\cite{alaoui2023sampling}]
For the dense $k$-spin model, no stable algorithms can efficiently Gibbs sample when $\beta > \beta_{\dyn}(k)$, where
\begin{equation} \label{eq:beta-dyn}
\beta_{\dyn}(k) = \inf\Big\{\beta > 0: \exists q\in (0,1] \text{ s.t. } q = \frac{\EE_G[\sinh(\lambda G) \tanh(\lambda G)]}{e^{\lambda^2/2}}, ~\lambda(q) = \beta \sqrt{k q^{k-1}}\Big\},
\end{equation}
where $G\sim \cN(0,1)$.
Asymptotically as $k\to\infty$, we have
\begin{equation}
\beta_\dyn(k) = \sqrt{\frac{2\ln k}{k}} \cdot [1+o_{k}(1)].
\end{equation}
\end{conjecture}

The static threshold, corresponding to the replica symmetry breaking transition, can be defined more directly through the free energy.  Let $\beta_\RS(k)$ denote this threshold for the dense Ising $k$-spin Hamiltonian used above, then it is the maximum inverse temperature where the free energy equals the ``annealed'' value of $\ln 2 + \beta^2/2$:
\begin{equation}
\label{eq:dense-beta-RS}
 \beta_\RS(k)
 :=\sup\left\{\beta :\
 \limsup_{N\to\infty}\frac1N\E\ln Z^{\rm dense}_{\beta, H}
 =\ln2+\frac{\beta^2}{2}\right\}.
\end{equation}
Talagrand's theorem \cite{talagrand2000rigorous}, in the normalization of \eqref{eq:dense-beta-RS}, gives $(1-2^{-k})\sqrt{2\ln2} \leq \beta_\RS(k) \leq \sqrt{2\ln2}$.
In particular, this shows
$\beta_\RS(k)=\sqrt{2\ln2}\,[1+o_{k}(1)]$. For finite $k$, $\beta_\RS(k)$ is determined by optimizing the 1RSB ansatz for the free energy~\cite{Panchenko2014parisi} finding the smallest $\beta$ where a nontrivial minimizer arises.

Besides the thermodynamical phase transitions, there exists an algorithmic threshold $E_\OGP$ associated with the overlap gap property~\cite{gamarnik2021overlap}, such that many classes of optimization algorithms (including all stable algorithms) fail to produce a configuration (bit string) with better energy $E<E_\OGP$.
For the dense $k$-spin model at any even $k$, this threshold $E_\OGP$ is proven to be tight in the sense that it is achievable by a family of stable algorithms based on approximate message passing~\cite{AMS21optimization,huang2022tight}.
In this case, the value of $E_\OGP$ can be computed numerically by an extended Parisi formula~\cite{AMS21optimization}, which is a generalization of the celebrated Parisi formula used to calculate the ground state energy~\cite{parisi1980sequence,talagrand2006parisi}.
To compare the optimization energy threshold $E_\OGP$ with the temperature thresholds on an equal footing, we use a standard conversion that maps energy to an effective (inverse) temperature:
\begin{equation}
\label{eq:beta-OGP}
\beta_\OGP = E^{-1}(E_\OGP), 
\qquad
\text{where}\quad 
E(\beta) := -\frac{\partial F}{\partial \beta}, \qquad F(\beta) = \lim_{N\to\infty} \frac{1}{N} \EE \ln Z^{\rm dense}_{\beta,H}.
\end{equation}

A numerical evaluation of $\beta_{\dyn}(k)$, $\beta_\RS(k)$, $\beta_\OGP(k)$  for various $k\ge 3$ is shown in \cref{tab:dense-betas}.
\begin{table}[h!]
\centering
\begin{tabular}{r@{\hspace{1.2em}}cccccc@{\hspace{1.2em}}c}
\hline
\hline
$k$ & $\beta_{\dyn}$ & $\beta_\RS $ & $E_{\rm ground}$ & $E_\OGP$ & $\beta_{\rm OGP}$  & $\beta_\OGP/\beta_\dyn$\\
\hline
 3 & 1.0374 & 1.0855 & $-$1.1500 & $-$1.1313 & 1.6163 & 1.558\\
 4 & 1.0423 & 1.1463 & $-$1.1673 & $-$1.1072 & 1.1072 & 1.062\\
 5 & 1.0099 & 1.1650 & $-$1.1732 & $-$1.0678 & 1.0678 & 1.057\\
 6 & 0.9720 & 1.1721 & $-$1.1756 & $-$1.0271 & 1.0271 & 1.057\\
 8 & 0.9017 & 1.1763 & $-$1.1770 & $-$0.9532 & 0.9532 & 1.057\\
10 & 0.8435 & 1.1772 & $-$1.1773 & $-$0.8920 & 0.8920 & 1.058\\
12 & 0.7956 & 1.1774 & $-$1.1774 & $-$0.8412 & 0.8412 & 1.057\\
%14 & 0.7555 & 1.1774 & $-$1.1774 & $-$0.7986 & 0.7986 & 1.057\\
16 & 0.7214 & 1.1774 & $-$1.1774 & $-$0.7621 & 0.7621 & 1.056\\
20 & 0.6662 & 1.1774 & $-$1.1774 & $-$0.7030 & 0.7030 & 1.055\\
\hline\hline
\end{tabular}
\caption{\label{tab:dense-betas}
The various critical temperatures and energies for the dense $k$-spin models. $\beta_\dyn$ and $\beta_\RS$ mark the dynamical and static phase transitions. $E_{\rm ground}$ is the ground energy density obtained from the Parisi formula, while $E_\OGP$ is the algorithmic threshold energy for stable algorithms. We also convert $E_\OGP$ to an effective temperature $\beta_{\OGP}$, and compare its value with $\beta_\dyn$. For $k\ge 4$ the listed $E_\OGP$ lies in the replica-symmetric regime, where $E(\beta)=-\beta$.
}
\end{table}

\paragraph{Connection to finite-connectivity models.}
Many results (e.g., \cite{sen2018optimization,chen2019suboptimality,chou2022limitations}) have connected the free energy of dense $k$-spin models to random $k$-XORSAT on degree-$D$ hypergraphs in the $D\to\infty$ limit. Following a rescaling to match $\EE_J[(H^{\rm dense}(\zv))^2]=N$, this suggests the following predictions:
\begin{align}
\label{eq:diluted-to-gaussian-beta-dyn}
    \lim_{D\to\infty}\beta_\dyn(k, D) \sqrt{\frac{D}{k}} = \beta_\dyn(k), \\
\label{eq:diluted-to-gaussian-beta-RS}
    \lim_{D\to\infty}\beta_{\RS}(k, D) \sqrt{\frac{D}{k}} = \beta_\RS(k).
\end{align}

\subsection{The Dynamical Threshold $\beta_{\dyn}(k,D)$  at Finite Connectivity}
\label{app:beta-dyn}
We now consider the finite-connectivity ensemble of Max-$k$-XORSAT on hypergraphs with average degree $D$, as in \cref{def:erdos_renyi}. 
In spin glass theory, a precise prediction of $\beta_\dyn(k,D)$ at any finite $(k, D)$  is determined by when there exists a nontrivial solution to the so-called ``1RSB equation''~\cite{Montanari2003}. According to \cite{mezard2006reconstruction}, the 1RSB equation can be mapped to a reconstruction problem on the infinite tree, which yields a more transparent method for determining $\beta_\dyn$. We will adopt this as our definition of $\beta_\dyn(k, D)$ below.

\paragraph{Reconstruction-threshold definition for $\beta_\dyn(k,D)$.}

Let $\tilde m^{(0)} = 1$ deterministically, and for $\ell \ge 0$ let
\begin{equation}\label{eq:xor-de}
    \tilde m^{(\ell+1)} \;\overset{d}{=}\; \tanh\Bigg( \sum_{a=1}^{d} \tanh^{-1}\Big( \theta\, S_a \prod_{j=1}^{k-1} \tilde m^{(\ell)}_{a,j} \Big) \Bigg),
\end{equation}
where $d\sim\Poisson(D)$ for Erd\H{o}s-R\'enyi hypergraphs (or fixed $d=D-1$ for $D$-regular hypergraphs), the $S_a$ are i.i.d. $\pm1$ signs with mean $\theta=\tanh\beta$, and the $\tilde m^{(\ell)}_{a,j}$ are i.i.d.\ copies of $\tilde m^{(\ell)}$.
Then let
\begin{equation}
    \rho_\ell := \EE\big[\tilde m^{(\ell)}\big]
    \qquad \text{and} \qquad
    \rho_\infty(\beta;k,D) := \lim_{\ell\to\infty} \rho_\ell.
\end{equation}
Following the reconstruction interpretation of the 1RSB equation~\cite{mezard2006reconstruction},  the dynamical phase transition temperature can be defined via
\begin{equation}\label{eq:xor-beta-dyn-def}
    \beta_\dyn(k,D) := \inf\{\beta>0:\ \rho_\infty(\beta;k,D)>0\}.
\end{equation}

\paragraph{Planted density evolution algorithm for $\beta_\dyn(k,D)$.} In practice, we can determine $\beta_\dyn(k, D)$ via the following planted density evolution algorithm:

\begin{enumerate}
    \item Initialize a pool of $\tilde{M}$ samples at the planted value, i.e. set $\{\tilde m^{(0)}_\mu = 1: \mu = 1,2,\ldots,\tilde{M}\}$.
    \item Replace the pool by $\tilde{M}$ fresh samples of \eqref{eq:xor-de}, drawing offspring $d\sim\Poisson(D)$ for Erd\H{o}s-R\'enyi (or fixed $d=D-1$ for regular ensembles), clause noises $S_a$, and children messages uniformly from the current pool.
    \item Track the order parameter $\rho_\ell = \EE[\tilde m^{(\ell)}]$. Declare the trivial fixed point reached when the order parameter drops below a small floor (e.g., $10^{-4}$), and a nontrivial fixed point when it stabilizes above $5\times10^{-3}$ with negligible drift. We then locate the threshold by bisection.
\end{enumerate}

We numerically evaluate $\beta_{\dyn}(k,D)$ using the above algorithm. The results are shown in \cref{tab:xorsat-grid} for the Erd\H{o}s-R\'enyi ensemble, and in \cref{tab:xorsat-grid-reg} for the random regular graph ensemble.

\begin{table}[ht]
\centering
\begin{tabular}{c|cccccccc}
$D_d(k)$ & 2.455 & 3.089 & 3.509 & 3.823 & 4.072 & 4.280 & 4.457 & 4.612\\
\hline
\hline
\diagbox{$D$}{$k$}
& 3 & 4 & 5 & 6 & 7 & 8 & 9 & 10\\
\hline
3 & 1.754(4) & --- & --- & --- & --- & --- & --- & ---\\
4 & 1.225(1) & 1.656(4) & 2.093(12) & 2.826(12) & --- & --- & --- & ---\\
5 & 1.006(1) & 1.275(1) & 1.489(1) & 1.681(2) & 1.850(10) & 2.034(18) & 2.244(10) & 2.479(14)\\
6 & 0.875(1) & 1.082(1) & 1.233(3) & 1.349(3) & 1.456(4) & 1.550(5) & 1.657(4) & 1.700(25)\\
7 & 0.788(2) & 0.967(9) & 1.080(3) & 1.170(4) & 1.241(5) & 1.316(2) & 1.373(5) & 1.411(11)\\
8 & 0.712(4) & 0.877(6) & 0.977(7) & 1.043(4) & 1.113(2) & 1.164(2) & 1.195(10) & 1.253(4)\\
9 & 0.667(2) & 0.802(1) & 0.892(13) & 0.954(2) & 1.009(3) & 1.050(5) & 1.093(4) & 1.128(5)\\
10 & 0.623(2) & 0.756(10) & 0.838(6) & 0.882(3) & 0.930(7) & 0.971(4) & 1.008(3) & 1.037(5)\\
11 & 0.588(3) & 0.706(7) & 0.775(2) & 0.830(2) & 0.871(3) & 0.909(3) & 0.938(4) & 0.964(4)\\
12 & 0.557(3) & 0.665(3) & 0.731(3) & 0.781(4) & 0.822(2) & 0.854(3) & 0.899(8) & 0.907(3)\\
13 & 0.530(6) & 0.646(11) & 0.698(2) & 0.735(6) & 0.783(2) & 0.810(2) & 0.846(7) & 0.841(18)\\
14 & 0.512(3) & 0.605(3) & 0.668(3) & 0.719(6) & 0.746(8) & 0.774(3) & 0.794(4) & 0.809(7)\\
15 & 0.492(2) & 0.582(3) & 0.638(2) & 0.677(3) & 0.703(11) & 0.737(3) & 0.759(3) & 0.782(3)\\
16 & 0.483(4) & 0.562(4) & 0.615(4) & 0.654(2) & 0.684(2) & 0.709(2) & 0.731(2) & 0.742(8)\\
17 & 0.461(2) & 0.540(5) & 0.595(4) & 0.631(6) & 0.659(4) & 0.683(2) & 0.703(2) & 0.690(28)\\
18 & 0.452(7) & 0.524(2) & 0.576(2) & 0.608(2) & 0.639(3) & 0.656(4) & 0.677(3) & 0.696(3)\\
19 & 0.432(2) & 0.501(8) & 0.564(9) & 0.593(3) & 0.612(4) & 0.636(3) & 0.682(21) & 0.678(3)\\
20 & 0.430(6) & 0.495(2) & 0.538(4) & 0.570(4) & 0.600(2) & 0.613(6) & 0.636(3) & 0.650(4)
\end{tabular}
\caption{Estimated $\beta_\dyn(k,D)$ for Max-$k$-XORSAT on the Erd\H{o}s--R\'enyi-like ensemble (where degree distribution is Poisson$(D)$), using pool size $\tilde{M}=2\times 10^5$. 
Digits in parentheses are $1\sigma$ confidence interval in units of the last quoted digit.
``---'' marks $D<D_d(k)=\inf\{ D : \exists\, q\in(0,1],\ q = 1-e^{-Dq^{k-1}} \}$, where no dynamical transition exists ($\beta_\dyn=\infty$). The first header row lists the values of the critical $D_d(k)$. The results at $k=3$ are consistent with \cite{Franz_2001}.
}
\label{tab:xorsat-grid}
\end{table}

\begin{table}[ht]
\centering
\begin{tabular}{c|cccccccc}
\diagbox{$D$}{$k$}& 3 & 4 & 5 & 6 & 7 & 8 & 9 & 10\\
\hline
3 & 1.961(4) & 2.429(5) & 2.742(13) & 2.987(9) & 3.171(18) & 3.328(30) & 3.452(47) & 3.605(28)\\
4 & 1.327(2) & 1.597(1) & 1.771(3) & 1.909(5) & 2.003(7) & 2.094(5) & 2.159(9) & 2.219(14)\\
5 & 1.080(6) & 1.272(4) & 1.407(16) & 1.482(14) & 1.574(6) & 1.639(3) & 1.698(6) & 1.725(19)\\
6 & 0.920(1) & 1.088(1) & 1.193(2) & 1.270(9) & 1.330(3) & 1.381(3) & 1.423(4) & 1.459(5)\\
7 & 0.821(3) & 0.967(1) & 1.056(2) & 1.118(3) & 1.162(11) & 1.215(4) & 1.257(5) & 1.275(10)\\
8 & 0.751(4) & 0.877(2) & 0.958(2) & 1.018(1) & 1.064(2) & 1.101(3) & 1.130(12) & 1.171(8)\\
9 & 0.688(2) & 0.813(4) & 0.880(2) & 0.938(2) & 0.973(4) & 1.006(4) & 1.036(4) & 1.061(3)\\
10 & 0.641(2) & 0.756(3) & 0.816(3) & 0.870(3) & 0.905(4) & 0.941(3) & 0.967(2) & 0.982(6)\\
11 & 0.605(2) & 0.710(3) & 0.772(2) & 0.817(2) & 0.853(2) & 0.887(2) & 0.886(14) & 0.906(18)\\
12 & 0.573(1) & 0.672(1) & 0.730(2) & 0.770(3) & 0.798(7) & 0.833(4) & 0.856(3) & 0.875(3)\\
13 & 0.543(5) & 0.636(3) & 0.684(6) & 0.738(3) & 0.769(5) & 0.795(2) & 0.814(2) & 0.831(3)\\
14 & 0.523(2) & 0.611(2) & 0.663(3) & 0.704(2) & 0.729(4) & 0.751(6) & 0.779(2) & 0.794(3)\\
15 & 0.509(5) & 0.587(3) & 0.633(4) & 0.673(3) & 0.712(12) & 0.727(2) & 0.744(5) & 0.765(3)\\
16 & 0.481(3) & 0.562(3) & 0.612(3) & 0.648(3) & 0.677(3) & 0.695(4) & 0.714(4) & 0.731(4)\\
17 & 0.476(7) & 0.540(3) & 0.596(2) & 0.627(2) & 0.652(2) & 0.672(3) & 0.690(4) & 0.709(3)\\
18 & 0.448(5) & 0.527(4) & 0.568(5) & 0.596(5) & 0.632(6) & 0.646(6) & 0.665(3) & 0.684(3)\\
19 & 0.439(3) & 0.509(2) & 0.557(2) & 0.582(4) & 0.617(9) & 0.626(4) & 0.634(15) & 0.639(13)\\
20 & 0.423(3) & 0.505(5) & 0.533(5) & 0.567(4) & 0.595(2) & 0.611(3) & 0.628(3) & 0.639(5)
\end{tabular}
\caption{Estimated $\beta_\dyn(k,D)$ for Max-$k$-XORSAT on the random $D$-regular ensemble (every variable in exactly $D$ clauses). Same method as \cref{tab:xorsat-grid}. We remark that $\beta_\dyn(k,2)=\infty$.}
\label{tab:xorsat-grid-reg}
\end{table}

\clearpage

\paragraph{High-connectivity limit reduces to the dense $k$-spin threshold.}
Consider a fixed point of \eqref{eq:xor-de}.  Let $m$ have the fixed-point message law and set $q=\E[m]$.  Since $m$ is a posterior magnetization under the planted law, the Nishimori identity gives
\begin{equation}
\label{eq:nishimori-message-identity}
    \E[m^2]=\E[m]=q.
\end{equation}
To see this directly, suppose $X\in\{\pm1\}$ is the planted spin and let $m_0=\E[X\mid\mathcal F]$, where $\mathcal F$ is the information received from the descendants. Then $\E[Xm_0]=\E[m_0^2]$ by conditioning on $\mathcal F$. Now $m=Xm_0$, so $\E[m]=\E[m^2]$.

For one clause incident to the root, let
\[
    P=\prod_{j=1}^{k-1}m_j,
    \qquad
    u=\tanh^{-1}(\theta S P),
    \qquad \theta=\tanh\beta,
\]
where the $m_j$ are independent copies of $m$, and $S\in\{\pm1\}$ is independent with $\E S=\theta$.
As $\theta\to0$, the Taylor expansion of $\tanh^{-1}$ and \eqref{eq:nishimori-message-identity} give
\begin{align}
\E u
&=\theta^2 q^{k-1}+O(\theta^4),
\nonumber
\\
\E u^2
&=\theta^2 q^{k-1}+O(\theta^4).
\label{eq:cavity-field}
\end{align}
Indeed, $\tanh^{-1}(\theta SP)=S\tanh^{-1}(\theta P)$, so the extra factor $\E S=\theta$ makes the error in the first line $O(\theta^4)$ rather than $O(\theta^3)$.  Also, independence and \eqref{eq:nishimori-message-identity} give $\E P=\E P^2=q^{k-1}$.

The quantity inside the outer $\tanh$ in \eqref{eq:xor-de}, sometimes referred to as the cavity field, is a sum of $d\sim\Poisson(D)$ independent copies of $u$.  Define
\begin{equation}
\label{eq:lambda-cavity-definition}
    \lambda^2=D\theta^2q^{k-1}.
\end{equation}
Eq.~\eqref{eq:cavity-field} shows that this field has mean and variance $\lambda^2+O(D\theta^4)$.  Its higher standardized cumulants vanish as $D\to\infty$.  Consequently, when $D\theta^4=o(1)$, the density evolution replaces this field, to leading order, by $\lambda^2+\lambda G$ with $G\sim\cN(0,1)$.  The fixed-point equation becomes
\begin{equation}
\label{eq:gaussian-cavity-fixed-point}
    q=\EE_G\tanh(\lambda^2+\lambda G)+o(1).
\end{equation}
Gaussian change of measure gives
\begin{align}
\EE_G\tanh(\lambda^2+\lambda G)
=e^{-\lambda^2/2}\EE_G \left[e^{\lambda G}\tanh(\lambda G)\right] =\frac{\EE_G[\sinh(\lambda G)\tanh(\lambda G)]}{e^{\lambda^2/2}},
\label{eq:gaussian-change-of-measure-dyn}
\end{align}
because the contribution $\E[\cosh(\lambda G)\tanh(\lambda G)]$ vanishes by symmetry.  If
\[
    \widehat\beta=\theta\sqrt{\frac{D}{k}},
\]
then \eqref{eq:lambda-cavity-definition} can be read as $\lambda=\widehat\beta\sqrt{kq^{k-1}}$.  Hence \eqref{eq:gaussian-cavity-fixed-point}--\eqref{eq:gaussian-change-of-measure-dyn} reproduce \eqref{eq:beta-dyn}, and $\theta=\beta[1+o(1)]$ in the regime considered here.  At the level of a fixed-point equation, this gives the relation
\begin{equation}
\label{eq:diluted-gaussian-reduction-summary}
    \beta_\dyn(k,D)
    =\sqrt{\frac{k}{D}}\,\beta_\dyn(k)\,[1+o(1)]
\end{equation}
within the reconstruction calculation.  If the appearance of the nonzero fixed point is stable under this limit, then for fixed $k$ it yields \eqref{eq:diluted-to-gaussian-beta-dyn}.

The remaining input is the known asymptotic
\[
    \beta_\dyn(k)
    =\sqrt{\frac{2\ln k}{k}}\,[1+o_{k}(1)]
\]
for the dense $k$-spin model defined in \eqref{eq:beta-dyn}; see \cite{Ferrari2012two, alaoui2023sampling}.  Substitution into \eqref{eq:diluted-gaussian-reduction-summary} gives \begin{equation}
\label{eq:beta-dyn-closed-form-prediction}
    \beta_\dyn(k,D)
    =\sqrt{\frac{2\ln k}{D}}\,[1+o(1)].
\end{equation}

At the predicted threshold, $\theta^2\simeq  \beta^2 \asymp(\ln k)/D$, so the Taylor error satisfies
\[
    D\theta^4\asymp\frac{(\ln k)^2}{D}
    \leq\frac{(\ln k)^2}{k}\longrightarrow0
    \qquad\text{as }k\to\infty,\quad D\geq k.
\]
Thus this error vanishes throughout $D\geq k$, including the regime $D=\alpha k$ with fixed $\alpha\geq1$.
The same Taylor expansion controls the logarithmic moment-generating function of the cavity field with error $O(D\theta^4)$ for bounded arguments, suggesting that the Gaussian approximation has a correction that vanishes as $k\to\infty$. Heuristically, assuming the nonzero fixed point is stable under perturbation by this correction, then the corresponding relative correction to $\beta_\dyn(k,D)$ also vanishes. Combining this with the dense large-$k$ asymptotic and $\theta=\beta[1+O((\ln k)/D)]$ gives the predicted $o_k(1)$ error in \eqref{eq:beta-dyn-asymp} uniformly for $D\ge k$.

\subsection{The Static Threshold $\beta_{\RS}(k, D)$  at Finite Connectivity}
\label{sec:beta-RS}

We next characterize the static threshold for the finite-connectivity ensemble in \cref{def:erdos_renyi}.

We begin by providing a variational description of $\beta_\RS(k,D)$, which is due to \cite{coja2018charting}.
Let $\theta=\tanh\beta$ and $\Lambda(x)=x\ln x$.  Let $\pi$ be a candidate distribution of a random variable $X\in[-1,1]$ with $\E X=0$.  Let $d\sim\Poisson(D)$, and let all variables $J_0,J_1,\ldots$, $X_1,\ldots$, and $X_{a,j}$ below be mutually independent, with the $J$'s being Rademacher variables (i.e., uniformly $\pm 1$) and the $X$'s distributed according to $\pi$.  Define
\begin{align}
\label{eq:sparse-static-functional}
\mathcal B_{k,D,\beta}(\pi)
:={}&\E\Bigg[\frac12\Lambda\Bigg(
 \sum_{\sigma\in\{\pm1\}}
 \prod_{a=1}^{d}
 \left(1+\theta J_a\sigma\prod_{j=1}^{k-1}X_{a,j}\right)
 \Bigg) -\frac{D(k-1)}{k}\,
 \Lambda\left(1+\theta J_0\prod_{j=1}^{k}X_j\right)\Bigg].
\end{align}
Applying the condensation theorem for random factor graphs
\cite[Theorem~2.2 and Lemma~2.15]{coja2018charting} to the factor weights
\[
    \psi_J(\sigma_1,\ldots,\sigma_k)
    =1+\theta J\prod_i\sigma_i \,,
\]
we get the following characterization
\begin{equation}
\label{eq:sparse-beta-RS-variational}
 \beta_\RS(k,D)
 :=\inf\left\{\beta>0:\
 \sup_{\pi:\,\E X=0}\mathcal B_{k,D,\beta}(\pi)>\ln2\right\}.
\end{equation}
This result uses only the symmetry of the Rademacher coupling and is valid for both even and odd $k$.  The cited theorem uses a Poisson number of ordered clauses.  It applies to our exact ensemble after two standard comparisons: conditioning the clause count to be $M=\lfloor ND/k\rfloor$ changes the expected free energy by $o(N)$, and replacing the $o(N)$ clauses that repeat a vertex within an ordered clause by uniform $k$-subsets also changes it by $o(N)$.
Finally,
$e^{\beta J\prod_i\sigma_i}=\cosh\beta\,\psi_J(\sigma_1,\ldots,\sigma_k)$, so the normalization of the factor weights does not change the quenched vs. annealed free-energy gap.

\paragraph{Comparing $\beta_{\rm Holevo}$ with $\beta_\RS$.}
Let us now compare $\beta_\RS(k, D)$ to $\beta_{\rm Holevo}(k, D)$ given in \eqref{eq:beta-Holevo} and \eqref{eq:beta-Holevo-app}.  
Let
\begin{equation}
\tau(\beta):=\beta\tanh\beta - \ln \cosh\beta,
\end{equation} 
and $\beta_{\rm ann}(k,D)$ be the unique positive solution of
\begin{equation}
\label{eq:annealed-entropy-zero}
 \frac{D}{k}\big(\beta\tanh\beta-\ln\cosh\beta\big)=\ln2.
\end{equation}
For $D> k$, the existence and uniqueness of the solution follows from the fact that $\tau'(\beta)=\beta \sech^2\beta >0$, $\tau(0)=0$, $\tau(\infty)=\ln 2$.

\begin{proposition}[Static versus Holevo threshold]
\label{prop:beta-RS-Holevo}
For every $k\geq2$ and $D>k$,
\begin{equation}
\label{eq:beta-RS-Holevo-order}
 \beta_\RS(k,D)<\beta_{\rm ann}(k,D)
 =\beta_{\rm Holevo}(k,D),
\end{equation}
where $\beta_{\rm Holevo}$ is defined in \eqref{eq:beta-Holevo-app}.
\end{proposition}

\begin{proof}
From the calculation in \cref{sec:RS}, we see that the annealed free energy (per spin) for the Erd\H{o}s-R\'enyi ensemble in \cref{def:erdos_renyi} is 
\[
 a(\beta)=\ln2+\frac{D}{k}\ln\cosh\beta.
\]
Let $f_N(\beta) = \frac{1}{N}\EE \ln Z_{\beta, H}$ be the quenched free energy. By definition of $\beta_\RS$, the quenched and annealed free energies agree throughout $\beta\in(0,\beta_\RS(k, D))$. In this interval, the convex function $f_N(\beta)$ converges to the differentiable function $a(\beta)$, implying $f'_N(\beta)\to a'(\beta)$.
So for $\beta  < \beta_\RS$, the expected Gibbs entropy per spin thus converge to
\[
\lim_{N\to\infty} (f_N(\beta)-\beta f'_N(\beta)) = a(\beta)-\beta a'(\beta)
% =\ln2-\frac{D}{k}\big(\beta\tanh\beta-\ln\cosh\beta\big)
=\ln2-\frac{D}{k}\tau(\beta).
\]
Note $\beta_{\rm ann}(k, D)$ is precisely the point where this entropy vanishes.

We next use the existence of isolated vertices in a random Erd\H{o}s-R\'enyi hypergraph to show $\beta_\RS(k, D) < \beta_{\rm ann}(k, D)$. In such a random hypergraph, the expected number of isolated vertices is $N(1-k/N)^M = N(e^{-D} + o_N(1))$. Since an isolated vertice does not appear in the Hamiltonian, its spin is independent and uniformly random under the Gibbs measure at all temperatures. Thus the Gibbs entropy per spin at any temperature is at least $e^{-D}\ln 2$, and so for all $\beta \in (0, \beta_\RS)$ we must have
\begin{equation}
\ln2-\frac{D}{k}\tau(\beta) \ge e^{- D}\ln 2.
\end{equation}
Observe that there is a unique $\beta_{\rm iso} > 0$ such that $\frac{D}{k} \tau(\beta_{\rm iso}) = (1-e^{-D})\ln 2$. Monotonicity of $\tau(\beta)$ then implies that $\beta_\RS \le \beta_{\rm iso} < \beta_{\rm ann}$.

It remains to show that $\beta_{\rm ann} = \beta_{\rm Holevo}$.  If
$p=H_2^{-1}(1-k/D)\in[0,1/2]$ and
$\beta=\sech^{-1}(2\sqrt{p(1-p)})$, then
$p=(1-\tanh\beta)/2$.  The elementary entropy identity
\[
 H_2\left(\frac{1-\tanh\beta}{2}\right) \cdot \ln2
 =\ln2+\ln\cosh\beta-\beta\tanh\beta
\]
shows that $H_2(p)=1-k/D$ is equivalent to \eqref{eq:annealed-entropy-zero}.  This is exactly the formula for $\beta_{\rm Holevo}$ in \cref{app:holevo_shannon}.
\end{proof}

Thus the ideal Holevo inverse temperature is strictly larger than the static transition inverse temperature of the Erd\H{o}s-R\'enyi ensemble for every finite $k\geq2$ and $D>k$.  To determine the size of the gap at any finite $(k, D)$, however, we would need to evaluate the variational problem in \eqref{eq:sparse-beta-RS-variational}. Furthermore, one should also account for the caveat that $\beta_{\rm Holevo}$ is an optimistic upper bound that assumes the hypergraph's incidence matrix $B$ has full rank, which is not true for a typical Erd\H{o}s-R\'enyi instance. We leave an ensemble-dependent characterization of $\beta_{\rm Holevo}$ and a more precise comparison to $\beta_\RS$ as future work.

\bibliographystyle{alpha}
\bibliography{refs}

\newcommand{\etalchar}[1]{$^{#1}$}
\begin{thebibliography}{FMRT{\etalchar{+}}01}

\bibitem[ACO08]{achlioptas2008algorithmic}
Dimitris Achlioptas and Amin Coja-Oghlan.
\newblock Algorithmic barriers from phase transitions.
\newblock In {\em 49th Annual IEEE Symposium on Foundations of Computer
  Science}, pages 793--802. IEEE, 2008.

\bibitem[AGL25]{anschuetz2025decoded}
Eric~R. Anschuetz, David Gamarnik, and Jonathan~Z. Lu.
\newblock {Spin Glass Transitions Obstruct Decoded Quantum Interferometry}.
\newblock {\em arXiv preprint arXiv:2509.14509}, 2025.

\bibitem[Ala26]{alaoui2024near}
Ahmed~El Alaoui.
\newblock {Near-optimal shattering in the Ising pure $p$-spin and rarity of
  solutions returned by stable algorithms}.
\newblock {\em Stochastic Processes and their Applications}, 192:104792, 2026.

\bibitem[AM23]{anshu2023concentration}
Anurag Anshu and Tony Metger.
\newblock Concentration bounds for quantum states and limitations on the qaoa
  from polynomial approximations.
\newblock {\em Quantum}, 7:999, 2023.

\bibitem[AMS21]{AMS21optimization}
Ahmed~El Alaoui, Andrea Montanari, and Mark Sellke.
\newblock {Optimization of mean-field spin glasses}.
\newblock {\em Ann. Probab.}, 49(6):2922--2960, 2021.

\bibitem[AMS22]{AMS22sampling}
Ahmed~El Alaoui, Andrea Montanari, and Mark Sellke.
\newblock Sampling from the {Sherrington-Kirkpatrick Gibbs} measure via
  algorithmic stochastic localization.
\newblock In {\em IEEE 63rd Annual Symposium on Foundations of Computer Science
  (FOCS)}, pages 323--334, 2022.

\bibitem[AMS25a]{alaoui2023sampling}
Ahmed~El Alaoui, Andrea Montanari, and Mark Sellke.
\newblock Sampling from mean-field {Gibbs} measures via diffusion processes.
\newblock {\em Probability and Mathematical Physics}, 6(3):961--1022, 2025.

\bibitem[AMS25b]{alaoui2023shattering}
Ahmed~El Alaoui, Andrea Montanari, and Mark Sellke.
\newblock Shattering in pure spherical spin glasses.
\newblock {\em Communications in Mathematical Physics}, 406(5):1--36, 2025.

\bibitem[AR05]{AR05}
Dorit Aharonov and Oded Regev.
\newblock Lattice problems in {NP $\cap$ coNP}.
\newblock {\em Journal of the ACM (JACM)}, 52(5):749--765, 2005.

\bibitem[ATS03]{ATS03}
Dorit Aharonov and Amnon Ta-Shma.
\newblock Adiabatic quantum state generation and statistical zero knowledge.
\newblock In {\em Proceedings of the 35th annual ACM symposium on Theory of
  computing}, pages 20--29, 2003.

\bibitem[BGMZ22]{basso2022performance}
Joao Basso, David Gamarnik, Song Mei, and Leo Zhou.
\newblock Performance and limitations of the qaoa at constant levels on large
  sparse hypergraphs and spin glass models.
\newblock In {\em IEEE 63rd Annual Symposium on Foundations of Computer Science
  (FOCS)}, page 335–343. IEEE, 2022.

\bibitem[BH22]{bresler2022kSAT}
Guy Bresler and Brice Huang.
\newblock {The Algorithmic Phase Transition of Random k-SAT for Low Degree
  Polynomials}.
\newblock In {\em IEEE 62nd Annual Symposium on Foundations of Computer Science
  (FOCS)}, pages 298--309, 2022.

\bibitem[BHJK25]{buhai2025quasi}
Rares-Darius Buhai, Jun-Ting Hsieh, Aayush Jain, and Pravesh~K Kothari.
\newblock The quasi-polynomial low-degree conjecture is false.
\newblock In {\em 2025 IEEE 66th Annual Symposium on Foundations of Computer
  Science (FOCS)}, pages 2577--2590. IEEE, 2025.

\bibitem[Cel24]{Celentano2024}
Michael Celentano.
\newblock {Sudakov–Fernique post-AMP, and a new proof of the local convexity
  of the TAP free energy}.
\newblock {\em Ann. Probab.}, 52(3), 2024.

\bibitem[CGPR19]{chen2019suboptimality}
Wei-Kuo Chen, David Gamarnik, Dmitry Panchenko, and Mustazee Rahman.
\newblock Suboptimality of local algorithms for a class of max-cut problems.
\newblock {\em Ann. Probab.}, 47(3):1587--1618, 2019.

\bibitem[CHM23]{chen2023local}
Antares Chen, Neng Huang, and Kunal Marwaha.
\newblock {Local algorithms and the failure of log-depth quantum advantage on
  sparse random CSPs}.
\newblock {\em arXiv preprint arXiv:2310.01563}, 2023.

\bibitem[CLSS22]{chou2022limitations}
Chi-Ning Chou, Peter~J. Love, Juspreet~Singh Sandhu, and Jonathan Shi.
\newblock {Limitations of Local Quantum Algorithms on Random MAX-k-XOR and
  Beyond}.
\newblock In {\em 49th International Colloquium on Automata, Languages, and
  Programming (ICALP 2022)}, volume 229, pages 41:1--41:20, 2022.

\bibitem[COEJ{\etalchar{+}}18]{coja2018charting}
Amin Coja-Oghlan, Charilaos Efthymiou, Nor Jaafari, Mihyun Kang, and Tobias
  Kapetanopoulos.
\newblock Charting the replica symmetric phase.
\newblock {\em Communications in Mathematical Physics}, 359(2):603--698, 2018.

\bibitem[COKKR24]{coja2024k}
Amin Coja-Oghlan, Mihyun Kang, Lena Krieg, and Maurice Rolvien.
\newblock The {$k$-XORSAT} threshold revisited.
\newblock {\em The Electronic Journal of Combinatorics}, pages P2--16, 2024.

\bibitem[CT23]{chailloux2023quantum}
André Chailloux and Jean-Pierre Tillich.
\newblock The quantum decoding problem.
\newblock {\em arXiv preprint arXiv:2310.20651}, 2023.

\bibitem[DGM{\etalchar{+}}10]{Dietzfelbinger2010cuckoo}
Martin Dietzfelbinger, Andreas Goerdt, Michael Mitzenmacher, Andrea Montanari,
  Rasmus Pagh, and Michael Rink.
\newblock Tight thresholds for cuckoo hashing via {XORSAT}.
\newblock In {\em International Colloquium on Automata, Languages, and
  Programming}, pages 213--225. Springer, 2010.

\bibitem[DK22]{diakonikolas2022non}
Ilias Diakonikolas and Daniel Kane.
\newblock {Non-Gaussian} component analysis via lattice basis reduction.
\newblock In {\em Conference on Learning Theory}, pages 4535--4547. PMLR, 2022.

\bibitem[DM02]{dubois20023}
Olivier Dubois and Jacques Mandler.
\newblock The {3-XORSAT} threshold.
\newblock In {\em Proceedings of the 43rd Symposium on Foundations of Computer
  Science}, pages 769--778, 2002.

\bibitem[DM10]{dembo2009gibbs}
Amir Dembo and Andrea Montanari.
\newblock Gibbs measures and phase transitions on sparse random graphs.
\newblock {\em Brazilian Journal of Probability and Statistics}, pages
  137--211, 2010.

\bibitem[FGG20]{farhi2020typical}
Edward Farhi, David Gamarnik, and Sam Gutmann.
\newblock {The Quantum Approximate Optimization Algorithm Needs to See the
  Whole Graph: A Typical Case}.
\newblock {\em arXiv preprint arXiv:2004.09002}, 2020.

\bibitem[FJ25]{FJ24}
Edward Farhi and Stephen~P. Jordan.
\newblock Efficiently constructing a quantum uniform superposition over bit
  strings near a binary linear code.
\newblock {\em Quantum Information and Computation}, 24(15/16):1326--1355,
  2025.
\newblock arXiv:2404.16129.

\bibitem[FLPR12]{Ferrari2012two}
U.~Ferrari, L.~Leuzzi, G.~Parisi, and T.~Rizzo.
\newblock {Two-step relaxation next to dynamic arrest in mean-field glasses:
  Spherical and Ising $p$-spin model}.
\newblock {\em Phys. Rev. B}, 86:014204, 2012.

\bibitem[FMRT{\etalchar{+}}01]{Franz_2001}
S.~Franz, M.~Mézard, F.~Ricci-Tersenghi, M.~Weigt, and R.~Zecchina.
\newblock A ferromagnet with a glass transition.
\newblock {\em Europhysics Letters (EPL)}, 55(4):465--471, 2001.

\bibitem[Fri90]{frieze1990independence}
Alan~M. Frieze.
\newblock On the independence number of random graphs.
\newblock {\em Discrete Mathematics}, 81(2):171--175, 1990.

\bibitem[Gam21]{gamarnik2021overlap}
David Gamarnik.
\newblock The overlap gap property: A topological barrier to optimizing over
  random structures.
\newblock {\em Proceedings of the National Academy of Sciences},
  118(41):e2108492118, 2021.

\bibitem[GC26]{garrido2026regev}
M.~Garrido and Andr{\'e} Chailloux.
\newblock Regev's reduction as a candidate quantum algorithm for the discrete
  logarithm problem in finite abelian groups.
\newblock {\em arXiv preprint arXiv:2605.03972}, 2026.

\bibitem[GJK25]{gamarnik2023shattering}
David Gamarnik, Aukosh Jagannath, and Eren~C. Kızıldağ.
\newblock {Shattering in the Ising p-spin glass model}.
\newblock {\em Probability Theory and Related Fields}, 193(1-2):89–141, 2025.

\bibitem[GJW20]{gamarnik2020low}
David Gamarnik, Aukosh Jagannath, and Alexander~S. Wein.
\newblock Low-degree hardness of random optimization problems.
\newblock In {\em IEEE 61st Annual Symposium on Foundations of Computer Science
  (FOCS)}, pages 131--140, 2020.

\bibitem[GJW24]{gamarnik2024low}
David Gamarnik, Aukosh Jagannath, and Alexander~S. Wein.
\newblock Hardness of random optimization problems for boolean circuits,
  low-degree polynomials, and langevin dynamics.
\newblock {\em SIAM Journal on Computing}, 53(1):1–46, 2024.

\bibitem[GKPX22]{GKPX2022}
David Gamarnik, Eren~C. Kızıldağ, Will Perkins, and Changji Xu.
\newblock Algorithms and barriers in the symmetric binary perceptron model.
\newblock In {\em IEEE 63rd Annual Symposium on Foundations of Computer Science
  (FOCS)}, pages 576--587, 2022.

\bibitem[GS14]{gamarnik2014limits}
David Gamarnik and Madhu Sudan.
\newblock Limits of local algorithms over sparse random graphs.
\newblock In {\em Proceedings of the 5th Conference on Innovations in
  Theoretical Computer Science}, ITCS '14, page 369–376, New York, NY, USA,
  2014. Association for Computing Machinery.

\bibitem[Hop18]{hopkins2018statistical}
Samuel Hopkins.
\newblock {\em Statistical inference and the sum of squares method}.
\newblock PhD thesis, Cornell University, 2018.

\bibitem[HPP25]{huang2025hardness}
Neng Huang, Will Perkins, and Aaron Potechin.
\newblock Hardness of sampling for the anti-ferromagnetic {I}sing model on
  random graphs.
\newblock In {\em 16th Innovations in Theoretical Computer Science Conference
  (ITCS 2025)}, pages 61--1, 2025.

\bibitem[HS22]{huang2022tight}
Brice Huang and Mark Sellke.
\newblock {Tight Lipschitz Hardness for optimizing Mean Field Spin Glasses}.
\newblock In {\em IEEE 63rd Annual Symposium on Foundations of Computer Science
  (FOCS)}, page 312–322. IEEE, 2022.

\bibitem[HS26]{bogpinprogress}
Brice Huang and Mark Sellke.
\newblock {Algorithmic threshold for high-dimensional projection pursuit II:
  asymptotics in constraint satisfaction problems}.
\newblock {\em In Preparation}, 2026.

\bibitem[HW21]{holmgren_wein}
Justin Holmgren and Alexander~S. Wein.
\newblock {Counterexamples to the Low-Degree Conjecture}.
\newblock In {\em 12th Innovations in Theoretical Computer Science Conference
  (ITCS 2021)}, volume 185 of {\em Leibniz International Proceedings in
  Informatics (LIPIcs)}, pages 75:1--75:9, 2021.

\bibitem[IKKM15]{Ibrahimi_2015}
Morteza Ibrahimi, Yash Kanoria, Matt Kraning, and Andrea Montanari.
\newblock The set of solutions of random {XORSAT} formulae.
\newblock {\em The Annals of Applied Probability}, 25(5), 2015.

\bibitem[JSW{\etalchar{+}}25]{jordan2024optimization}
Stephen~P. Jordan, Noah Shutty, Mary Wootters, Adam Zalcman, Alexander
  Schmidhuber, Robbie King, Sergei~V. Isakov, Tanuj Khattar, and Ryan Babbush.
\newblock Optimization by decoded quantum interferometry.
\newblock {\em Nature}, 646(8086):831--836, 2025.

\bibitem[KS26]{koehler2026overlap}
Frederic Koehler and Joonhyung Shin.
\newblock Overlap analysis of the shortest path problem: Local search,
  landscapes, and franz-parisi potential.
\newblock In Steve Hanneke and Tor Lattimore, editors, {\em Proceedings of
  Thirty Ninth Conference on Learning Theory}, volume 336 of {\em Proceedings
  of Machine Learning Research}, pages 4050--4228. PMLR, 2026.

\bibitem[KWB19]{kunisky2019notes}
Dmitriy Kunisky, Alexander~S. Wein, and Alfonso~S. Bandeira.
\newblock Notes on computational hardness of hypothesis testing: Predictions
  using the low-degree likelihood ratio.
\newblock In {\em ISAAC Congress (International Society for Analysis, its
  Applications and Computation)}, pages 1--50. Springer, 2019.

\bibitem[Lop26]{lopatto2026replica}
Patrick Lopatto.
\newblock {Full replica symmetry breaking in the Sherrington-Kirkpatrick
  model}.
\newblock {\em arXiv preprint arXiv:2607.11756}, 2026.

\bibitem[LS25]{li2024some}
Shuangping Li and Tselil Schramm.
\newblock Some easy optimization problems have the overlap-gap property.
\newblock 291:3582--3622, 2025.

\bibitem[Lyo17]{lyons2017factors}
Russell Lyons.
\newblock {Factors of IID on trees}.
\newblock {\em Combinatorics, Probability and Computing}, 26(2):285--300, 2017.

\bibitem[Mao26]{mao2026polynomial}
Songtao Mao.
\newblock The polynomial-time low-degree conjecture is false.
\newblock {\em arXiv preprint arXiv:2607.20318}, 2026.

\bibitem[MM06]{mezard2006reconstruction}
Marc Mézard and Andrea Montanari.
\newblock Reconstruction on trees and spin glass transition.
\newblock {\em Journal of Statistical Physics}, 124(6):1317–1350, 2006.

\bibitem[Mon19]{montanari2019optimization}
A.~Montanari.
\newblock {Optimization of the Sherrington-Kirkpatrick Hamiltonian}.
\newblock In {\em Proceedings of 60th Annual Symposium on Foundations of
  Computer Science (FOCS)}, pages 1417--1433, 2019.

\bibitem[MPV87]{mezard1987spin}
Marc M{\'e}zard, Giorgio Parisi, and Miguel~Angel Virasoro.
\newblock {\em Spin glass theory and beyond: An Introduction to the Replica
  Method and Its Applications}, volume~9.
\newblock World Scientific Publishing Company, 1987.

\bibitem[MRT03]{Montanari2003}
A.~Montanari and F.~Ricci-Tersenghi.
\newblock On the nature of the low-temperature phase in discontinuous
  mean-field spin glasses.
\newblock {\em The European Physical Journal B - Condensed Matter},
  33(3):339--346, 2003.

\bibitem[Pan14]{Panchenko2014parisi}
Dmitry Panchenko.
\newblock The parisi formula for mixed $p$-spin models.
\newblock {\em The Annals of Probability}, 42(3), 2014.

\bibitem[Par80]{parisi1980sequence}
Giorgio Parisi.
\newblock {A sequence of approximated solutions to the SK model for spin
  glasses}.
\newblock {\em Journal of Physics A: Mathematical and General}, 13(4):L115,
  1980.

\bibitem[Pra62]{prange1962use}
Eugene Prange.
\newblock The use of information sets in decoding cyclic codes.
\newblock {\em IRE Transactions on Information Theory}, 8(5):5--9, 1962.

\bibitem[PS16]{pittel2016satisfiability}
Boris Pittel and Gregory~B. Sorkin.
\newblock The satisfiability threshold for {k-XORSAT}.
\newblock {\em Combinatorics, Probability and Computing}, 25(2):236--268, 2016.

\bibitem[Reg04]{R04}
Oded Regev.
\newblock Quantum computation and lattice problems.
\newblock {\em SIAM Journal on Computing}, 33(3):738--760, 2004.

\bibitem[Reg09]{R09}
Oded Regev.
\newblock On lattices, learning with errors, random linear codes, and
  cryptography.
\newblock {\em Journal of the ACM (JACM)}, 56(6):1--40, 2009.

\bibitem[RV17]{rahman2017local}
Mustazee Rahman and Bálint Virág.
\newblock Local algorithms for independent sets are half-optimal.
\newblock {\em The Annals of Probability}, 45(3), 2017.

\bibitem[Sen18]{sen2018optimization}
Subhabrata Sen.
\newblock Optimization on sparse random hypergraphs and spin glasses.
\newblock {\em Random Structures \& Algorithms}, 53(3):504--536, 2018.

\bibitem[SLS{\etalchar{+}}25]{schmidhuber2025hamiltonian}
Alexander Schmidhuber, Jonathan~Z. Lu, Noah Shutty, Stephen Jordan, Alexander
  Poremba, and Yihui Quek.
\newblock Hamiltonian decoded quantum interferometry.
\newblock {\em arXiv preprint arXiv:2510.07913}, 2025.

\bibitem[SMR{\etalchar{+}}26]{SMR26}
Noah Shutty, Avijit Mandal, Seyoon Ragavan, Quentin Buzet, André Chailloux,
  Nicholas~C. Rubin, Abid Khan, Sami Boulebnane, Ruslan Shaydulin, John
  Azariah, and Stephen~P. Jordan.
\newblock Optimization using locally-quantum decoders.
\newblock {\em arXiv preprint arXiv:2604.24633}, 2026.

\bibitem[SZ81]{Sompolinsky1981}
H.~Sompolinsky and Annette Zippelius.
\newblock Dynamic theory of the spin-glass phase.
\newblock {\em Phys. Rev. Lett.}, 47:359--362, Aug 1981.

\bibitem[SZB21]{song2021cryptographic}
Min~Jae Song, Ilias Zadik, and Joan Bruna.
\newblock On the cryptographic hardness of learning single periodic neurons.
\newblock {\em Advances in neural information processing systems},
  34:29602--29615, 2021.

\bibitem[Tal98]{talagrand1998sherrington}
Michel Talagrand.
\newblock The {Sherrington--Kirkpatrick} model: A challenge for mathematicians.
\newblock {\em Probability Theory and Related Fields}, 110:109--176, 1998.

\bibitem[Tal00]{talagrand2000rigorous}
Michel Talagrand.
\newblock Rigorous low-temperature results for the mean field {$p$}-spins
  interaction model.
\newblock {\em Probability Theory and Related Fields}, 117:303--360, 2000.

\bibitem[Tal06]{talagrand2006parisi}
Michel Talagrand.
\newblock {The Parisi formula}.
\newblock {\em Annals of mathematics}, pages 221--263, 2006.

\bibitem[Tal11]{talagrand2011mean1}
Michel Talagrand.
\newblock {\em {Mean Field Models for Spin Glasses. Volume I: Basic Examples}},
  volume~54.
\newblock Springer Science \& Business Media, 2011.

\bibitem[Wei22]{Wein2022}
Alexander~S. Wein.
\newblock Optimal low-degree hardness of maximum independent set.
\newblock {\em Mathematical Statistics and Learning}, 4(3):221–251, 2022.

\bibitem[Wei25]{wein2025computational}
Alexander~S Wein.
\newblock Computational complexity of statistics: New insights from low-degree
  polynomials.
\newblock {\em arXiv preprint arXiv:2506.10748}, 2025.

\bibitem[ZKA25]{zlokapa2025average}
Alexander Zlokapa, Bobak~T. Kiani, and Eric~R. Anschuetz.
\newblock Average-case quantum complexity from glassiness.
\newblock {\em arXiv preprint arXiv:2510.08497}, 2025.

\bibitem[ZSWB22]{zadik2022lattice}
Ilias Zadik, Min~Jae Song, Alexander~S. Wein, and Joan Bruna.
\newblock Lattice-based methods surpass sum-of-squares in clustering.
\newblock In {\em Conference on Learning Theory}, pages 1247--1248. PMLR, 2022.

\end{thebibliography}
\end{document}